\documentclass{tlp}
\usepackage{mycommands}
\usepackage{subfig}
\usepackage{graphicx}
\usepackage{color}
\usepackage{latexsym}
\usepackage{amssymb}
\usepackage{amsmath}
\usepackage{proof} 
\usepackage{myproof}
\usepackage{listings}
\usepackage{lineno}

\usepackage{url}
\usepackage{galois}
\usepackage{verbatim}

\usepackage{tikz}
\usetikzlibrary{arrows}

\begin{document}
\long\def\comment#1{}


\title[Abstract Interpretation of Temporal CCP]{Abstract Interpretation of Temporal Concurrent Constraint Programs \footnote{This paper has been accepted for publication in Theory and Practice of Logic Programming (TPLP), Cambridge University Press. }} 

\author[M. Falaschi, C. Olarte and C. Palamidessi]{MORENO FALASCHI\\
Dipartimento di Ingegneria dell'Informazione e Scienze Matematiche\\
 Universit\`a di Siena, Italy \\
E-mail: moreno.falaschi@unisi.it
\and 
CARLOS OLARTE \\
Departamento de Electr\'onica y Ciencias de la Computaci\'on\\
Pontificia Universidad Javeriana-Cali, Colombia\\
E-mail: carlosolarte@javerianacali.edu.co
\and
CATUSCIA PALAMIDESSI\\
INRIA and LIX\\
 Ecole Polytechnique, France\\
E-mail: catuscia@lix.polytechnique.fr
}

\submitted{19 June 2012}
\revised{15 May 2013}
\accepted{3 December 2013}


\maketitle

\begin{abstract}

Timed Concurrent Constraint Programming (\tccp) is a declarative
model for concurrency offering a logic for specifying reactive
systems, i.e. systems that  continuously interact with the
environment. The universal \tccp\ formalism (\utcc)  is an extension
of \tccp\ with the ability to express mobility. Here mobility is
understood as communication of private names as typically done for
mobile  systems and security protocols. 
In this paper we consider the denotational semantics for \tccp, and
we extend it to a ``collecting"  semantics for \utcc\ based on closure
operators over sequences of constraints. Relying on this semantics,
we formalize a general framework for data flow analyses of \tccp\ and
\utcc\ programs by abstract  interpretation techniques. The concrete
and abstract semantics we propose are  compositional, thus  allowing
us to reduce the complexity of data  flow analyses. We show that our
method is sound and parametric with respect to the abstract domain. Thus,
different analyses can be performed by instantiating the framework.
We illustrate how it is possible to reuse abstract domains previously
defined for logic programming to perform, for instance,  a groundness
analysis for \tccp\ programs. We show the applicability of this
analysis in the context of  reactive systems. Furthermore, we  make
also use of the abstract semantics to exhibit a secrecy flaw in a
security protocol. We also show how it is possible
to make an analysis which may show that
\tccp\ programs are suspension free. This can be useful for several
purposes, such as for optimizing compilation or for
debugging. 
\end{abstract}
\begin{keywords}
Timed Concurrent Constraint Programming,
Process Calculi, Abstract Interpretation, Denotational Semantics, Reactive Systems\end{keywords}

\section{Introduction}
Concurrent Constraint Programming
(\ccp)  \cite{SRP91,cp-book} has emerged as a simple but
powerful paradigm for concurrency tied to logic that  extends and subsumes both
concurrent logic programming \cite{shapiro90} and constraint logic programming
\cite{DBLP:conf/popl/JaffarL87}. The \ccp\ model combines the
traditional operational view of process calculi with a
\emph{declarative} one based upon logic. This combination allows \ccp\  to benefit
from the  large body  of reasoning techniques  of  both process calculi
and  logic.  In fact, \ccp-based calculi have successfully been used in the modeling and
verification of several concurrent scenarios such as biological, security,
timed, reactive and stochastic  systems 
\cite{SRP91,Olarte:08:SAC,NPV02,tcc-lics94,JagadeesanM05} (see a survey in \cite{DBLP:journals/constraints/OlarteRV13}). 

In the \ccp\ model,
agents interact by \emph{telling} and \emph{asking} pieces of
information (\emph{constraints}) on a shared store of partial information.
The type of constraints that  agents can tell and ask is  parametric in an underlying constraint system. This makes \ccp\ a flexible model able to adapt to different application domains.

The \ccp\ model has been extended to consider the execution of processes along  time intervals or time-units. In  \texttt{tccp} \cite{bgm99}, 
 the notion of time is identified with the time needed to ask and tell information to the  store. In this model, the information in the store is carried through the time-units. On the other hand, in Timed \ccp\ (\tccp) \cite{tcc-lics94},
 stores are not automatically transferred between time-units. This way,  computations during a time-unit proceed 
monotonically but outputs of two different time-units are 
not supposed to be related to each other.
 More precisely, computations in \tccp\ take place in bursts of activity at a rate controlled by the environment. In this model,  the environment provides a stimulus (input) in the form of  a constraint. Then the system, after a finite 
number of internal reductions,  outputs the final store (a constraint) and  waits for the next interaction 
with the environment.  This view of \emph{reactive computation} is 
akin to synchronous languages such as Esterel \cite{BeGo92}  where the system 
reacts continuously 
with the environment at a rate controlled by the environment.  
Hence, these languages allow to program safety critical applications as control systems, 
for which it is fundamental to provide tools aiming at helping  to 
develop correct, secure, and efficient programs.

Universal \tccp\  \cite{Olarte:08:SAC}  (\utcc), adds to \tccp\ the expressiveness needed for \emph{mobility}. Here we understand mobility as the ability to communicate private names (or variables) much like in the $\pi$-calculus \cite{milner.parrow.ea:calculus-mobile}. Roughly, a \tccp\ \emph{ask} process $\whenp{c}{P}$ executes the process $P$ only if the constraint  $c$ can be entailed from the store. This idea is generalized in \utcc\ by  a parametric ask that executes $P[\vec{t}/\vec{x}]$ when the constraint $c[\vec{t}/\vec{x}]$ is entailed from the store. Hence the variables in $\vec{x}$ act as formal parameters of the ask operator. This simple change allowed to widen  the spectrum of application of \ccp-based languages to scenarios such as verification of security protocols \cite{Olarte:08:SAC} and service oriented computing \cite{Lopez-Places09}.


Several domains and frameworks (e.g., \cite{CC92,armstrong98two,Codish99} )  
have been proposed  for the analysis of logic programs. The particular characteristics  of  timed \ccp\  programs  pose additional difficulties 
for the development of such tools in this language. Namely, 
 the concurrent, timed nature of the language, and the synchronization mechanisms based on  entailment of constraints (blocking asks). 
 Aiming at statically analyzing \utcc\ as well as \tccp\  programs, we have to  consider the additional technical issues due to  the infinite internal computations generated by parametric asks as we shall explain later.

We develop here a \emph{compositional} semantics for \tccp\ and \utcc\ that 
allows us to describe the behavior of programs and collects all concrete 
information needed to properly abstract the properties of interest. This 
semantics is  based on closure operators over sequences of constraints  
along the lines of \cite{tcc-lics94}.  
We show that parametric asks in \utcc\ of the form  $\absp{\vx}{c}{P}$  can 
be neatly characterized as  closure operators. This characterization is 
shown to be somehow dual to the semantics for the local operator 
$\localp{\vx}{P}$ that restricts the variables in $\vx$ to be local to $P$. 
We  prove  the
 semantics  to be  fully 
 abstract w.r.t.  the operational semantics for a significant fragment of the  
 calculus.

We  also propose an abstract   semantics which approximates the  concrete one.  Our  framework is formalized by  abstract interpretation techniques and is  parametric w.r.t. 
 the abstract domain. It  allows us to exploit  the work done for developing abstract domains for 
 logic programs. Moreover, we can make new analyses for reactive 
 and mobile systems, thus widening the reasoning techniques   available for   \tccp\ and \utcc, such as 
type systems \cite{Lopez09}, logical characterizations \cite{MendlerPSS95,NPV02,Olarte:08:SAC} and semantics \cite{tcc-lics94,BoerPP95,NPV02}.

The abstraction we propose proceeds in two-levels. First, we approximate the  
constraint system leading to an abstract constraint system. We give the sufficient 
conditions which 
have to be satisfied for ensuring the soundness of the abstraction. Next, to obtain efficient analyses, we abstract the infinite sequences of (abstract) constraints obtained from the previous step. Our semantics is then computable and compositional. Thus, it allows us to  master  the complexity of the data-flow analyses. Moreover, 
the  abstraction  \emph{over-approximates} the concrete semantics, thus preserving safety properties.

To the best of our knowledge,  this is the first attempt to propose a compositional semantics and an abstract interpretation framework  for a language adhering to the above-mentioned characteristics of  \utcc.  Hence we can develop 
analyses for several applications of \utcc\ or its sub-calculus 
\tccp\  (see e.g., \cite{DBLP:journals/constraints/OlarteRV13}). In particular, we instantiate our framework in 
three different 
scenarios.  The first one presents 
an abstraction of  a cryptographic constraint system. We  use 
the abstract semantics to bound the number of messages that a spy may generate, in order  to exhibit  a secrecy flaw in a security protocol written 
in \utcc. The second one 
tailors an abstract domain for groundness and type dependency analysis in 
logic programming to perform a groundness analysis of a  \tccp\ program. This analysis is proven useful 
to derive a property of a control system specified
in \tccp.  Finally, we present an analysis that may show that a \tccp\ program is suspension free. This  analysis can be used later for optimizing compilation or for debugging purposes. 



The ideas of this paper stem mainly from the works of the authors in \cite{BoerPP95,Falaschi:97:TCS,DBLP:journals/iandc/FalaschiGMP97,NPV02,Olarte:08:PPDP}  to give semantic characterization of \ccp\ calculi and from the works in  
\cite{FalaschiGMP93,CFM94,Falaschi:97:TCS,ZaffanellaGL97,FalaschiOPV07} to provide abstract interpretation frameworks to analyze concurrent logic-based languages. A preliminary short version of this paper without proofs was 
published in \cite{Falaschi:PPDP:09}. In this paper we give many more examples and explanations. We also refine several  technical details and present full proofs. Furthermore, we develop a new application 
for analyzing suspension-free \tccp\ programs. 

The rest of the paper is organized as follows. Section \ref{sec:cc} recalls the notion of constraint 
system and the operational semantics of \tccp\ and  \utcc. 
In Section \ref{sec:denotsem}  we  develop  the denotational semantics based on 
sequences of constraints. Next, in 
Section \ref{sec:absframework}, we study the abstract interpretation framework 
for \tccp\ and \utcc\ programs.  The three instances and 
the applications of the framework 
are presented in Section \ref{sec:app}. Section \ref{sec:concluding} concludes. 


\section{Preliminaries} \label{sec:cc}
Process calculi based on the \ccp\ paradigm are parametric in a
\emph{constraint system} specifying the basic constraints  agents can tell and ask.  These constraints  represent
a piece of (partial) information  upon  which processes may act. The
constraint system hence provides a signature from which  constraints
can be built. Furthermore, the constraint system provides an
\emph{entailment} relation ($\entails$) specifying inter-dependencies
between constraints. Intuitively,  $c\entails d$  means that the
information $d$ can be deduced from the information represented by
$c$. For example, $x > 60 \entails x >42$. 

Here we consider an abstract definition of constraint systems as
cylindric algebras as in \cite{BoerPP95}. The notion of constraint system as first-order formulas  \cite{DBLP:conf/ccl/Smolka94,NPV02,Olarte:08:SAC}  can be
seen as an instance of this definition. All results of this paper
still hold, of course, when more concrete systems are considered. 

\begin{definition}[Constraint System] \label{def:cs}
A cylindric constraint system  is a structure 
 $
{\bf C} = \langle \cC,\leq,\sqcup,\true,\false,{\it Var}, \exists,
D \rangle
$ s.t. 
\\\noindent{-}
 $\langle \cC,\leq,\sqcup,\true,\false \rangle$ is a lattice
with $\sqcup$ the $\lub$ operation (representing the logical
\emph{and}), and $\true$, $\false$ the least and the greatest
elements in $\cC$ respectively (representing $\texttt{true}$ and
$\texttt{false}$). Elements in $\cC$ are called \emph{constraints}
with typical elements $c,c',d,d'...$. If $c\leq d$ and $d\leq c$ we write $c \equivC d$. If $c\leq d$ and $c\not\equivC d$, we write $c<d$. 
\\\noindent{-}${\it Var}$ is a denumerable set of variables and for each
$x\in {\it Var}$ the function $\exists x: \cC \to \cC$ is a
cylindrification operator satisfying:
		    (1) $\exists x (c) \leq c$. 
		(2) If $c\leq d$ then $\exists x (c) \leq \exists x (d)$.
		(3) $\exists x(c \sqcup \exists x (d)) \equivC \exists x(c) \sqcup \exists x(d)$.
		(4) $\exists x\exists y(c) \equivC \exists y\exists x (c)$.
		(5) For an increasing chain $c_1 < c_2 < c_3...$, $ \exists x
\bigsqcup_i c_i \equivC \bigsqcup_i \exists x ( c_i) $.
\\\noindent{-} For each $x,y \in {\it Var}$, the constraint $d_{xy} \in D$ is a
\emph{diagonal element} and it satisfies:
			(1) $d_{xx} \equivC \true$.
		(2) If $z$ is different from $x,y$ then $d_{xy} \equivC \exists z(d_{xz}
\sqcup d_{zy})$.
		(3) If $x$ is different from $y$ then $c \leq d_{xy} \sqcup
\exists x(c\sqcup d_{xy})$.
\end{definition}
The cylindrification operators model a sort of existential
quantification, helpful for hiding information. 
We shall use $\fv(c)=\{ x\in Var \ |  \ \exists x(c) \not\equivC c \}$ to denote the set of free variables that occur in $c$. If $x$ occurs in $c$  and $x\not\in\fv(c)$, we say that $x$ is bound in $c$. We use $\bv(c)$ to denote the set of bound variables in $c$.

 Properties (1) to (4) are standard. Property (5) is shown to be required in \cite{BoerPP95} to establish the semantic adequacy of  \ccp\ languages when infinite computations are considered. Here,   the   continuity of the semantic operator 
 in Section \ref{sec:denotsem} relies on the continuity of $\exists$ (see Proposition \ref{prop:cont-td}). Below we give some examples on the requirements to satisfy this property in the context of different constraint systems. 
 

 The diagonal element $d_{xy}$ can be thought of as the equality $x=y$. Properties (1) to (3) are standard and  they allow us to define  substitutions of the form $[t/x]$  required, for instance, to represent the substitution of formal and actual parameters in procedure call. We shall give a formal definition of them in Notation \ref{not:terms}. 

  Let us give some examples of constraint systems.  The finite domain constraint system  (FD) \cite{HentenryckSD98} assumes variables to range over finite domains and, in addition to equality, one may have predicates  that restrict the possible values of a variable to some finite set, for instance $x<42$. 

The Herbrand constraint system  $\cH$ consists of a first-order language with equality. 
The entailment relation is the one we expect from equality, for instance, $f(x,y) = f(g(a),z)$ must entail $x=g(a)$ and $y=z$. 
 $\cH$ may contain non-compact elements to represent the limit of infinite chains. To see this, let $s$ be the successor constructor,  $ \exists y (x = s(s^n(y)))$ be denoted as the constraint  $\gtc{x}{n}$ (i.e., $x>n$) and  $\{ \gtc{x}{n}\}_n$ be the ascending chain  $\gtc{x}{0} < \gtc{x}{1} < \cdots$. We note that $\exists x (\gtc{x}{n})  =  \true$ for any $n$ and then,  $\bigsqcup  \{ \exists x (\gtc{x}{n}) \}_n = \true$. 
Property (5) in Definition \ref{def:cs} dictates that 
$\exists x \bigsqcup \{\gtc{x}{n}\}_n$ must be equal to $\true$ (i.e., there exists an $x$ which is greater than any $n$). For that,  we need a constraint, e.g.,  $\infc{x}$ (a non-compact element), to be the limit $\bigsqcup \{\gtc{x}{n}\}_n$. We know that    $\infc{x} \entails \gtc{x}{n}$ for any $n$
and then,  $\bigsqcup \{\gtc{x}{n}\}_n=\infc{x}$ and $\exists x( \infc{x}) = \true$ as wanted.
A similar phenomenon arises  in the definition of constraint system   as Scott information systems in \cite{SRP91}. There,  constraints are represented as finite subsets of \emph{tokens} (elementary constraints) built from a given set $D$. The entailment  is similar to that  in Definition \ref{def:cs} but   restricted to compact elements, i.e.,  a constraint can be entailed only from a finite set of elementary constraints. Moreover, $\exists$   is extended to  be a continuous  function,  thus satisfying  Property (5) in Definition \ref{def:cs}. Hence, the  Herbrand constraint system in  \cite{SRP91}  considers also a non-compact element (different from $\false$)  to be the   limit of the chain $\{\gtc{x}{n}\}_n$.

Now consider the Kahn constraint system underlying data-flow languages where equality is assumed along with the constant $\nilp$ (the empty list), the predicate $\nempty{x}$ ($x$ is not  $\nilp$),  and the functions $\first{x}$ (the first element of $x$), $\tail{x}$ ($x$ without its first element) and $\append{x}{y}$ (the concatenation of $x$ and $y$). If we consider the Kahn constraint system in \cite{SRP91},  the constraint $c$ defined as $\{\first{\tailn{n}{x}} = \first{\tailn{n}{y}} \mid n\geq 0\}$ does not entail   $\{x=y\}$   since the entailment relation is defined only on compact elements. In Definition \ref{def:cs}, we are free to decide if $c$ is different or not from $x=y$. If we equate them, the constraint $x=y$ is not longer a compact element and then, one has to be careful to only use a compact version of ``$=$'' in programs (see Definition \ref{tcc:syntax}). A similar situation occurs with the  Rational Interval Constraint System \cite{SRP91} and the constraints  $ \{ x \in [0, 1 + 1/n]  \mid n \geq 0 \}$  and  $x \in [0,1]$. 

All in all many different constraint systems satisfy Definition \ref{def:cs}. Nevertheless, one has to be careful since the constraint systems might not be the same as what is naively expected due to the presence of non-compact elements.

We conclude this section by setting some notation and conventions
about terms, sequences of constraints, substitutions and diagonal
elements. We first lift the relation $\leq$ and the cylindrification
operator to sequences of constraints.

\begin{notation}[Sequences of Constraints]\label{not:seq-constraints}
We denote  by  $\cC^\omega$ (resp. $\cC^*)$ the set of infinite (resp. finite) sequences of constraints with typical elements $w,w',s,s',...$. We
use $W,W',S,S'$ to range over subsets of $\cC^\omega$ or $\cC^*$. 
We use $c^\omega$ to denote the sequence $c.c.c...$.  The length of
$s$ is denoted by $|s|$ and the  empty sequence by $\epsilon$. The
$i$-th element in $s$ is denoted by $s(i)$. 
 We write $s\leq s'$ iff $|s| \leq |s'|$ and for all $i \in \{1,
\ldots, |s|\}$, $s'(i) \entails s(i)$. If  $|s| = |s'|$ and for all
$i \in \{1,..., |s|\}$ it holds $s(i) \equivC s'(i)$, we shall write
$s \equivC s'$. Given a sequence of variables $\vx$, with  $\exists \vec{x}(c)$ we mean  $\exists x_1\exists x_2
...\exists x_n (c) $ and 
 with $\exists \vx(s)$ we mean the pointwise application of the
cylindrification operator to the constraints in $s$. 
 \end{notation}
 
We shall assume that 
the diagonal element $d_{xy}$
is  interpreted as the equality $x=y$. Furthermore, 
following \cite{DBLP:journals/jlp/GiacobazziDL95}, we extend the use of $d_{xy}$ to consider terms as in $d_{xt}$. More precisely,
%

\begin{convention}[Diagonal elements]\label{conv:diag}
We assume that the constraint system under consideration contains an
equality theory. Then, diagonal elements $d_{xy}$ can be 
thought of as formulas of the form $x=y$.  
We shall use indistinguishably both notations.  Given a variable $x$ and
a term $t$ (i.e., a variable, constant or $n$-place function of $n$ terms symbol), we shall use $d_{xt}$ to denote the equality $x=t$.
Similarly, given a sequence of distinct variables $\vec{x}$ and a sequence of terms  $\vec{t}$, if $|\vec{x}|=|\vec{t}|=n$ then   $\dxt $  
denotes the constraint $\bigsqcup\limits_{1\leq i \leq n} x_i = t_i$.
If $|\vec{x}|=|\vec{t}|=0$ then $\dxt  = \true$. Given a set of diagonal elements  $E$, we shall write $E\DEQ \dxt$ whenever $d_i \entails  \dxt$ for some $d_i \in E$. Otherwise, we write $E \DNEQ \dxt$.
\end{convention}
 
Finally,  we set the notation for substitutions.

\begin{notation}[Admissible substitutions] \label{not:terms}
Let $\vx$ be a sequence of pairwise distinct variables and $\vt$ be a sequence of terms s.t. $|\vec{t}|=|\vec{x}|$. 
We denote by  $c\sxt $ the constraint $\exists \vx(c \sqcup \dxt )$
which represents abstractly the constraint obtained from $c$ by
replacing the variables $\vec{x}$ by $\vec{t}$.
We say that $\vec{t}$ is admissible for $\vec{x}$, notation
$\adm{\vec{x}}{\vec{t}}$, if the variables in $\vt$ are different from those in $\vx$. 
If $|\vec{x}| =|\vec{t}| =0$ then trivially
$\adm{\vec{x}}{\vec{t}}$.
Similarly, we say that the substitution $\sxt $ is admissible iff
$\adm{\vec{x}}{\vec{t}}$. 
Given an admissible
substitution $\sxt $, from   Property (3) of diagonal elements in Definition \ref{def:cs},  we note that   $c\sxt  \sqcup \dxt  \entails c$. 
\end{notation}

\subsection{Reactive Systems and  Timed CCP}
\label{sec:tcc-definition}
Reactive systems \cite{BeGo92} are those that react continuously with
their environment at a  rate controlled by the environment. 
For example, a controller or a  signal-processing system, receives a
stimulus (input) from the environment. It computes an output  and
then, waits for the next interaction with the environment. 

In the \ccp\ model,  the shared store of constraints grows
monotonically, i.e., agents cannot drop information (constraints)
from it. Then, a system that changes the state of a  variable as in  ``${\it signal=on}$" and ``${\it signal=off}" $ leads to an inconsistent store. 

Timed \ccp\   (\tccp) \cite{tcc-lics94} extends \ccp\  for
reactive  systems. Time is conceptually divided into \emph{time
intervals }(or \emph{time-units}). In a particular time
interval, a \ccp\ process $P$  gets an input  $c$
from the environment, it executes with this input as the initial
\emph{store}, and when it reaches
its resting point, it \emph{outputs} the resulting store $d$ to the
environment. The resting point determines also a residual process $Q$
which is then executed in the next time-unit. The resulting store $d$
is not automatically transferred to the next time-unit.
This way, computations during a time-unit proceed monotonically but
outputs of two different time-units are not supposed to be related to
each other. Therefore, the variable ${\it signal}$ in the example
above may change its value when passing from one 
time-unit to the next one.  
\begin{definition}[\tccp\  Processes]\label{tcc:syntax} The set
$Proc$ of \tccp\ processes is built from the syntax
\[
\begin{array}{lll}
P,Q &:=&  \skipp  \ \ |\  \  \tellp{c} \ \ |\  \whenp{c}{P} \ \ |\  \
P\parallel
Q \ \ |\  \  \localp{\vx }{P} \ \ |\  \\ 
  & &  \nextp{P}  \ \ |\  \  \unlessp{c}{P}     
  \ \ | \ \ p(\vt)
   \end{array}
\]
where $c$ is a compact element of the underlying constraint system.
Let $\cD$ be a set of process declarations of the form
$p(\vx)\defsymbol P$. A \tccp\ program takes the form $\cD.P$.  We assume $\cD$ to have a unique process definition for every process name, and recursive calls to be guarded by a ${\bf next}$ process. 
\end{definition}

The process $\skipp$ does nothing thus representing inaction. The
process  $\tellp{c}$ adds $c$ to the store in the current time
interval making it available to the other processes. 
The process $\whenp{c}{P}$ \emph{asks} if $c$ can be deduced from the
store. If so, it behaves as $P$. In other case, it remains blocked
until the store contains at least as much information as $c$.
 The parallel composition of $P$ and $Q$  is denoted by
$\parp{P}{Q}$. Given a  set of indexes $I=\{1,...,n\}$, we shall use
$\prod\limits_{i\in I} P_i$ to denote the parallel composition $P_{1}
\parallel ... \parallel P_{n}$. 
The process  \(\localp{\vx}{P} \)  \emph{binds}  $\vec{x}$ in
$P$ by declaring it private to $P$.  It behaves like $P$, except that
all the information on  the variables $\vec{x}$
produced by $P$ can only be seen by $P$  and the information on the
global variables in $\vec{x}$  produced by other processes cannot be
seen by $P$. 

The process  \( \nextp{P} \) is a \emph{unit-delay} that executes $P$
in the next time-unit.  The \emph{time-out} \( \unlessp{c}{P} \) is
also a unit-delay, but  \( P \)  is executed in the next time-unit if
and only if  \( c \) is not entailed by the final store at the
current time interval. We use \( \mathbf{next}^{n}P \) as a shorthand
for
 \( \mathbf{next} \dots \mathbf{next}\, P \), 
with  \(\mathbf{next}\) repeated $n$ times.  

We extend the definition of free variables to processes as follows:   $\fv(\skipp)=\emptyset$; $\fv(\tellp{c}) = \fv(c)$; $\fv(\whenp{c}{Q}) =  \fv(c) \cup \fv(Q)$; $\fv(\unlessp{c}{Q})= \fv(c) \cup \fv(Q)$; $\fv(Q \parallel Q')=\fv(Q) \cup \fv(Q')$; $\fv(\localp{\vx}{Q})=\fv(Q)\setminus \vx$; $\fv(\nextp{Q})=\fv(Q)$;  $\fv(p(\vt)) = vars(\vt)$ where $vars(\vt)$ is the set of variables occurring in $\vt$. 
A variable $x$ is bound in $P$ if $x$ occurs in $P$ and $x\notin \fv(P)$. We use $\bv(P)$ to denote the set of bound variables in $P$.

Assume a (recursive) process definition 
$
\ \ p(\vx) \defsymbol P \ \ 
$ where  $ \fv(P) \subseteq \vx $. The call $p(\vt)$ reduces to $P\sxt$.
Recursive calls in $P$ are assumed to be guarded by a 
$\nextp$ process to avoid non-terminating sequences of recursive
calls during a time-unit  (see \cite{tcc-lics94,NPV02}). 




In the forthcoming sections  we shall use the idiom $\bangp{P}$ defined as follows:
\begin{notation}[Replication]\label{rem:bang}
The replication of $P$, denoted as 
$\bangp{P}$,  is a short hand for a call to a process definition  $\texttt{bang}_P() \defsymbol P \parallel \nextp{\texttt{bang}_P()}$. 
Hence, $\bangp{P}$ means    $P \parallel \nextp{P} \parallel \nextp{^2 P} ...$. 
\end{notation}

%
%


\subsection{Mobile behavior and \utcc}\label{sec:mobility}
As we have shown, interaction of \tccp\ processes is asynchronous as
communication takes place through the shared store of partial
information. Similar to other formalisms, by defining local (or
private) variables, \tccp\ processes specify boundaries in the
interface they offer to interact with each other. Once these
interfaces are established, there are few mechanisms to modify them.
This is not the case e.g., in the $\pi$-calculus \cite{milner.parrow.ea:calculus-mobile} where
processes can change their communication patterns by exchanging their
private names.  The following example illustrates the limitation of  $ask$ processes to communicate values and local variables.

\begin{example}
Let  $\outp(\cdot)$ be a constraint and let  $P=\whenp{\outp(x)}{R}$
be 
 a system that must react when receiving a stimulus (i.e., an input) of the form
$\outp(n)$ for $n>0$.   We notice that   $P$ in a store $\outp(42)$ 
does not execute $R$ since   $\outp(42)\notentails \outp(x)$. 
\end{example}

The key point in the previous example  is that $x$ is a free-variable and
hence,  it  does not act as a formal parameter (or place holder) for every
term $t$  such that $\outp(t)$ is entailed by the store.

In  \cite{Olarte:08:SAC}, \tccp\ is extended for \emph{mobile
reactive} systems leading to  \emph{universal timed} \ccp\   (\utcc).
To model mobile behavior,  \utcc\  replaces the  ask operation
$\whenp{c}{P}$ with a   parametric ask construction,
namely $\absp{\vec{x}}{c}{P}$.  This process can be viewed as a
$\lambda$-\emph{abstraction} of the process $P$ on the variables
$\vec{x}$ under the constraint (or with the \emph{guard}) $c$.
Intuitively, for all admissible substitution $\sxt$ s.t. the current store entails $c\sxt$, the process  $\absp{\vec{x}}{c}{P} $ performs $P\sxt $. For example, $\absp{x}{\outp(x)}{R}$   in a store entailing both $\outp(z)$ and $\outp(42)$  executes $R[42/x]$ and $R[z/x]$. 
 
\begin{definition}[\utcc\  Processes and Programs]\label{utcc:syntax}
 The \utcc\ processes and programs result from replacing in  
Definition \ref{tcc:syntax} the expression $\whenp{c}{P}$ with 
$\absp{\vec{x}}{c}{P}$ where the variables in $\vec{x}$ are pairwise distinct. 
\end{definition}

  When
$|\vec{x}|=0$ we  write $\whenp{c}{P}$ instead of
$\absp{\epsilon}{c}{P}$.  
Furthermore, the process $\absp{\vec{x}}{c}{P} $ binds $\vec{x}$ in $P$ and $c$. We thus extend accordingly the sets $\fv(\cdot)$ and $\bv(\cdot)$ of free and bound variables.  

 From a programming point of view,  we can see the variables
$\vec{x}$ in the abstraction $\absp{\vec{x}}{c}{P}$ as the formal
parameters of    $P$. In fact, 
the \utcc\ calculus was introduced in  \cite{Olarte:08:SAC}  
with replication ($\bangp{P}$) and 
without process definitions since replication and abstractions are enough to encode recursion. 
 Here we add process definitions to properly deal with  \tccp\ programs with
recursion which are more expressive than those without it (see
\cite{DBLP:conf/ppdp/NielsenPV02}) and we omit replication to avoid
redundancy in the set of operators (see Notation \ref{rem:bang}). 
We thus could have  dispensed with the next-guarded restriction in Definition \ref{tcc:syntax} for  \utcc\ programs. Nevertheless,  in order to give a unified presentation of the forthcoming results, we assume that  \utcc\ programs adhere also to that restriction.

We conclude  with an example of mobile behavior where  a process $P$ sends a local variable to  $Q$. Then,  both
processes can communicate through the shared variable. 

\begin{example}[Scope extrusion]\label{ex:mobility}
Assume two components  $P$ and $Q$ of a system such that $P$ creates
a local variable that must be shared with $Q$. This system can be modeled
as
\[
\begin{array}{lll l lll}
P & = & \localp{x}{(\tellp{\outp(x)} \parallel P')}  & \quad & 
Q & = &  \absp{z}{\outp(z)}{Q'}
\end{array}
\]
We shall show later that the parallel composition of
$P$ and $Q$ evolves to a process of the form
$
P' \parallel Q'[x/z]
$ where $P'$ and $Q'$ share the local variable $x$ created by $P$.
Then, any information produced by $P'$ on $x$ can be seen by $Q'$ and
vice versa. 
\end{example}

\subsection{Operational Semantics (SOS)}\label{sec:opersem}
We take inspiration on the structural operational semantics (SOS)  for {\em linear} \ccp\ in \cite{DBLP:journals/iandc/FagesRS01,DBLP:conf/fsttcs/HaemmerleFS07} to 
define the behavior of processes.  We consider \emph{transitions}  between \emph{configurations} of the form  \( \mconf{\vx;P ; c}  \) 
where $c$ is a constraint representing the current store, $P$ a process and $\vx$ is a set of distinct variables representing the  bound
(local) variables of $c$ and $P$.
  We shall use $\gamma,\gamma',\ldots$ to range over configurations.  Processes 
 are  quotiented by $\equiv$ defined as follows. 
\begin{definition}[Structural Congruence]\label{struct} Let \( \equiv  \)
be the smallest congruence  satisfying:  
(1)  \( P \equiv Q \) if they differ only by a renaming of bound variables (alpha-conversion); (2)  \( P \parallel \skipp \equiv P \); (3)  \( P \parallel Q \equiv Q \parallel P \); and  (4)  \( P \parallel (Q \parallel R) \equiv (P \parallel Q) \parallel R \).
\end{definition}
 The congruence relation $\equiv$ is extended to configurations by decreeing that $\mconf{\vx;P;c} \equiv \mconf{\vy;Q;d}$ iff 
  $\localp{\vx}{P} \equiv \localp{\vy}{Q}$ and $\exists \vx(c) \equivC \exists \vy(d)$. 
\begin{figure}
\resizebox{\textwidth}{!}{
$
\begin{array}{ccc}
 \rightinfer[\rTell]{\mconf{\vx;
\tellp{c};d} \redi_{{}} \mconf{\vx;
\skipp;d \sqcup c}}{ } & \quad &
\rightinfer[\rPar]{\mconf{\vx; P\parallel
Q;c} \redi_{{}}  \mconf{\vx \cup \vy;P'\parallel
Q;d}}{ \mconf{\vx;P;c} \redi_{{}}
\mconf{\vx\cup \vy; P';d} \mbox{ , } \vy\cap \fv(Q)=\emptyset}\\\\
\multicolumn{3}{c}{
\rightinfer[\rLocal]{\mconf{\vx; \localp{\vy}{P};d } \redi_{{}} \mconf{\vx \cup \vy; P;d}}
{ \vy \cap \vx = \emptyset,  \vy\cap \fv(d) = \emptyset }}\\\\
\multicolumn{3}{c}{
\rightinfer[\rAbs]
{
\mconf{\vx;\absp{\vy}{c;E}{P} ; d}
\redi  \mconf{\vx; P[\vt/\vy] \parallel \absp{\vy}{c; E\cup \{\dyt\}  }{P}; d}
} 
{   
d \entails c\syt, \adm{\vy}{\vt}, \ 
\mbox{and }  E  \DNEQ \dyt }}\\\\
\multicolumn{3}{c}{
\rightinfer[\rStructVar]{\mconf{\vy;P;c} \redi \mconf{\vy\cup\bigcup \vx_i;P;c_1\sqcup ... \sqcup c_n}}
{\nf(c)=\exists \vx_1 c_1 \sqcup \cdots \sqcup \exists \vx_n c_n   \ \ \ \ \   \vy \cap \vx_i = \emptyset\ \ \mbox{forall } i\in 1..n}
} 
\\\\

\multicolumn{3}{c}{
\rightinfer[\rStruct]{\mconf{\vx;P;c} \redi \mconf{\vx';P';c'}}
{\mconf{\vx;Q;c} \redi \mconf{\vy;Q';c''}} \qquad\mbox{if }
P \equiv Q \mbox{ and } \mconf{\vx';P';c'} \equiv \mconf{\vy;Q';c''} }
\\\\

\rightinfer[\rCall]{\mconf{\vx;
p(\vt);d} \redi \mconf{\vx;
P\sxt;d}}{p(\vx)\defsymbol P \in \mathcal{D} 
\ \ \ \ adm(\vx,\vt)} & \quad & 
\rightinfer[\rUnless]{\mconf{\vx;\unlessp{c}{P};d}\redi_{{}} 
\mconf{\vx; \skipp;d} }{d\entails c}\\\\
\multicolumn{3}{l}{\mbox{ Observable  Transition}}\\
\multicolumn{3}{c}{
 \rObserv \ \frac{\raisebox{.1cm}{$\mconf{\emptyset;
P;c} \redi_{{}}^{*}  \mconf{\vx; Q;d}
\not\redi$}}
{\raisebox{-.2cm}{$P\rede{(c,\exists \vx (d))} \localp{\vx}{F(Q)}$}} \ \mbox{\ \  where \ \ }

{F}(P)=\left\{ 
\begin{array}{ll} 
     F(\skipp) = F(\absp{\vx}{c;D}{Q}) = \skipp\\
     F(P_1 \parallel P_2) = F(P_1) \parallel F(P_2)\\
     F(\nextp{Q}) = F(\unlessp{c}{Q}) = Q
\end{array} \right. }
\end{array}
$
}
\caption{\label{opersem} 
SOS.  In $\rStruct$, $\equiv$ is given in 
Definition \ref{struct}. 
In $\rAbs$ and $\rCall$,  $\adm{\vx}{\vt}$ is defined in Notation \ref{not:terms}. In $\rAbs$, $E$ is assumed to be a set of diagonal elements and $\DNEQ$ is defined in Convention \ref{conv:diag}. In $\rStructVar$, $\nf(d)$ is defined in Notation \ref{not:nf}. 
}
\end{figure}


Transitions are given by the relations $\redi$ and $\Longrightarrow$  in Figure \ref{opersem}.   The \emph{internal} transition 
$\mconf{\vx;P;c} \redi \mconf{\vx';P';c'}$
  should be read as ``\( P \)  with store \( c \) reduces, in one internal step, to \( P' \)  with store \( c'\)\ ''. We shall use $\redi^*$ as the reflexive and transitive closure of $\redi$. If $\gamma \redi \gamma'$ and $\gamma' \equiv \gamma''$ we write $\gamma \redi\equiv \gamma''$. Similarly for $\redi^*$.

 The \emph{observable transition} {\small \( P\rede{(c,d)} R \)} should be read as ``\( P \)
on input \( c \), reduces
in one \emph{time-unit} to \( R \) and outputs \( d \)''. The observable 
transitions are obtained from finite sequences of internal ones. 

The rules in Figure \ref{opersem} 
are easily seen to realize the operational  intuitions given in Section \ref{sec:tcc-definition}. As clarified below, the seemingly missing rule  for a $\mathbf{next}$  process is given by $\rObserv$.  Before  explaining such rules, let us introduce the following notation needed for  $\rStructVar$.

\begin{notation}[Normal Form]\label{not:nf}
We observe that  the store $c$ in a configuration takes the form $\exists \vx_1(d_1) \sqcup ... \sqcup\exists \vx_n(d_n)$ where each $\vx_i$ may be an empty set of variables. The normal form of $c$, notation $\nf(c)$, is the constraint obtained by renaming the variables in $c$ such that for all $i,j \in 1..n$, if  $i\neq j$ then the variables in  $\vx_i$ do not occur neither bound nor free in $d_j$. It is easy to see that $c \equivC \nf(c)$. 
\end{notation}

\noindent{-} $\rTell$ says that the process $\tellp{c}$ adds $c$ to the current store $d$ (via the lub operator of the constraint system) and then evolves into $\skipp$.
\\\noindent{-}  $\rPar$ says that if $P$ may evolve into $P'$, this reduction also takes place when running in parallel with  $Q$. 
\\\noindent{-} The process $\localp{\vy}{Q} $ 
adds $\vy$ to the local variables of the configuration and then evolves into $Q$.  The side conditions of the rule $\rLocal$ guarantee that $Q$ runs with a different set of variables from those in the store and those used by other processes. 

\noindent{-} We extend the transition relation to consider processes of the form $\absp{\vy}{c;E}{Q}$ where $E$ is a set of diagonal elements. 
If $E$ is empty, we write $\absp{\vy}{c}{Q}$ instead of $\absp{\vy}{c;\emptyset}{Q}$. If $d$ entails $c\syt$, then   $P\syt$ is executed (Rule $\rAbs$).
Moreover, the abstraction persists in the current time interval to allow other potential replacements of $\vy$ in $P$. Notice that $E$ is augmented with  $\dyt$
and the side condition $E  \DNEQ \dyt$ prevents  executing   $P\syt$ again.   The process $P\syt$ is obtained by equating  $\vy$ and $\vt$  and then, hiding the information about  $\vy$, i.e., $
\localp{\vy}{(\bangp\tellp{d_{\vy\vt}} \parallel  P)}$. 

\noindent{-} Rule $\rStructVar$ allows us to \emph{open} the scope of  existentially quantified constraints in the store  (see Example \ref{ex:ex-mob-sos} below). If $\gamma$ reduces to $\gamma'$ using this rule then $\gamma\equiv \gamma'$. 
\\\noindent{-} Rule $\rStruct$ says that one can use the structural congruence on processes to continue a derivation (e.g., to do alpha conversion). It is worth noticing that we do not allow in this rule to transform the store via the relation $\equiv$ on configurations and then, via $\equivC$ on constraints. We shall discuss the reasons  behind  this choice  in  Example \ref{ex:ex-mob-sos}. 
\\\noindent{-}What we observe from $p(\vt)$ is $P\sxt$ where the formal 
parameters are substituted by the actual parameter (Rule $\rCall$). 
\\\noindent{-} Since the process $P=\unlessp{c}{Q}$ executes $Q$ in the next time-unit only if  the final store at the current time-unit  does not entail $c$, in the rule $\rUnless$ $P$ evolves into $\skipp$ if the current store $d$ entails $c$. 

For the observable transition relation, rule $\rObserv$ says that an observable transition from \( P \) labeled
with  $(c,\exists \vx(d))$ is obtained from a terminating sequence of internal transitions from $\mconf{\emptyset; P;c}$ to  $\mconf{\vx;Q;d}$. The process  to be executed in the next time interval is   \(\localp{\vx}{F(Q)}\) (the ``future'' of $Q$). $F(Q)$  is obtained by removing from \(Q\) the ${\bf abs}$ processes that could not be executed  and  by 
``unfolding'' the sub-terms within $\mathbf{next}$ and $\mathbf{unless}$ expressions. Notice that 
the output of a process hides the local variables ($\exists \vx (d)$) and those variables are also hidden in the next time-unit ($\localp{\vx}{F(Q)}$).

Now we are ready to show 
 that processes in Example \ref{ex:mobility} evolve into a configuration where a (local) variable can be communicated and shared. 

\begin{example}[Scope Extrusion and Structural Rules]\label{ex:ex-mob-sos}
Let $P$ and $Q$ be as in Example \ref{ex:mobility}. 
In the following we show the evolution of the process  $ P \parallel Q$ starting from the store  $\exists w (\outp(w))$:\\

\resizebox{\textwidth}{!}{
$
\begin{array}{llll}
\mbox{\bf\tiny 1} &\mconf{\emptyset;P\parallel Q;\exists w(\outp(w))} & \redi^* & 
\mconf{\{x\}; \tellp{\outp(x)}\parallel P' \parallel Q ; \exists w(\outp(w))}  \\
\mbox{\bf\tiny 2} && \redi^* & \mconf{\{x\};P'\parallel Q; \exists w(\outp(w))\sqcup \outp(x)}\\
\mbox{\bf\tiny 3} && \redi^* & \mconf{\{x,w\};P'\parallel Q; \outp(w)\sqcup \outp(x)}\\
\mbox{\bf\tiny 4} && \redi^* & \mconf{\{x,w\}; P'\parallel Q_1 \parallel Q'[w/z] ; \outp(w)\sqcup \outp(x) } \\
\mbox{\bf\tiny 5} && \redi^* & \mconf{\{x,w\}; P'\parallel Q_2 \parallel Q'[w/z] \parallel Q'[x/z]; \outp(w)\sqcup \outp(x) } 
\end{array}
$
}\\

where $Q_1= \absp{z}{\outp(z);\{d_{wz}\}}{Q'}$
and $Q_2= \absp{z}{\outp(z);\{d_{wz},d_{xz}\}}{Q'}$.  Observe that $P'$ and $Q'[x/z]$ share the local variable $x$ created by $P$. The derivation from line 2 to line 3   uses the Rule $\rStructVar$ to \emph{open} the scope of $w$ in the store $\exists w (\outp(w))$. 
Let $c_1 = \exists w(\outp(w)) \sqcup \outp(x)$  (store in line 2) and $c_2 = \outp(x)$. We know that $c_1 \equivC c_2$. 
 As we said before,  Rule $\rStruct$  allows us to replace structural congruent processes ($\equiv$) but it does not  modify the store via the   relation $\equivC$  on constraints. The reason is that if we replace $c_1$ in line 2 with  $c_2$, then we will not observe the execution of $Q'[w/x]$. 
%
\end{example}

\subsection{Observables and  Behavior} \label{sec:observables}
In this section we study the input-output behavior of programs and we show that such relation is a function. More precisely, we show that the input-output relation is a (partial) upper closure operator. Then, we characterize the behavior of a process by the sequences of constraints such that the process cannot add any information to them. We shall call this behavior the strongest postcondition. This relation is fundamental to later develop the denotational semantics for \tccp\ and \utcc\ programs. 

Next lemma states some fundamental properties of the internal relation. The proof  follows from simple induction on the inference $\gamma \redi \gamma'$. 
\begin{lemma}[Properties of $\redi$]\label{lemma:redi-properties}
Assume that $ \mconf{\vx; P;c}  \redi \mconf{\vx'; Q;d} $. Then, $\vx \subseteq \vx'$. Furthermore: 
\\\noindent 1. (Internal Extensiveness): $ \exists \vx' (d) \entails \exists \vx(c)$, i.e.,  the store can only be augmented. 
\\\noindent 2. (Internal Potentiality): If $e \entails  c$ and $d \entails  e$    then $\mconf{\vx; P; e} \redi\equiv  \mconf{\vx'; Q;d}$, i.e.,   a stronger store triggers more internal transitions. 
\\\noindent 4. (Internal Restartability): $\mconf{\vx; P;d}  \redi \equiv \mconf{\vx'; Q;d}$.
\end{lemma}
\subsubsection{Input-Output Behavior}
Recall that \tccp\ and \utcc\  allows for the modeling of reactive systems where processes react according to the stimuli (input) from the environment. We  define the behavior of a process $P$ as the relation of its outputs under the influence of a sequence of inputs (constraints) from the environment. Before  formalizing this idea, it is worth noticing that unlike \tccp, some \utcc\ 
processes may exhibit infinitely many internal reductions during a time-unit due to the $\mathbf{abs}$ operator. 

\begin{example}[Infinite Behavior]\label{ex:inf-behav}
Consider a constant symbol ``$a$'', a function symbol $f$, a unary predicate (constraint) $c(\cdot)$ and let $Q=\absp{x}{c(x)}{\tellp{c(f(x))}}$. Operationally, $Q$  in a store  $c(a)$  engages in an infinite sequence of internal transitions producing the constraints  $c(f(a))$, $c(f(f(a)))$, $c(f(f(f(a))))$ and so on. 
\end{example}
The above behavior will arise, for instance,  in applications to security as those in Section \ref{sec:appsec}. We shall see that  the model of the attacker may generate infinitely many messages (constraints) if we do not restrict the length  of the messages (i.e., the number of nested applications of  $f$). 



\begin{definition}[Input-Output Behavior] \label{def:behavior}
Let $s=c_1.c_2...c_n$,  $s'=c_1'.c_2'...c_n'$ 
(resp. $w=c_1.c_2...$, $w'=c_1'.c_2'...$)
be finite (resp. infinite) sequences of
constraints. If   $P = P_1 \rede{(c_1,c_1')} P_2
\rede{(c_2,c_2')}...P_n\rede{(c_n,c_n')}P_{n+1}$
(resp. $P = P_1 \rede{(c_1,c_1')} P_2
\rede{(c_2,c_2')} ...$ )
, we  write $P\rede{(s,s')}$ (resp. $P\redew{(w,w')}$). We define the 
\emph{input-output} behavior of $P$ as $\io{(P)} = \iofin{(P)} \cup \ioinf{(P)}$ where
\[
\begin{array}{lll}
\iofin{(P)} &=& \{(s,s') \ | \ P \rede{(s,s')}\} \mbox{ for } s,s' \in \cC^{*}\\
\ioinf{(P)} &=& \{(w,w') \ | \ P \redew{(w,w')}\} \mbox{ for } w,w' \in \cC^{\omega}
\end{array}
\]
\end{definition}

We recall that  the observable transition ($\rede{}$) is defined through a finite number of internal transitions (rule $\rObserv$ in Figure \ref{opersem}). Hence,   it may be the case that for some \utcc\ processes (e.g., $Q$ in Example \ref{ex:inf-behav}), $\ioinf=\emptyset$.  For this reason, we distinguish finite and infinite sequences in the input-output behavior relation.  We notice that if $w\in \ioinf(P)$ then any finite prefix  of $w$ belongs to $\iofin(P)$.
 We shall call \emph{well-terminated}  the  processes which do not exhibit infinite internal behavior.

 \begin{definition}[Well-termination]\label{def:wellterminated}
The process $P$ is said to be \emph{well-terminated} w.r.t. an infinite sequence $w$ if  there exists $w' \in \cC^\omega$ s.t.   $(w,w')\in \ioinf(P).$
\end{definition}

 Note that  \tccp\ processes are well-terminated since recursive calls must be ${\bf next}$ guarded. The fragment of well-terminated \utcc\ processes has been shown to be a meaningful one. For instance, in \cite{Olarte:08:PPDP} the authors show that such fragment is enough to  encode  Turing-powerful formalisms and  \cite{Lopez-Places09} shows the use of this fragment in the  declarative interpretation of languages for structured communications. 

We conclude here by showing that the \utcc\ calculus is deterministic. The result follows from Lemma \ref{lemma:redi-properties} (see  \ref{app:sos}).
%
 \begin{theorem}[Determinism] \label{theo:SOS-determinism}
Let $s,w$ and $w'$ be (possibly infinite) sequences of constraints. If both
$(s,w)$, $(s,w') \in \iobehav{P}$ then $w\equivC w'$. 
\end{theorem}


\subsubsection{Closure Properties and Strongest Postcondition}
The $\mathbf{unless}$ operator is the only construct in the language that exhibits 
non-monotonic input-output behavior in the following sense: Let $P=\unlessp{c}{Q}$ and  $s \leq s'$.  If $(s,w), (s',w')\in \iobehav{P}$, it may be the case that $w \not\leq w'$. For example, take $Q=\tellp{d}$, $s=\true^\omega$ and  $s'=c.\true^\omega$.
The reader can verify that  $w=\true.d.\true^\omega$, $w'=c.\true^\omega$ and then,  $w \not\leq w'$. 

\begin{definition}[Monotonic Processes]\label{def:monotonic}
We say that $P$ is a monotonic process if it does not have occurrences of ${\bf unless}$ processes. Similarly, the program $\cD.P$ is monotonic if $P$ and all $P_i$ in a process definition $p_i(\vx)\defsymbol P_i$ are monotonic. 
\end{definition}

Now we  show that $\iobehav{P}$  is a \emph{partial upper closure operator}, i.e., it is a function satisfying  \emph{extensiveness} and  \emph{idempotence}. Furthermore, if $P$ is \emph{monotonic},  $\iobehav{P}$ is a \emph{closure operator} satisfying additionally monotonicity. The proof of this result follows from   Lemma \ref{lemma:redi-properties} (see details in \ref{app:sos}). 

\begin{lemma}[Closure Properties]\label{lem:COP}
Let $P$ be a process. Then, $\iobehav{P}$ is a function. Furthermore, $\iobehav{P}$  is a  partial  upper  closure operator,  namely it satisfies:
\\ \noindent {\bf Extensiveness}: If $(s,s') \in  \iobehav{P}$  then $s \leq s'$.\\
\noindent {\bf Idempotence}: If $(s,s') \in  \iobehav{P}$  then $(s',s') \in  \iobehav{P}$.

Moreover, if $P$ is monotonic, then:

\noindent {\bf Monotonicity}: If  $(s_1,s_1') \in  \iobehav{P}$,   $(s_2 ,s_2' ) \in  \iobehav{P}$ and  $s_1 \leq s_2$, then  $s_1' \leq s_2'$.

\end{lemma}

A pleasant property of closure operators is that they are uniquely determined by their set of fixpoints, here called the \emph{strongest postcondition}.
\begin{definition}[Strongest Postcondition]\label{def:sp}
Given a \utcc\ process $P$, the strongest postcondition of $P$,
denoted by $\spbehav{P}$, is defined as the set $\{s\in \cC^\omega \cup \cC^* \ | \ (s,s) \in \iobehav{P}\}$.
\end{definition}

Intuitively,   $s \in \spbehav{P}$ iff $P$ under input $s$ cannot 
add any information whatsoever, i.e.  $s$  is a quiescent sequence for $P$.  We can also think of   $\spbehav{P}$ as the set of sequences that $P$ can output under the influence of an arbitrary environment. 
Therefore, proving whether $P$ satisfies a given  property $A$,  in the presence of any environment, reduces to proving whether  $\spbehav{P}$ is a subset of the set of sequences (outputs) satisfying the property $A$. Recall that $\io(P) = \iofin(P) \cup \ioinf(P)$. Therefore, the sequences in $\spbehav{P}$ can be finite or infinite. 

We conclude here by showing that for the    monotonic fragment,   the input-output behavior can be retrieved 
from  the strongest postcondition. The proof of this result follows  straightforward from Lemma \ref{lem:COP} 
and it can be found in \ref{app:sos}.

\begin{theorem}\label{the:col:CO}
Let $min$ be the minimum function w.r.t. the order induced by $\leq$ and  $P$ be a monotonic process. Then, 
$(s, s') \in \iobehav{P} \mbox{\ \ iff\ \  } 
s' = min(\spbehav{P} \cap \{w \ | \ s \leq w \})
$.
\end{theorem}


\section{A Denotational model for TCC and UTCC}\label{sec:denotsem}

As we explained before,  the strongest postcondition relation fully captures the behavior of a process  considering any possible output under an arbitrary environment. In this section we develop a denotational model for the strongest postcondition. The semantics is compositional and it is the basis for the 
abstract interpretation framework that we 
develop in Section \ref{sec:absframework}.

Our semantics is built on the closure operator semantics for \ccp\ and \tccp\ in \cite{SRP91,tcc-lics94} and \cite{deBoer:97:TOPLAS,NPV02}.
Unlike the denotational semantics for \utcc\ in \cite{Olarte:08:PPDP}, our semantics is more appropriate for the data-flow analysis due to its simpler domain based on sequences of constraints instead of sequences of temporal formulas. In  Section \ref{sec:concluding} we elaborate  more on the differences between both semantics. 

  Roughly speaking,  the semantics is  based on a continuous immediate consequence operator $T_{\cD}$, which computes in a bottom-up fashion  the \emph{interpretation} of each  process definition $p(\vx) \defsymbol P$  in  $\cD$. Such an interpretation is given in terms of the set of the quiescent sequences for  $p(\vx)$.

Assume a \utcc\ program $\cD.P$. We shall denote  the set of process names with their  formal parameters in $\cD$ as ${\it ProcHeads}$. We shall call \emph{Interpretations} the set of functions in the domain ${\it ProcHeads} \rightarrow  {\mathcal{P}}(\C^\omega )$. We shall define the  semantics as a function  $\os \cdot\cs_I: ({\it ProcHeads} \rightarrow {\mathcal{P}}(\C^\omega) ) \rightarrow ({\it Proc} \rightarrow {\mathcal{P}}(\C^\omega) )$
which  given an interpretation $I$, associates to each process a set of  sequences of constraints.

\begin{figure}
{
$
\begin{array}{llcl}
\rdSkip \ \ \ \ & \os \skipp \cs_{I} & = & \cC^\omega \\
\rdTell \ \ \ \ & \os \tellp{c} \cs_{I} & = & \up{c}.\cC^\omega \\

\rdAsk \ \ \ \ & \os \whenp{c}{P} \cs_{I} & = & \overline{\up{c}}.\cC^\omega  \ \cup  \  (\up{c}.\cC^\omega \cap \os P \cs_I) \\


\rdAbs \ \ \ \ & \os \absp{\vx}{c}{P} \cs_{I} & = & \Forall \vx (\os \whenp{c}{P}\cs_I)  \\

\rdPar \ \ \ \ & \os P \parallel Q\cs_{I} & = & \os P \cs_{I} \cap \os Q \cs_{I}  \\

\rdLocal \ \ \ \ & \os \localp{\vx}{P} \cs_{I} & = & 
\Exists \vx(\os P \cs_I)
\\

{\rdNext} \ \ \ \ & \os \nextp P \cs_{I} & = & \cC.\os P\cs_I\\

{\rdUnless} \ \ \ \ & \os \unlessp{c}{P} \cs_{I} & = & \overline{\up{c}}.\os P \cs_I  \ \cup \ \up{c}.\cC^\omega \\




\rdCall \ \ \ \ & \os {p(\vt)} \cs_{I} & = & I(p(\vt))
\end{array}
$
}
\caption{Semantic Equations for \tccp\ and \utcc\ constructs. Operands ``$.$'', $\up$ \ \ ,\ $\Forall$   and $\ \Exists$ are defined in Notation \ref{not:clos-seq}. $\overline{A}$ denotes the set complement of $A$ in $\cCW$. \label{tab:densem} }
\end{figure}
Before  defining the semantics, we introduce the following notation.
 \begin{notation}[Closures and Operators on Sequences]\label{not:clos-seq}
    Given a constraint $c$, we shall use $\up{c}$ (the upward closure)  to denote the set $\{d \in \cC \  |\  d \entails c \}$, i.e., the set of constraints entailing $c$. Similarly, we shall use $\up{s}$ to denote the set of sequences $\{s' \in \cC^\omega \ | \ s \leq s' \}$.
Given $S\subseteq \cC^\omega$ and $\cC' \subseteq \cC$, we shall extend the use of  the sequences-concatenation operator ``$.$'' by declaring that $c.S = \{ c.s \ | \ s \in S\}$, $\cC'.s = \{ c.s \ | \ c \in \cC' \}$ and $\cC'.S = \{ c.s \ | \ c\in \cC' \mbox{ and } s\in S\} $. 
Furthermore, given a set of sequences of constraints $S\subseteq \cC^\omega$, we define:
 \[
 \begin{array}{lll}
 \Exists \vx(S) &=& \{s\in \cC^\omega \ | \ \mbox{ there exists } s'\in S \mbox{ s.t. }  \exists \vx (s) \equivC \exists \vx (s') \}\\
  \Forall \vx (S) &=& \{ \exists \vy(s) \in S \  |  \
  \vy \subseteq {\it Var}, s \in S 
    \mbox{ and for all  } s' \in \cC^\omega, \mbox{ if } 
  \exists \vx (s) \equivC \exists \vx (s') \mbox{,} \\
  & & \qquad\qquad\quad\ 
 \ \dxt^\omega \leq s' \mbox{ and } \adm{\vx}{\vt}  \mbox{ then } s' \in S  \}
 \end{array}
 \]
 \end{notation}
The operators above are used to define the semantic equations  in Figure \ref{tab:densem} and explained in the following. 
 Recall that $\os P \cs_I$ aims at capturing the strongest postcondition (or quiescent sequences) of $P$, i.e. those  sequences $s$ such that $P$ under input $s$ cannot add any information whatsoever. The process $\skipp$ cannot add any information to any sequence and hence, its denotation is $\cC^\omega$ (Equation $\rdSkip$).    The sequences to which $\tellp{c}$ cannot add information are those whose first element entails $c$, i.e., the upward closure of $c$ (Equation $\rdTell$). 
  If neither $P$ nor $Q$ can add any information to $s$, then $s$ is quiescent for $P\parallel Q$. (Equation $\rdPar$). 
  
  We say that  $s$ is an ${\vx}$-variant of $s'$ if
$\exists\vx(s)\equivC\exists\vx(s')$, i.e., $s$ and $s'$ differ only on the information about  $\vx$. Let $S=\Exists \vx(S')$. We note that 
$s\in S$ if there is an $\vx$-variant $s'$ of $s$ in $S'$. 
Therefore, a sequence $s$ is  quiescent for $Q=\localp{\vx}{P}$ if there exists an $\vx$-variant $s'$ of $s$ s.t. $s'$ is quiescent for $P$.  Hence, if $P$ cannot add any information to $s'$ then  $Q$ cannot add any information to $s$ (Equation $\rdLocal$).

The process $\nextp{P}$ has no influence on the first 
element of a sequence. Hence if $s$ is quiescent for $P$ then $c.s$ is quiescent for $\nextp{P}$ for any $c\in \cC$   (Equation $\rdNext$). 
Recall that the process $Q=\unlessp{c}{P}$ executes $P$ in the next time interval if and only if the guard $c$ cannot be deduced from the store in the current time-unit. Then, a sequence $d.s$ is quiescent for $Q$  if either  $s$ is quiescent for $P$ or $d$ entails $c$  (Equation $\rdUnless$). This equation can be equivalently written as  $\cC.\os P \cs_I  \ \cup \ \up{c}.\cC^\omega$.

Recall that the interpretation $I$ maps process names to sequences of constraints. Then, the meaning of  $p(\vt)$ is directly given by the interpretation $I$ (Rule $\rdCall$).

   Let $Q=\whenp{c}{P}$. A sequence $d.s$ is quiescent for $Q$ if $d$ does not entail $c$. If $d$ entails $c$, then  $d.s$ must be quiescent for $P$ (rule $\rdAsk$). In some cases, for the sake of presentation, we may write this equations as:
   \[
   \os \whenp{c}{P} \cs_I = \{ d.s \ | \  \mbox{ if } d\entails c \mbox { then } d.s \in \os P \cs_I \} 
   \]
      Before explaining the Rule $\rdAbs$, let us show some properties of $\Forall \vx (\cdot)$. First, we note that the $\vx$-variables satisfying the condition 
   $\dxt^\omega \leq s$ in the definition of $\Forall$ are equivalent (see the proof in \ref{app:proofs-den}). 
   
   \begin{observation}[Equality and  $\vx$-variants]\label{l-vx-eq}
Let $S \subseteq \cC^\omega$, $\vz \subseteq {\it Var}$ and   $s,w \in \cC^\omega$ be $\vx$-variants such that $\dxt^\omega \leq s$,  $\dxt^\omega \leq w$  and $\adm{\vx}{\vt}$. (1)  $s\equivC w$. (2)   $\exists \vz(s) \in \Forall \vx (S)$ iff $s \in \Forall \vx (S)$. 
\end{observation}
      
Now we  establish the correspondence between 
the sets  $\Forall \vx ( \os P \cs_I)$  and  $\os P\sxt \cs_I$ which is fundamental to  understand the way we defined the operator $\Forall$.

   \begin{proposition}\label{prop:forall-subs}
$s \in  \Forall \vx ( \os P \cs_I)$ if and only if  $s \in \os P\sxt \cs_I$ for all admissible substitution $\sxt$. 
\end{proposition}
\begin{proof}
($\Rightarrow$)Let $s\in\Forall \vx ( \os P \cs_I)$ and  $s'$ be an $\vx$-variant of $s$ s.t. $\dxt^\omega \leq s'$ where $\adm{\vx}{\vt}$. By definition of~ $\Forall$, we know that $s'\in \os P \cs_I$. Since $ \dxt^\omega \leq s'$ then $s' \in \os P \cs_I \cap \up(\dxt^\omega)$. Hence $s\in \Exists \vx (\os P \cs_I \cap \up(\dxt^\omega) )$ and we conclude $s\in \os P\sxt\cs_I$. 

\noindent($\Leftarrow$) Let $\sxt$ be an admissible substitution. Suppose, to obtain a contradiction, that $s \in \os P\sxt\cs_I$, there exists $s'$ $\vec{x}$-variant of $s$ s.t. $\dxt^\omega \leq s'$ and $s' \notin \os P\cs_I$ (i.e., $s\notin \Forall\vx (\os P\cs_I)$). Since $s\in \os P\sxt\cs_I$ then $s\in \Exists \vx (\os P \cs_I \cap \up \dxt^\omega)$.  Therefore, there exists $s''$ $\vx$-variant of $s$ s.t. $s'' \in \os P\cs_I$ and $\dxt^\omega \leq s''$. By Observation \ref{l-vx-eq},  $s'\equivC s''$ and thus,  $s' \in \os P \cs_I$, a contradiction. 
\end{proof}

A sequence $d.s$ is quiescent for  the process $Q=\absp{x}{c}{P}$   if for all admissible substitution $\sxt$, either $d \not\entails c\sxt $ or  $d.s$ is also quiescent for $P\sxt $, i.e., $d.s \in \Forall \vx  (\os (\whenp{c}{P})\cs_I)$ (rule $\rdAbs$).    Notice that we can simply write Equation $\rdAbs$ by unfolding the definition of $\rdAsk$ as follows:
\[
\os \absp{\vx}{c}{P}\cs_I = \Forall \vx(\overline{\up{c}}.\cC^\omega  \ \cup  \  (\up{c}.\cC^\omega \cap \os P \cs_I))
\]
 
 The reader may wonder why the operator $\Forall$ (resp. Rule $\rdAbs$) is not entirely  dual w.r.t. $\Exists$ (resp. Rule $\rdLocal$), i.e., why we only consider $\vx$-variants entailing $\dxt$ where $\sxt$ is an admissible substitution.   
 To explain this issue, let $Q = \absp{x}{c}{P}$ where $c=\outp{(x)}$ and $P=\tellp{\outp'(x)}$. We know that 
\[
s= (\outp(a) \wedge \outp'(a)).\true^\omega \in \spbehav{Q}
\]
 for a given constant $a$.  Suppose that we were to define:
\[
\os Q \cs_I =
{\
\{s \ |\  \mbox{for all $x$-variant $s'$ of $s$  if }   s'(1) \entails c \mbox{ then } s' \in \os P\cs_I\}
}
\]

Let $c'=\outp(a)\wedge \outp'(a)\wedge \outp(x)$
and $s'=c'.\true^\omega$.  Notice that  $s'$ is an  $x$-variant   of $s$,  $s'(1) \entails c$ but $s' \notin \os P \cs_I$ (since $c' \not\entails \outp'(x)$). Then $s \notin \os Q\cs_I$ under this naive definition of $\os Q\cs_I$. We thus consider only the $\vx$-variants $s'$ s.t. each element of $s'$ entails $\dxt$.    Intuitively, this condition 
requires that  $s'(1) \entails c \sqcup \dxt$ in Equation $\rdAbs$ 
and hence that $s'(1) \entails c \sxt$. Furthermore  $s \in \os P\sxt \cs_I$ realizes the operational intuition that $P$ runs under the substitution $\sxt$. 
The operational rule $\rStructVar$ makes also echo in the design of our   semantics: the operator $\Forall$ considers   constraints of the form $\exists \vz(s)$ where $\vz$ is a (possibly empty) set of variables, thus  allowing us to open the existentially quantified constraints as shown in the following example.

\begin{example}[Scope extrusion]
Let $P = \whenp{\outp(x)}{\tellp{\outp'(x)}}$, $Q = \absp{\vx}{\outp(x)}{\tellp{\outp'(x)}}$.
We know that $\os Q \cs_I = \Forall x(\os P\cs_I)$. Assume that $d.s \in \os P\cs_I$. Then, $d$ must be in the set:
\[
C= \{\exists x(\outp(x)), \outp(x) \sqcup \outp'(x), \exists x(\outp(x) \sqcup \outp'(x)) , \outp(y),\outp(y)\sqcup \outp'(y) \cdots\}
\]
where either, $d\not\entails \outp(x)$ or $d \entails \outp'(x)$. 
We note that: (1) $(\exists x(\outp(x))).s \notin \os Q\cs_I$ since $\outp(x)  \not\in C$. Similarly, $\exists y(\outp(y)).s \notin \os Q \cs_I$ since $\outp(y) \in C$ but the $x$-variant $\outp(x)\sqcup d_{xy} \not\in C$ (it does not entail $\outp'(x)$).
(3) $\outp(y).s\not\in \os P \cs_I$ for the same reason. (4) Let $e=(\outp(x)\sqcup \outp'(x))$.
We note that $e.s \in \os Q\cs_I$ since $e\in C$ and
there is not an admissible substitution $[t/x]$ s.t. 
$\exists x(e) \equivC \exists x (e[t/x])$. (5) Let 
$e=(\outp(y)\sqcup \outp'(y))$. Then, $e.s\in \os Q \cs_I$ since $e\in C$ and the $x$-variant $e \sqcup d_{xy} \in C$. (6) Finally, if 
$ e = \exists x (\outp(x)\sqcup \outp'(x)).s$, then $e.s \in \os Q \cs_I$  as in  (4) and (5). 

\end{example}


\subsection{Compositional Semantics}
We choose as semantic domain  $\E=(E,\sqsubseteq^c)$ where $E = \{ X  \ | \  X \in \cP(\cC^\omega) \mbox{ and } \false^\omega \in X\}$ and $X \sqsubseteq^c Y$ iff $X \supseteq Y$. 
The bottom of  $\E$ is then $\cC^\omega$ (the set of all the sequences) and 
the top element is the singleton $\{\false^\omega\}$ (recall that   $\false$ is the greatest element in ($\cC,\leq$)). 
Given two interpretations $I_1$ and $I_2$, we write $I_1 \sqsubseteq^c I_2$ 
iff for all $p$, $I_1(p) \sqsubseteq^{c} I_2(p)$. 

\begin{definition}[Concrete Semantics]\label{def:conc-semantics}
Let $\os \cdot \cs_{I}$ be defined as in Figure \ref{tab:densem}. 
The semantics of a program $\cD.P$ is  the least fixpoint of the continuous 
operator:
\[
\begin{array}{lll}
 T_{\cD}(I)(p(\vt)) =  
\os Q\sxt \cs_{I} \mbox{ if } p(\vx) \defsymbol Q \in 
\cD
\end{array}
\]
We shall use $\os P  \cs$  to represent  $\os P \cs_{\it lfp(T_{\cD})}$.
\end{definition}


In the following we prove some fundamental properties of the semantic operator $ T_{\cD}$, namely, monotonicity and continuity. Before that, 
we shall show that  $\Forall$ is a closure operator and it is continuous on the domain $\E$.  

\begin{lemma}[Properties of $\Forall$]\label{lem:prop-for-all}
$\Forall$ is a closure operator, i.e., it satisfies (1) {\bf Extensivity}: $S \orderc \Forall \vx(S) $; (2) {\bf Idempotency}: $\Forall \vx (\Forall \vx (S)) = \Forall \vx(S)$; and (3)  {\bf Monotonicity}: If $S \orderc S'$ then $\Forall \vx(S) \orderc \Forall \vx(S')$. 
Furthermore,  (4)  $\Forall$ is continuous on $(E,\orderc)$.
\end{lemma}
\begin{proof}
The proofs of (1),(2) and (3) are straightforward from the definition of   $\Forall \vx$. The proof of (4) proceeds as follows. Assume a non-empty ascending chain $S_1 \orderc S_2 \orderc S_3 \orderc ...$.
Lubs in $E$ correspond to set intersection.
We shall prove that $\bigcap \fx (S_i) = \fx (\bigcap S_i)$.
The ``$\subseteq$'' part (i.e., $\sqsupseteq^c$) is trivial since $\Forall$ is monotonic. As for the $\bigcap \fx (S_i) \subseteq  \fx (\bigcap S_i)$ part, by extensiveness we know that $\fx(S_i) \subseteq S_i$ for all $S_i$ and then, $\bigcap \fx(S_i) \subseteq \bigcap S_i$. Let $s \in \bigcap \fx(S_i)$. By definition we know that $s$   and all $\vx$-variant $s'$ of $s$ satisfying $\dxt^\omega \leq s'$ for $\adm{\vx}{\vt}$ belong to  $\bigcap \fx(S_i)$ and then in $\bigcap S_i$. Hence, $s \in \fx(\bigcap S_i)$ and we conclude
$\bigcap \fx (S_i) \subseteq  \fx (\bigcap S_i)$.
\end{proof}

\begin{proposition}[Monotonicity of $\os \cdot \cs $ and continuity of $T_{\cD}$]\label{prop:cont-td}
Let $P$ be a process and 
$I_1 \sqsubseteq^c I_2 \sqsubseteq^c I_3 ...$ be an 
ascending chain. Then,  $\os P \cs_{I_{i}} \sqsubseteq^c \os P \cs_{I_{i+1}} $ ({\bf Monotonicity}). Moreover,    $\os P \cs_{\bigsqcup_ {I_i}} =  \bigsqcup_ {I_i} \os P \cs_{I_i} $ ({\bf Continuity}). 
\end{proposition}
\begin{proof}
Monotonicity follows easily by induction on the structure of $P$ and it implies the the  ``$\sqsupseteq^c$'' part of continuity. As for the part ``$\orderc$'' we proceed by induction on the structure of $P$.  The interesting cases are those of the local  and the abstraction operator. For $P=\localp{\vx}{Q}$, by inductive hypothesis we know that $\os Q \cs_{\bigsqcup_ {I_i}} \sqsubseteq^c  \bigsqcup_ {I_i} \os Q \cs_{I_i} $. Since $\exists$ (and therefore $\Exists$) is continuous (see Property (5) in Definition \ref{def:cs}),  we conclude 
$\Exists _{\vx}(  \os Q \cs_{\bigsqcup_ {I_i}}) \sqsubseteq^c   \bigsqcup_ {I_i} \Exists _{\vx} (  \os Q \cs_{I_i}) $.
The result for  $P=\absp{\vx}{c}{Q}$ 
follows similarly from the continuity of $\Forall$ (Lemma \ref{lem:prop-for-all}). 
\end{proof}

\begin{figure}
\resizebox{.8\textwidth}{!}{
$
\begin{array}{lll}
I_{1}\ : & p \to \up\outp_a(x).\cC^\omega \ \cap \  \cC.\up\outp_a(y).\cC^\omega \mbox{ i.e., }
p \to \up\outp_a(x).\up\outp_a(y).\cC^\omega 
 \\
&  q \to \Forall z(   A.\cC^{\omega} ) \cap  \ \cC.I_\bot(q) 
\mbox{ i.e., }
 q \to \Forall z(A).I_\bot(q) 
\\
&  r \to \cC^\omega \cap \cC^\omega = \cC^\omega \\
I_2\ : & p \to I_1(p) \\
  &  q \to \Forall z (  A.\cC^\omega) \cap  \ \cC. I_1(q)\mbox{ i.e., }
  q \to \Forall z (A).\Forall z (  A.\cC^\omega) \cap  \ \cC. \cC.\cC^\omega
   \\
&  r \to I_1(p) \cap I_1(q) \\
\dots \\
I_{\omega} :& p \to I_1(p) \\
 &  q \to \Forall z (A).\Forall z (A).\Forall z (A)  ...   
 \\
&  r \to I_{\omega}(p) \cap I_{\omega}(q) 
\end{array}
$
}
\caption{Semantics of  the processes   in Example \ref{ex:reduction}.  
 $A_1=\up{(\outp_a(z)\sqcup \outp_b(z))}$, $A_2=\overline{\up{\outp_a(z)}}$ and $A=A_1 \cup A_2$. We abuse of the notation and we 
write $\Forall z(A).S$ instead of $\Forall z(A.\cC^\omega) \cap \cC.S$.
\label{fig:comp-example-den-sem}}
\end{figure}
\begin{example}[Computing the semantics]\label{ex:reduction}
Assume two constraints $\outp_a(\cdot)$ and $\outp_b(\cdot)$, intuitively representing  outputs of names on two different channels $a$ and $b$. Let $\cD$ be the following procedure definitions
 \[
 \begin{array}{lll}
 \cD & = & p() \defsymbol\  \tellp{\outp_a(x)}  \parallel \nextp \tellp{\outp_a(y)}  \\
        &    & q() \defsymbol\  \absp{z}{\outp_a(z)}{(\tellp{\outp_b(z)})} \parallel \nextp q()\\
        &    & r() \defsymbol\  p() \parallel q()
 \end{array}
 \]
 
The procedure  $p()$ outputs on channel $a$ the variables $x$ and $y$ in the first and second time-units respectively. The procedure  $q()$ resends on channel $b$ every message received on channel $a$. 
The computation of $\os r()\cs$ can be found in Figure \ref{fig:comp-example-den-sem}. 
  Let $s \in \os r()\cs$. 
  Then, it must be the case that $s\in \os p()\cs$ and then,  $s(1) \entails \outp_a(x)$ and $s(2) \entails \outp_a(y)$. Since $r\in \os q()\cs$, 
  for $i\geq1$, if $s(i) \entails \outp_a(t)$ then $s(i) \entails \outp_b(t)$ for any term $t$. Hence, $s(1) \entails \outp_b(x)$ and $s(2) \entails \outp_b(y)$. 
\end{example}

\subsection{Semantic Correspondence}
In this section we prove the soundness and completeness of the semantics.

\begin{lemma}[Soundness]\label{lem:soundness}
Let $\os \cdot \cs$ be as in Definition \ref{def:conc-semantics}.
If $P\rede{(d, d')}{R}$ and $d \equivC d'$,  then $d.\os R \cs \subseteq  \os P\cs$.
\end{lemma}
\begin{proof}
Assume that $\langle\vx; P;d\rangle \redi^* \langle \vx'; P';d'\rangle \not\redi$,  $\exists \vx(d) \equivC \exists \vx'(d')$.
 We shall prove that $\exists \vx(d). \Exists \vx' (\os F(P')) \cs \subseteq \Exists \vx(\os P\cs)$. 
	We proceed by induction on the lexicographical order on the length of the internal derivation   and the structure of  $P$, where the predominant component is the length of the derivation. We present the interesting cases. The others can be found in \ref{app:proofs-den}.

\noindent \underline{{\bf Case}  $P=Q \parallel S$}.   Assume a derivation for $Q=Q_1$ and $S=S_1$ of the form
	\[
	\begin{array}{lll}
	\langle \vz ;Q\parallel S,d \rangle &\redi^*& 
	\langle \vz \cup \vx_1\cup\vy_1;Q_1\parallel S_1, c_1 \sqcup e_1 \rangle \\
	&\redi^*& \langle \vz \cup  \vx_i \cup \vy_j; Q_i \parallel S_j , c_i \sqcup e_j\rangle \\
	&\redi^*& \langle \vz \cup  \vx_m\cup\vy_n; Q_m\parallel S_n; c_m \sqcup e_n \rangle \not\redi
	\end{array}
	\]
	such that for $i > 0$, each $Q_{i+1}$ (resp. $S_{i+1}$) is an evolution of $Q_i$ (resp.  $S_i$);
	$\vx_i$ (resp. $\vy_j$) are the variables added by $Q$ (resp. $S$); and $c_{i}$ (resp $e_{j}$)
	is the information added by $Q$ (resp. $S$). We assume by alpha-conversion that $\vx_m \cap \vy_n = \emptyset$. 
	We know that  $\exists \vz(d) \equivC  \exists \vz,\vx_m,\vy_n(c_m \sqcup e_n) $ and from $\rPar$ we can derive:
	\[
	\begin{array}{lll}
	\langle \vz\cup\vy_n; Q ;d\sqcup e_n \rangle &\redi^*\equiv&
  \langle \vz\cup   \vx_{m}\cup\vy_n;Q_m,c_m \sqcup  e_n\rangle \not\redi
	\mbox{\ \ \ and \ \ \ } \\
	\langle \vz\cup \vx_m; S ;d\sqcup c_m \rangle &\redi^*\equiv&
  \langle \vz\cup   \vx_{m}\cup\vy_n;S_n,c_m \sqcup  e_n\rangle \not\redi
	\end{array}
	\]
	By (structural) inductive hypothesis, we know that  $\exists \vz,\vy_n (d\sqcup e_n).\Exists \vz,\vx_m,\vy_n \os F(Q_m)\cs \subseteq \Exists \vz,\vy_n(\os Q\cs$) and also   $\exists \vz,\vx_m(d\sqcup c_m).\Exists \vz,\vy_n,\vx_m \os F(S_n)\cs \subseteq \Exists \vz,\vx_m(\os S\cs)$. We note that  $\Exists \vx( \os P \cs \cap \os Q\cs) = \Exists \vx(\os P\cs) \cap \os Q \cs$ if $\vx\cap\fv(Q) = \emptyset$ (see Proposition \ref{prop:den-se-ext} in \ref{app:proofs-aux}). Hence,  from the fact that   $\vx_m \cap   \fv(S_n) = \vy_n  \cap \fv(Q_m) = \emptyset$,  we conclude: 
	\[ 
	\exists \vz(d).\Exists \vz, \vx_m,\vy_n (\os F(Q_m)\cs \cap \os F(S_n)\cs) \subseteq \Exists \vz (\os Q\cs \cap \os S \cs)
	\] 
\noindent \underline{{\bf Case}  $P = \absp{\vx}{c}{Q}$}.   From the rule $\rAbs$, we can show that
\[
\begin{array}{ll}
\langle \vy; P ; d\rangle & \redi^*   \langle \vy_1;P_1 \parallel Q_1^1[\vec{t_1}/\vx] ; d_1\rangle \\
  & \redi^*   \langle \vy_2; P_2 \parallel Q_1^2[\vec{t_1}/\vx] \parallel Q_2^1[\vec{t_2}/\vx] ; d_2 \rangle \\
 & \redi^*   \langle \vy_3 ; P_3 \parallel Q_1^3[\vec{t_1}/\vx] \parallel Q_2^2[\vec{t_2}/\vx] \parallel Q_3^1[\vec{t_3}/\vx] ; d_3\rangle \\
  & \redi^*  \cdots \\
 & \redi^*   \langle \vy_n; P_n \parallel Q_1^{m_1}[\vec{t_1}/\vx] \parallel Q_2^{m_2}[\vec{t_2}/\vx] \parallel Q_3^{m_3}[\vec{t_3}/\vx] \parallel \cdots \parallel  Q_n^{m_n}[\vec{t_n}/\vx]; d_n\rangle 
 \end{array}
\]
%
%
where  $P_n$ takes the form 
$\absp{\vx}{c;E_n}Q$, $E_n=\{d_{\vx\vec{t_1}},...,d_{\vx\vec{t_n}}\}$ and $\exists \vy(d) \equivC \exists \vy_n (d_n)$. Hence, there is a  derivation (shorter than that for $P$) for each $d_{\vx\vec{t_i}} \in E_n$:
\[
\langle \vy_i;Q_i^1[\vec{t_i}/\vx];d_i\rangle \redi^*\equiv  
\langle\vy_i'; Q_i^{m_i}[\vec{t_i}/\vx];d_i'\rangle \not\redi
\]
with $Q[\vec{t_i}/\vx] = Q^1_i[\vec{t_i}/\vx]$ and 
$\exists \vy_i (d_i) \equivC \exists \vy_i'(d_i')$. 
%
Therefore, by  inductive hypothesis, 
\[\exists \vy_i(d_i).\Exists \vy_i' \os F(Q_i^{m_i}[\vec{t_i}/\vx]) \cs\subseteq \Exists \vec{y_i} \os Q[\vec{t_i}/\vx] \cs
\]
 for all $d_{\vx\vec{t_i}} \in E_n$. We assume, by alpha conversion,  that the variables added for each   $Q_i^j$ are distinct and then, their intersection is empty.   Furthermore, we note that $\exists \vy(d) \equivC \exists \vy_1 (d_1)$. Since  $F(P_n) = \skipp$, we then conclude: 
\[
\exists \vy(d).\Exists \vy_n \os F(P_n \parallel \prod\limits_{d_{\vx\vec{t_i}} \in E_n} Q_i^{m_i}[\vec{t_i}/\vx]) \cs \subseteq \Exists \vy \os\prod\limits_{d_{\vx\vec{t_i}} \in E_n} Q[\vec{t_i}/\vx])   \cs
\]
Let $d.s \in \Exists \vy \os\prod\limits_{d_{\vx\vec{t_i}} \in E_n} Q[\vec{t_i}/\vx])   \cs$. 
	  For an admissible 
	  $\dxt$, either $d \notentails c\sxt$ or $d\entails c\sxt$. In the first case, trivially 
	  $d.s \in \os (\whenp{c}{Q})\sxt\cs$. In the second case, 
	   $E_n  \DEQ \dxt $. Hence,   $d.s \in \os Q\sxt\cs $ and   $d.s \in \os(\whenp{c}{Q})\sxt \cs$. 
	   Here we conclude that for all admissible $\sxt$, $d.s \in \os (\whenp{c}{Q})\sxt\cs$ and by   Proposition \ref{prop:forall-subs} we derive:
	  \[
	  \exists \vy(d).\Exists \vy \os F(\prod\limits_{d_{\vx\vec{t_i}} \in E_n} Q_i^{m_i}[\vec{t_i}/\vx]) \cs \subseteq \Exists \vy \Forall \vx \os (\whenp{c}{Q})  \cs
	  \]

\noindent \underline{{\bf Case}  $P=p(\vt)$}. Assume that 
$p(\vx) :- Q \in  \cD$.    We can verify that 
\[
\langle\vy;  p(\vt); d\rangle \redi \langle \vy;Q\sxt ; d\rangle \redi^* \langle  \vy' ; Q'; d'\rangle \not\redi
\]
where $\exists \vy' (d') \equivC \exists \vy(d)$. 
By induction  $\exists \vy(d).\Exists \vy' \os F(Q')\cs \subseteq \Exists \vy \os Q\sxt  \cs$ and we conclude 
$\exists \vy(d).\Exists \vy \os F(Q')\cs \subseteq \Exists \vy \os p(\vt)\cs$. 
%
\end{proof}

The previous lemma allows us to prove the soundness of the semantics. 
\begin{theorem}[Soundness]\label{theo:sound}
If $s\in \spbehav{P}$ then there exists $s'$ s.t. $s.s' \in \os P\cs$. 
\end{theorem}
\begin{proof}
If $P$ is well-terminated under input $s$, let $s'=\epsilon$. By repeated applications of Lemma \ref{lem:soundness}, $s \in \os P\cs$. If $P$ is not well-terminated, then $s$ is finite and let $s'=\false^\omega$
(recall that $\false^\omega$ is quiescent for any process). Via Lemma \ref{lem:soundness} we can show $s.s' \in \os P\cs$.
\end{proof}

Moreover, the semantics approximates any infinite computation.
\begin{corollary}[Infinite Computations]
       Assume that $d.s \in \Exists \vx_1 (\os P_1\cs \cap \up(c_1.\cC^\omega))$ and  that $\mconf{\vx_1;P_1;c_1} \redi^* \mconf{\vx_i;P_i;c_i} \redi^*  \mconf{\vx_n;P_n;c_n} \redi^* \cdots.$ Then,  $\bigsqcup  \exists \vx_i(c_i) \leq d$. 
\end{corollary}
\begin{proof}
 Recall that procedure calls must be next guarded. Then, any infinite behavior in $P_1$ is due to a process of the form $\absp{\vx}{c}{Q}$
that executes $Q[\vt_i/\vx]$ and adds new information of the form $e[\vt_i/\vx]$. By an analysis  similar to that of Lemma \ref{lem:soundness}, we can show that $d$ entails $e[\vt_i/\vx]$. 
\end{proof}


\begin{example}[Infinite behavior]
Let  $P=\absp{z}{\outp(z)}{\localp{x}{(\tellp{\outp(x))}}}$ and 
let $c=\outp(w)$. Starting from the store $c$, the process $P$   engages in infinitely many internal transitions of the form 
\[
\begin{array}{c}
\mconf{\emptyset;P ;c} \redi^* \mconf{\{x_1,\cdots, x_i\}; P_i; \outp(x_1)\sqcup \cdots \sqcup \outp(x_i)\sqcup \outp(w)} \redi^* \\
\mconf{\{x_1,\cdots,x_i,\cdots, x_n\}; P_n; \outp(x_1)\sqcup \cdots \sqcup \outp(x_n)\sqcup \outp(w)} \redi^* \cdots
\end{array}
\]
At any step of the computation, the observable store is $\outp(w) \sqcup \bigsqcup\limits_{i\in 1..n}\exists x_i  \outp(x_i) $ which is equivalent to $\outp(w)$. Note also that $\outp(w).\cC^\omega \in \os P \cs$. 
\end{example}

 For the converse of Theorem \ref{theo:sound}, we have similar technical problems as in the case of \tccp,  namely: the combination of the $\mathbf{local}$ operator with the $\mathbf{unless}$ constructor.  Thus, similarly to \tccp, completeness is verified only for the  fragment of \utcc\ where there are no occurrences of $\mathbf{unless}$ processes in the body of $\mathbf{local}$ processes. The reader may refer \cite{BoerPP95,NPV02} for  counterexamples showing that  $\os P\cs \not\subseteq \spbehav{P}$ when $P$ is not locally independent. 

 \begin{definition}[Locally Independent Fragment] \label{def:LI}
Let $\cD.P$ be a program where $\cD$ contains process definitions of the form $p_i(\vx) \defsymbol P_i$.
We say that $\cD.P$ is 
 locally independent if for each process of the form $\localp{\vx;c}{Q}$ in $P$ and $P_i$  it holds that 
(1)  $Q$ does not have occurrences of $\mathbf{unless}$ processes; and (2) if $Q$ calls to $p_j(\vx)$, then $P_j$ satisfies also conditions (1) and (2).
\end{definition}

\begin{lemma}[Completeness]\label{theo:comp}
Let  $\cD.P$ be a locally independent program s.t. $d.s \in \os P \cs$. If $P\rede{(d, d')} R$ then $d' \equivC d$ and 
$s \in \os R \cs$. 
\end{lemma}
\begin{proof} 
Assume that   $P$ is locally independent, $d.s \in \os P \cs$ and there is a derivation of the form  $\langle \vx;P;d\rangle \redi^*\langle \vx';P';d'\rangle \not\redi$. We shall prove that $\exists \vx(d) \equivC \exists \vx'(d')$ and
$s\in \Exists \vx' \os  F(P')\cs$. 
	We proceed by induction on the lexicographical order on the length of the internal derivation  ($\redi^*$)  and the structure of  $P$, where the predominant component is the length of the derivation. 
 The locally independent condition is used for the case $P=\localp{\vx;c}{Q}$. We only present the interesting cases. The others can be found in \ref{app:proofs-den}.

\noindent \underline{{\bf Case}  $P=Q \parallel S$}. We know that $d.s \in \os Q\cs$ and $d.s \in \os S\cs$ and by (structural) inductive hypothesis,  there are derivations  $\langle \vz; Q;d\rangle \redi^* \langle \vz\cup \vx'; Q';d'\sqcup c\rangle\ \not\redi $ and $\langle \vz ;S;d\rangle \redi^* \langle \vz\cup  \vy'; S';d'' \sqcup e\rangle\ \not\redi $ s.t.  $s \in  \Exists  \vz,\vx' \os F(Q')\cs$,  $ s\in  \Exists \vz,\vy'\os  F(S')\cs$, $  \exists \vz(d) \equivC \exists \vz,\vx'(d' \sqcup c)$ and $  \exists \vz(d) \equivC   \exists \vz,\vy' (d''\sqcup  e)$. Therefore,
assuming by alpha conversion that $\vx' \cap \vy' = \emptyset$, 
 $\exists \vz(d) \equivC \exists \vz,\vx',\vy' (d' \sqcup d'' \sqcup c \sqcup e)$ and by  rule $\rPar$, 
\[
 \langle \vz; Q\parallel S, d \rangle \redi^* \equiv \langle \vz\cup\vx'\cup \vy' ;Q'\parallel S'; d'\sqcup d'' \sqcup c \sqcup e \rangle \not\redi
\]

We note that  $\Exists \vx( \os P \cs \cap \os Q\cs) = \Exists \vx(\os P\cs) \cap \os Q \cs$ if $\vx\cap\fv(Q) = \emptyset$ (see Proposition \ref{prop:den-se-ext} in \ref{app:proofs-aux}). Since $F(Q'\parallel S') = F(Q') \parallel F(S')$
and $\vx' \cap \fv(S') = \vy' \cap \fv(Q')=\emptyset$,  we conclude $s\in \Exists \vz,\vx',\vy' (\os F(Q'\parallel R')\cs)$.


 
\noindent \underline{{\bf Case}  $P = \absp{\vx}{c}{Q}$}.   
	 By using the rule $\rAbs$ we can show that:
	 \[
\begin{array}{ll}
\langle \vx; P ; d\rangle & \redi^*   \langle \vy_1; P_1 \parallel Q_1^1[\vec{t_1}/\vx] ;   d_1^1\rangle \\
  & \redi^*   \langle \vy_2; P_2 \parallel Q_1^2[\vec{t_1}/\vx] \parallel Q_2^1[\vec{t_2}/\vx] ; d_1^2 \sqcup  d_2^1\rangle \\
 & \redi^*   \langle \vy_3 ; P_3 \parallel Q_1^3[\vec{t_1}/\vx] \parallel Q_2^2[\vec{t_2}/\vx] \parallel Q_3^1[\vec{t_3}/\vx] ;  d_1^3 \sqcup d_2^2 \sqcup d_3^1\rangle \\
  & \redi^*  \cdots \\
 & \redi^*   \langle \vy_n; P_n \parallel Q_1^{m_1}[\vec{t_1}/\vx] \parallel \cdots \parallel  Q_n^{m_n}[\vec{t_n}/\vx] ; d_1^{m_1} \sqcup ... \sqcup d_n^{m_n}\rangle 
 \end{array}
\]
%
%
where $P_n$ takes the form 
$\absp{\vx}{c;E_n}Q$ and $E_n = \{d_{\vx\vec{t_1}},...,d_{\vx\vec{t_n}}\}$.
In the derivation above,  $d_i^j$ represents the constraint  added by  $Q_i^j[\vec{t_i}/\vx]$. Note that  $Q[\vec{t_i}/\vx]=Q_i^1[\vec{t_i}/\vx]$.  There is a derivation (shorter than that for $P$) for each $d_{\vx\vec{t_i}} \in E_n$ of the form
\[
\langle \vy_i; Q_i^1[\vec{t_i}/\vx];d_i \rangle \redi^*  \equiv
\langle \vy_i'; Q_i^{m_i}[\vec{t_i}/\vx];  d_i^{m_i}\rangle \not\redi
\]
Since $d.s \in \os P\cs$, by Proposition \ref{prop:forall-subs} we know that  $d.s \in \os Q_i^1[\vec{t_i}/\vx] \cs$ and by induction,  $ \exists \vy_i(d_i) \equivC \exists \vy_i' ( d_i^{m_i})$. Furthermore,   it must be the case that $s\in \Exists \vy_i' \os F(Q_i^{m_i}[\vec{t_i}/\vx])\cs$. 
Let $e$ be the constraint $\exists \vy_n( d_1^{m_1} \sqcup ... \sqcup d_n^{m_n})$.  Given that $ \exists \vy_i(d_i) \equivC \exists \vy_i' ( d_i^{m_i})$, we have $\exists \vx(d) \equivC e$. Furthermore, given that  $F(P_n) = \skipp$:
\[
\absp{\vx}{c}{Q}\rede{(d,e)} \localp{\vy_n}F\left(\prod\limits_{d_{\vx\vec{t_i}} \in E_n} Q_i^{m_i}[\vec{t_i}/\vx]\right)
\]
Since  $s\in \Exists \vy_i' \os F(Q_i^{m_i}[\vec{t_i}/\vx])\cs$
for all $d_{\vx\vt_i} \in E_n$,   we conclude 
\[s \in \Exists \vy_n\os F(\prod\limits_{d_{\vx\vec{t_i}} \in E_n} Q_i^{m_i}[\vec{t_i}/\vx])\cs\]
\noindent \underline{{\bf Case}  $P = \localp{\vx}{Q}$}. By alpha conversion assume $\vx \not\in \fv(d.s)$. We know that there exists $d'.s'$ ($\vx$-variant of $d.s$) s.t. $d'.s' \in \os Q \cs$, $\exists \vx (d.s) \equivC d.s$ and  $ d.s \equivC \exists \vx(d'.s')$. By (structural) inductive hypothesis, there is a derivation 
$\langle \vy; Q;d'\rangle \redi^* \langle \vy';Q'; d''\rangle\not\redi$ and 
$\exists \vy(d') \equivC \exists \vy' ( d'')$ and 
$s' \in \Exists \vy' \os F(Q')\cs$. 
We assume by alpha conversion that $\vx \cap \vy = \emptyset$. Consider now the following  derivation:
\[
\langle \vy;\localp{\vx}{Q}; d \rangle
\redi
\langle  \vx\cup \vy;Q ; d \rangle
 \redi^* \langle \vy''; Q'', c\rangle\not\redi
\]
where $\vx\cup \vy \subseteq \vy''$.  We know that $d' \entails d$ and by monotonicity, we have $\exists \vy' ( d'') \entails \exists \vy''(c)$ and then,  $d' \entails  \exists \vy''(c)$. We then conclude  $\exists \vy(d) \entails  \exists \vy''(c)$.

Since  $s' \in \Exists \vy' \os F(Q')\cs$ then $s \in \Exists \vx \Exists \vy' \os F(Q') \cs$. Nevertheless, notice that in the above derivation of $\localp{\vx}{Q}$, the final process is $Q''$ and not $Q'$. Since $Q$ is monotonic, there are no $\mathbf{unless}$ processes in it. Furthermore, since $d' \entails d$, it must be the case that $Q'$ may contain sub-terms (in parallel composition) of the form  $R'\sxt $ resulting from  a process of the form $\absp{\vy}{e}{R}$ s.t. $d''\entails e\sxt $ and $c\not\entails e\sxt $. 
 Therefore, 
by Rule $\rdPar$, it must be also the case that $s' \in \os F(Q'')\cs$ and then,  $s \in \Exists \vx,\vy' \os {F(Q'')}\cs$. Finally, note that $\vy''$ is not necessarily equal to $\vy'$. With a  similar analysis
 we can show that in $Q'$ there are possibly more $\mathbf{local}$ processes running in parallel than in $Q''$ and then,  $s \in \Exists \vy'' \os {F(Q'')}\cs$.
%
%
%
%
\end{proof}

By repeated applications of the previous Lemma, we show the completeness of the denotation with respect to the strongest postcondition relation. 
\begin{theorem}[Completeness]\label{theo:completeness}
Let  $\cD.P$ be a locally independent program,  $w=s_1.s_1'$ and $w \in \os P\cs$. If $P\rede{(s_1,s_1')}$ then $s_1 \equivC s_1'$. 
Furthermore, if  $P\rede{(w,w')}_\omega$ then 
$w \equivC w'$. 
\end{theorem}
Notice that completeness of the semantics holds only for the locally independent fragment, while soundness is achieved for the whole language. For the abstract interpretation framework we 
develop in the next section, we require the semantics to be a sound approximation of the operational semantics and then, the restriction imposed for completeness does not affect the applicability of the framework.


\section{Abstract Interpretation Framework}\label{sec:absframework}
In this section we develop an abstract interpretation framework \cite{CC92} 
for the analysis of \utcc\ (and \tccp) programs. 
The framework is based on the above denotational semantics, thus allowing for a 
compositional analysis. 
The abstraction  proceeds as a composition of two different abstractions:  (1) we  abstract the constraint system and then (2) we abstract the  infinite sequences of \emph{abstract} constraints. The abstraction in (1) allows us to reuse the most popular abstract domains previously defined for logic programming. Adapting those domains, it is possible to perform, e.g., groundness, freeness, type and suspension analyses of \utcc\ programs. 
 On the other hand, the abstraction in (2) along with (1) allows for computing the approximated output of the program in a finite number of steps. 

\subsection{Abstract Constraint Systems}
Let us recall some notions from \cite{Falaschi:97:TCS} and \cite{ZaffanellaGL97}.

\begin{definition}[Descriptions] \label{def:description}
A description $(\cC,\alpha, \cA)$ between  two constraint systems 
\[
\begin{array}{lll}
{\bf C} &=& \langle \cC,\leq\ , \sqcup,\true, \false, {\it Var}, \exists, D \rangle \\
{\bf A} &=&  \langle \cA,\leq^\alpha, \sqcup^\alpha,\true^\alpha, \false^\alpha, {\it Var}, \exists^\alpha, D^\alpha \rangle
\end{array}
\]
  consists of an abstract domain  $(\cA,\leq^\alpha)$ and a surjective and monotonic abstraction function $\alpha: \cC \to \cA$. We lift $\alpha$ to sequences of constraints in the obvious way. 
\end{definition}

We shall use $c_{\alpha}$, $d_{\alpha}$ to range over constraints in ${\cA}$ and $s_\alpha, s'_\alpha, w_\alpha, w'_\alpha,  $ to range over sequences in $\cA^*$ and $\cA^\omega$ (the set of finite and infinite sequences of constraints in $\cA$). To simplify the notation, we omit the 
subindex ``$\alpha$'' when no confusion arises.  The entailment $\absentails$ is defined as in the concrete counterpart, i.e.  $c_{\alpha} \leq^\alpha d_{\alpha}$ iff $d_{\alpha} \absentails c_{\alpha}$. Similarly,   $d_\alpha \equivC_\alpha c_\alpha$ iff 
$d_{\alpha} \absentails c_{\alpha}$ and 
$c_{\alpha} \absentails d_{\alpha}$. 

Following standard lines in \cite{DBLP:journals/jlp/GiacobazziDL95,Falaschi:97:TCS,ZaffanellaGL97}  we impose the following restrictions over $\alpha$ relating  the cylindrification, diagonal and $lub$ operators of ${\bf C}$ and ${\bf A}$.

\begin{definition}[Correctness]\label{dec:corapp}
Let $\alpha : \cC \to \cA$ be monotonic and surjective. We say that ${\bf A}$ is \emph{upper correct} w.r.t. the constraint system 
 ${\bf C} $ if for all $c\in \cC$ and $x,y \in Var$:

\noindent (1)  $\alpha(\exists \vx (c)) \equivC_\alpha \exists^\alpha \vx (\alpha(c))$.

\noindent (2) $\alpha(\dxt) \equivC_\alpha \dxta$.

Since $\alpha$ is monotonic, we also have $\alpha(c \sqcup d) \absentails \alpha(c) \sqcup^\alpha \alpha(d)$. 

\end{definition}


In the example below we illustrate an abstract domain for the groundness analysis of \tccp\ programs. Here we give just an intuitive description of it. We shall elaborate more on this domain and its applications in  Section \ref{sec:ground}.
 
\begin{example}[Constraint System for Groundness]\label{ex:hcs}
Let the concrete constraint system ${\bf C}$ be the  Herbrand constraint system. As abstract constraint system {\bf A}, let constraints be propositional formulas 
 representing groundness information as in  $x \wedge (y \leftrightarrow z)$ that means, $x$ is a ground variable and,  $y$ is ground iff $z$ is ground. 
 In this setting,  $\alpha(x=[a])= x$ (i.e., $x$ is a ground variable). Furthermore, $\alpha(x=[a | y]) = x \leftrightarrow y$ meaning $x$ is ground if and only if $y$ is ground. 

\end{example}

In the following definition we  make precise the idea  when an  abstract constraint   approximates a concrete one. 

\begin{definition}[Approximations]\label{def:approximations}
Let $(\cC, \alpha,  \cA)$ be a description  satisfying the conditions in Definition \ref{def:description}.   Given $d_{\alpha} = \alpha(d)$, we say that $d_{\alpha}$ is the best approximation of $d$. Furthermore, for all $c_{\alpha} \leq^\alpha d_{\alpha}$ we say that $c_{\alpha}$ approximates $d$ and we write $c_{\alpha} \propto d$.  This definition is pointwise extended  to sequences of constraints in the obvious way (see Figure \ref{fig:abs-domains}a). 

\end{definition}

\begin{figure}
\resizebox{11cm}{!}{
\subfloat[]{\includegraphics{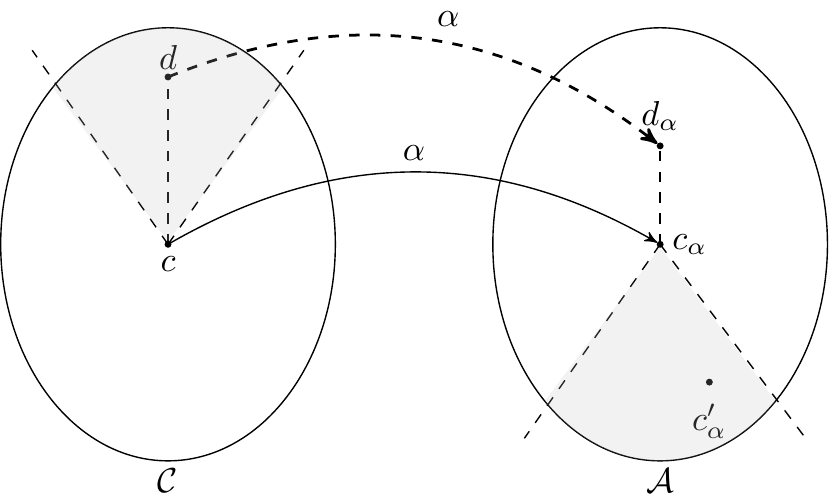}}
\qquad
\subfloat[]{\includegraphics{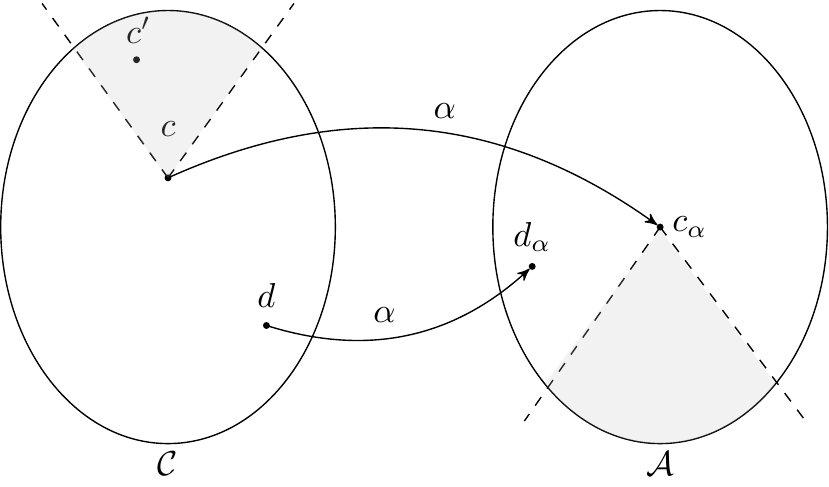}} 
}
\caption{(a). $c'_\alpha$ approximates $c$ (i.e., $c'_{\alpha} \propto c$) and $c_\alpha = \alpha(c)$ is the best approximation of $c$ (Definition \ref{def:approximations}). Since  $\alpha$ is monotonic and $c\leq d$,  $c_\alpha \leq^\alpha d_\alpha$. In (b), assume that for all $d$ s.t. $d\notentails c$, $d$ is not approximated by $c_\alpha$. Then, all  constraint $c'$ approximated by $c_\alpha$ (the upper cone of $c$) entails $c$. In this case, $c_\alpha \absconcentails c$ (Definition \ref{def:absconcentails}). 
 \label{fig:abs-domains}}

\end{figure}

\subsection{Abstract Semantics}\label{sec:abssemantics}
Now we define an abstract semantics that  approximates the observable behavior of a program and is adequate for modular data-flow analysis. The semantic equations are given in Figure \ref{absdensems} and they are parametric on the abstraction function $\alpha$ of the description $(\cC,\alpha,\cA)$. We shall dwell a little upon the description of the  rules  $\raAsk$ and $\raUnless$. The other cases are self-explanatory.
 
 
\begin{figure}
{
$
\begin{array}{llcl}
\raSkip \ \ \ \ & \os \skipp \cs^{\alpha}_{X} & = & \cA^{\omega} \\

\raTell \ \ \ \ & \os \tellp{c} \cs^{\alpha}_{X} & = & \up(\alpha(c)).\cA^\omega \\

\raAsk \ \ \ \ & \os \whenp{c}{P} \cs^{\alpha}_{X} & = & 
\overline{\upa{c}}.\cA^\omega \cup (\upa{c}.\cA^\omega \cap \os P \cs^\alpha_X)
\\

\raAbs \ \ \ \ & \os \absp{\vx}{c}{P} \cs^{\alpha}_{X} & = & \Forall \vx  (\os \whenp{c}{P}  \cs^\alpha_X) \\ 

\raPar \ \ \ \ & \os P \parallel Q\cs^{\alpha}_{X} & = & \os P \cs^{\alpha}_{X} \cap \os Q \cs^{\alpha}_{X}  \\

\raLocal \ \ \ \ & \os \localp{\vx}{P} \cs^{\alpha}_{X} & = & \Exists \vx (  \os  P\cs^{\alpha}_X )
\\

\raNext \ \ \ \ & \os \nextp P \cs^{\alpha}_{X} & = &  \cA.\os P \cs^\alpha_X\\
\raUnless \ \ \ \ & \os \unlessp{c}{P} \cs^{\alpha}_{X} & = & \cA^\omega\\


\raCall \ \ \ \ & \os p(\vt) \cs^{\alpha}_{X} & = & X(p(\vt))
\end{array}
$
}
\caption{Abstract denotational semantics for \utcc.  $\absconcentails$ and $\upa{}$\ \  are  in Definition \ref{def:absconcentails}. $\overline{A}$ denotes the set complement of $A$.\label{absdensems}}

\end{figure}

Given the right abstraction of the synchronization mechanism of blocking asks in \ccp\ is crucial to give a safe approximation of the behavior of programs. In abstract interpretation, abstract elements are \emph{weaker} than the concrete ones. Hence, if we approximate the behavior of $\whenp{c}{P}$ by replacing the guard $c$ with  $\alpha(c)$, it could be the case that $P$ proceeds in the abstract semantics but it does not in the concrete one. More precisely,  let $d,c \in \cC$. Notice that  from  $\alpha(d) \absentails \alpha(c)$ 
we cannot, in general,  conclude $d \entails c$. Take for instance  the constraint systems in Example \ref{ex:hcs}. We know that $\alpha(x=a) \equivC^\alpha \alpha(x=b)$ but $x=a \not\entails x=b$. 
Assume now we were to define the abstract semantics of  ask processes as:
\begin{equation}\label{ask-wrong-eq}
\os \whenp{c}{Q}\cs^\alpha_X =   \overline{\uparrow(\alpha(c))}.\cA^\omega \cup (\uparrow(\alpha(c)).\cA^\omega \cap \os Q\cs^\alpha_X) 
\end{equation}
A correct analysis of the process
 $P=\tellp{x=a} \parallel \whenp{x=b}{\tellp{y=b}}$  
 should conclude that only $x$ is definitely ground. 
 Since $\alpha(x=a) \entails^\alpha \alpha (x=b)$, 
 if  we use   Equation \ref{ask-wrong-eq},  the analysis   ends with the result $(x\wedge y).\cA^\omega$, i.e., it wrongly concludes that $x$ and $y$ are definitely ground. 
  
  We  thus follow 
\cite{ZaffanellaGL97,FalaschiGMP93,Falaschi:97:TCS} for the abstract semantics of the ask operator. 
For this, we need to define the entailment $\absconcentails$ that relates constraints in $\cA$ and $\cC$. 

\begin{definition}[$\absconcentails$ relation]\label{def:absconcentails}
Let $d_\alpha\in \cA$ and $c\in \cC$. 
We say that $d_\alpha$ entails $c$, notation $d_\alpha \absconcentails c$, if for all $c'\in \cC$ s.t. $d_\alpha \propto c'$ it holds that $c' \entails c$.  We shall use $\upa{c}$ to denote the set $\{d_\alpha \in \cA \ | \ d_\alpha \absconcentails c\}$.
\end{definition}

In words, the (abstract) constraint $d_\alpha$ 
entails the (concrete) constraint $c$ if all constraints 
approximated by $d_\alpha$ entail $c$ (see Figure \ref{fig:abs-domains}b).  Then, in Equation $\raAsk$, 
we guarantee that if  the abstract computation proceeds
(i.e., $d_\alpha \absconcentails c$) then every concrete computation 
it approximates proceeds too.


In Equations  $\rdAbs$ and $\rdLocal$ we use the operators $\Forall$ and $\Exists$ analogous to those in  Notation \ref{not:clos-seq}.  In this context, they are defined on sequences of constraints in $\cA^\omega$ and they use the elements  $\exists^\alpha$,  $\sqcup^\alpha$ and  $\dxta$ instead of their concrete counterparts:
 \[
 \begin{array}{lll}
 \Exists \vx(S_\alpha) &=& \{s_\alpha\in \cA^\omega \ | \ \mbox{ there exists } s'_\alpha\in S_\alpha \mbox{ s.t. }  \exists^\alpha \vx (s_\alpha) \equivC_\alpha \exists^\alpha \vx (s'_\alpha) \}\\
   \Forall \vx (S_\alpha) &=& \{ \exists^\alpha \vy(s_\alpha) \in S_\alpha \  |  \
  \vy \subseteq {\it Var}, s_\alpha \in S_\alpha 
    \mbox{ and for all  } s'_\alpha \in \cA^\omega, \\
  & & \quad\    \mbox{ if } 
  \exists^\alpha \vx (s_\alpha) \equivC \exists^\alpha \vx (s'_\alpha) \mbox{,} 
 (\dxt^\alpha)^\omega \leq s_\alpha' \mbox{ and } \adm{\vx}{\vt}  \mbox{ then } s'_\alpha \in S_\alpha  \}
 \end{array}
 \]

 We omitted the superindex ``$\alpha$'' in these operators since it can be easily inferred from the context.


The abstract semantics of the $\mathbf{unless}$ operator poses similar difficulties as in the case of the ask operator. Moreover, even if we 
make use of the entailment 
$\absconcentails$  in Definition \ref{def:absconcentails}, we do not obtain a safe approximation. Let us explain this. 
One could think of defining  the semantic equation for the {\bf unless} process as follows:
\begin{equation} \label{eq:unless-1}
\os \unlessp{c}{Q} \cs^{\alpha}_{X}=
\overline{\upa{c}}.\os Q\cs^{\alpha}_X \cup \upa{c}.\cA^\omega
\end{equation}
The problem here is that $\alpha(d)  \not \absconcentails c$  does not imply, in general,     $d \not\entails c$. 
Take for instance $\alpha$ in Example \ref{ex:hcs}. We know that  $x \not\absconcentails x=[a]$ and $x=[a] \entails x=[a]$. 
Now let  $Q=\unlessp{c}{\tellp{e}}$, $d$ be a constraint s.t. $d\entails c$
and $d_{\alpha}= \alpha(d)$. 
 We know by rule $\rdUnless$ that  $d.\true^\omega \in \os Q \cs$. If $\alpha(d) \not\absconcentails c$, then 
by using the Equation (\ref{eq:unless-1}), we conclude that 
$  d_{\alpha}.(\true^{\alpha})^\omega \notin \os Q\cs^{\alpha}
$. Hence, we have a sequence $s$ such that  $s \in \os Q\cs$ and $\alpha(s) \not \in \os Q\cs^{\alpha}$ and the abstract semantics cannot be shown to be   a sound approximation of the concrete semantics (see Theorem \ref{teo:corr}). 

Notice that defining $d_\alpha \not\absconcentails c$ as true iff  $c' \notentails c$
for  all $c'$ approximated by  $d_\alpha $   does not solve the problem. This is  because 
under this definition,   $d_\alpha \not\absconcentails c$ does not hold 
for any $d_\alpha$ and $c$. To see this, notice that  
$\false$ entails all the concrete constraints 
and it is approximated by any abstract constraint.
Therefore,  we cannot give a better (safe) approximation of the semantics of  $\unlessp{c}{P}$ than $\cA^\omega$ (Rule $\raUnless$).


Now we can formally define the abstract semantics as we did in Section \ref{sec:denotsem}. Given a description $(\cC,\alpha, \cA)$, we choose as  abstract domain is $\A = (A, \sqsubseteq^\alpha)$ where 
$A=\{ X \ | \ X \in \cP(\cA^\omega) \mbox { and } (\false^\alpha)^\omega \in X \}$  and $X\sqsubseteq^\alpha Y$ iff $X \supseteq Y$. 
The bottom and top of this domain are similar to the concrete domain, i.e.,  $\cA^\omega$ and $\{(\false^\alpha)^\omega\}$ respectively.

\begin{definition}\label{def:tpalpha}  Let  $\os \cdot \cs^{\alpha}_{X} $ be  as in Figure \ref{absdensems}. The abstract semantics of a program $\cD.P$ is defined as the 
least fixpoint of the  continuous  semantic operator:
\[
		\TR[\alpha]{\cD}(X)(p(\vt)) = 
		{\os (Q\sxt) \cs^{\alpha}_{X}}
		\mbox{ if } p(\vec{x}) \defsymbol Q \in 
		\mathcal{D}
\]
We shall use $\os P \cs^\alpha $ to denote $\os P \cs^\alpha_{\it lfp(T_{\cD}^{\alpha})}$.
\end{definition} 

The following proposition shows the monotonicity of $\os \cdot \cs^\alpha$ and 
the continuity of $T_{\cD}^\alpha$. The proof is analogous to that of Proposition  \ref{prop:cont-td}.

\begin{proposition}[Monotonicity of $\os \cdot \csa $ and Continuity of $T_{\cD}^\alpha$]
Let $P$ be a process and  $X_1 \sqsubseteq^{\alpha} X_2 \sqsubseteq^{\alpha} X_3 ...$ be an ascending chain. Then, 
$\os P \csa_{X_i} \sqsubseteq^c \os P \csa_{X_{i+1}} $ ({\bf Monotonicity}). Moreover,  $\os P \csa_{\bigsqcup_ {X_i}} =  \bigsqcup_ {X_i} \os P \csa_{X_i} $ ({\bf Continuity}). 
\end{proposition}

\subsection{Soundness of the Approximation}
This section proves the correctness of the abstract semantics in Definition \ref{def:tpalpha}. We first establish a Galois insertion between the concrete and the abstract domains. 

\begin{proposition}[Galois Insertion]\label{prop:gc}
Let $(\cC,\alpha',\cA)$ be a description and $\E$,  $\A$ be  the concrete and  abstract domains.  If $\mathbf{A}$ is  upper correct w.r.t. $\mathbf{C}$  then there exists an upper Galois insertion $\E \galois{\alpha}{\gamma} \A$. 
\end{proposition}
\begin{proof}
Let $\A = (A , \sqsubseteq^\alpha)$, $\E = (E, \sqsubseteq^c)$ and   $\alpha:E \to A$ and $\gamma:A \to E $ be defined as follows:
\[
\begin{array}{ll}
\alpha(S) &= \{\beta(s) \ | \  s \in S\} \mbox{ for } S \in \{ X  \ | \  X \in \cP(\cC^\omega) \mbox{ and } \false^\omega \in X\}  \\
\gamma(S_\alpha) &= \{s \ |\  \beta(s) \in S_\alpha \} \mbox{ for } S_\alpha  \in \{ X \ | \ X \in \cP(\cA^\omega) \mbox { and } (\false^\alpha)^\omega \in X\}  
\end{array} 
\]
where $\beta$ is the pointwise extension of $\alpha'$ over sequences. Notice that $\beta$ is a monotonic and surjective function between $\cC^\omega$ and $\cA^\omega$ and set intersection is the lub in both $\E$ and $\A$. We conclude by the fact that any additive and surjective function between complete lattices defines a Galois insertion \cite{CC79}. 
\end{proof}

 We  lift, as  standardly done in abstract interpretations \cite{CC92},
the approximation induced by the above abstraction. 
Let $I: ProcHeads \rightarrow E$, $X: ProcHeads \rightarrow A$, $\beta$ be as in Proposition \ref{prop:gc} 
and $p$ be a process definition. Then
\[
\begin{array}{lll}
	\alpha(I(p)) = \{ \beta(s) \mid s \in I(p) \} & \quad & 
	\gamma(X(p)) = \{ s \mid \beta(s)\in X(p) \}
\end{array} 
\]


We conclude here by showing that concrete computations are safely  approximated by the abstract semantics. 

\begin{theorem}[Soundness of the approximation]\label{teo:corr}
Let $(\cC, \alpha, \cA )$  be a description and   ${\mathbf{A}}$  be upper correct w.r.t. $\mathbf{C}$.   Given a  \utcc\ program $\cD.P$, if 
$ s \in \os P \cs $ then  $\alpha(s) \in \os P \cs^{\alpha}$.  
\end{theorem}
\begin{proof} 
Let $d_{\alpha}.s_{\alpha} = \alpha(d.s)$ and assume that $d.s \in  \os P \cs$. Then, $d.s \in \os P\cs_I$ where $I$ is the $lfp$ of $T_{\cD}$. By the continuity of $T_{\cD}$,  there exists $n$ s.t.  $I = T_{\cD}^n(I_\bot)$ (the $n$-th application of $T_{\cD}$). 
We proceed by  induction on the lexicographical order on 
the pair $n$ and the structure of $P$, where the predominant component is  $n$. We only present the interesting cases. The others can be found in \ref{app-sec-abs}.

\noindent {\underline{\bf Case $P=\absp{\vx}{c}{Q}$}}. 
Let $\sxt$ be an admissible substitution.
		We shall prove that  $s \in \os (\whenp{c}{Q})\sxt\cs$ implies $s_\alpha \in \os (\whenp{c}{Q})\sxt\cs^\alpha$. 
		The result follows from Proposition 
		 \ref{prop:forall-subs}  and from the fact that 
		 $s_\alpha \in \Forall \vx(\os \whenp{c}{Q}\cs^\alpha)$ iff $s_\alpha \in \os (\whenp{c}{Q})\sxt  \cs^\alpha $ for all $\adm{\vx}{\vt}$. The proof of the previous statement is similar to that of 
		 Proposition 
		 \ref{prop:forall-subs} and it appears in   \ref{app:proofs-aux}. 

Assume that $d\entails c\sxt$. Then,   $d.s\in \os Q\sxt\cs$ and we distinguish two cases:

\noindent(1) 
$d_\alpha \absconcentails c\sxt$.
		Since $d.s\in \os Q\sxt\cs$ then 
		$d.s \in \Exists \vx (\os Q \cs \cap \up(\dxt^\omega))$. Therefore, there exists $d'.s'$, an $\vx$-variant of $d.s$, s.t.
		$d'.s'\in \os Q \cs$ and $d'.s' \in \up(\dxt^\omega)$. 
		By (structural) inductive hypothesis, $\alpha(d'.s') \in \os Q\csa$. Furthermore, by monotonicity of $\alpha$ and Property (2) in Definition \ref{dec:corapp}, we derive
		 $\alpha(d'.s') \in \uparrow (\dxta)^\omega$. Hence
		$\alpha(d'.s') \in  (\os Q \csa \cap \up((\dxta)^\omega)$. Since $\exists \vx(d.s) = \exists \vx(d'.s')$, 
		by Property (1) in Definition \ref{dec:corapp}, we have $\existsa \vx(\alpha(d.s)) = \existsa \vx(\alpha(d'.s'))$ (i.e., $\alpha(d'.s')$ is an $\vx$-variant of $d_\alpha.s_\alpha$).
		Then, $d_\alpha.s_\alpha \in  \Exists \vx (\os Q\csa \cap \up((\dxta)^\omega))$ and we conclude  $d_\alpha.s_\alpha \in \os Q\sxt\csa$. 
		
\noindent(2) $d_\alpha \not\absconcentails c\sxt$. Hence trivially $d_\alpha.s_\alpha \in \os (\whenp{c}{Q})\sxt \csa$. 

We conclude by noticing that if $d\notentails c\sxt$ then 
$d_\alpha \not\absconcentails c\sxt$ and therefore $d_\alpha.s_\alpha \in \os (\whenp{c}{Q})\sxt \csa$.

\noindent \underline{\bf Case 	 $P\defsymbol p(\vt)$}. 
Let  $p(\vx) \defsymbol  Q$ in $\cD$ be a process definition. 
	If $d.s \in \os p(\vt) \cs$ then $d.s \in I(p(\vt))$ (recall that $I=lfp(T_{\cD})$). We know that $d.s \in \os Q\sxt \cs$ and then, 
	 $d.s \in \os Q\sxt \cs_{I'}$ where 
	 $I' = T_{\cD}^{m}(I_\bot)$ with $m < n$. 
	By induction, and continuity of $T_{\cD}^\alpha$,  we know that 
	$d_\alpha.s_\alpha \in \os Q\sxt\cs^\alpha$ and then 
	$d_\alpha.s_\alpha \in \os p(\vt)\csa$.
\end{proof}


%

\subsection{Obtaining a finite analysis}
As standard in Abstract Interpretation, it is possible to obtain an 
analysis which terminates, by imposing several alternative conditions 
(see for instance Chapter 9 in \cite{CC92}).
So, one possibility is to impose that the abstract domain is 
noetherian (also called finite ascending chain condition). Another 
possibility is to use widening operators, or to find an abstract 
domain that guarantees termination after a finite number of steps. So, our framework allows to 
use all this classical methodologies.
In the examples that we have developed 
we shall focus our attention on a special class of abstract interpretations 
obtained 
by defining 
what we call a \emph{sequence abstraction} mapping possibly infinite sequences of (abstract) 
constraints into finite ones. Actually we can define these 
abstractions 
as Galois connections.

\begin{definition}[$k$-sequence Abstraction]
   A $k$-sequence abstraction  is given by the
following pair of functions $(\alpha_{k},\gamma_{k})$, with 
$\alpha_{k}: (\cA^{\omega},\leq^\alpha) \rightarrow (\cA_{k}^{*},\leq^\alpha)$, 
and $\gamma_{k}: (\cA_{k}^{*},\leq^\alpha) \rightarrow (\cA^{\omega},\leq^\alpha)$.
As for  the function $\alpha_{k}$, we set $\alpha_{k}(s)=s'$ where $s'$ has length $k$  and $s'(i)=s(i)$ for $i\leq k$. Similarly,  $\gamma_{k}(s') = s$ where $s'(i)=s(i)$ for $i\leq k$ and $s'(i)=\true$ for $i>k$.
\end{definition}

It is easy to see that, for any $k$, $(\alpha_{k},\gamma_{k})$ 
defines 
a Galois connection between $(\cA^{\omega},\leq^\alpha)$ and $(\cA^{*}_{k},\leq^\alpha)$.
Thus it is possible to use compositions of Galois connections for 
obtaining a new abstraction \cite{CC92}.

If $\cA$ in $(\cC ,\alpha, \cA)$ leads to a  Noetherian abstract domain $\A$, 
then the abstraction obtained from the composition of $\alpha$ and any 
$\alpha_{k}$ above guarantees that the fixpoint of  the abstract semantics 
can be reached in a finite number of iterations. 
Actually the domain that we obtain in this way is given by  
sequences cut at length $k$. The number $k$ determines the length of 
the cut and hence the precision of the approximation. The bigger $k$ 
the better the approximation.



\section{Applications}\label{sec:app}
This section is devoted to show some applications of the abstract semantics developed here. We shall describe  three specific abstract domains as instances of our framework: (1) we abstract a constraint  system  representing cryptographic primitives.  Then we use the  abstract  semantics to exhibit a secrecy flaw in a security protocol modeled in \utcc.   Next, (2) we tailor two abstract domains from logic programming to perform a  groundness 
and a type analysis of a \tccp\ program. We then apply this analysis in the verification of a reactive system in \tccp. Finally, (3) we propose an abstract constraint system for the suspension analysis of \tccp\ programs.

 \subsection{Verification of Security Protocols}\label{sec:appsec}
 The ability of \utcc\ to express mobile behavior, as in Example \ref{ex:mobility},  allows for the  modeling of  security protocols. Here we describe an abstraction of a  cryptographic constraint system  in order to bound the length of the messages to be considered in a secrecy analysis. We start by recalling the constraint system in \cite{Olarte:08:SAC} whose terms represent the  messages  generated by the protocol and cryptographic  primitives are represented as functions over such terms.

\begin{definition}[Cryptographic Constraint System]\label{def:css} Let  $\Sigma$ be a signature with constant symbols in $\cP \cup \cK$, function symbols $\encp$,  $\pairp$,  $\priv$ and $\pub$ and predicates $\outp(\cdot)$ and  $\secretp(\cdot)$. Constraint  in $\cC$ are formulas 
built from predicates in $\Sigma$, conjunction ($\sqcup$) and $\exists$. 

\end{definition}

Intuitively, $\cP$ and $\cK$ represent respectively the  principal identifiers, e.g.  $A,B,\ldots$  and  keys $k,k'$. We use $\{m\}_k$ and  $(m_1,m_2)$ respectively,  for  $\encp(m,k)$ (encryption) and   $\pairp(m_1,m_2)$ (composition).  For the generation of keys, $\priv(k)$ stands for the private key associated to the value $k$ and $\pub(k)$ for its public key. 

As standardly done in the verification of security protocols, a Dolev-Yao attacker \cite{dolev-yao} is presupposed, able to  eavesdrop, disassemble, compose, encrypt and decrypt messages with available keys. The ability to eavesdrop all the messages in transit in the network is implicit in our model due to the shared store of constraints. The other abilities are  modeled by the following \utcc\ processes:
\[
\begin{array}{lll}
{Disam}() &\defsymbol & \absp{x,y}{\outp(\ (x,y) \ ) }{ \tellp{\outp{(x)} \sqcup \outp{(y)}}} \\
{Comp}()  &\defsymbol & \absp{x,y}{\outp(x) \sqcup \outp(y)}{  \tellp{\outp{( \ (x,y) \  )}}} \\
{Enc}()&\defsymbol &     \absp{x,y}{\outp(x) \sqcup \outp(y)}{  \tellp{\outp{(\{x\}_{\pub(y)})}}}  \\
{Dec}()&\defsymbol &    \absp{x,y}{\outp(\priv(y)) \sqcup \outp(\{ x\}_{\pub(y)})}{  \tellp{\outp{(x)}}}   \\
{Pers}() & \defsymbol  &    \absp{x}{\outp(x)}\nextp{\tellp{\outp(x)}} \\
{Spy}() & \defsymbol  &  {Disam}() \parallel {Comp}() \parallel {Enc}() 
\parallel  {Dec}() \parallel {Pers}() \parallel \nextp{{Spy}()}
\end{array}
\]

 Since the final store is not automatically transferred to the next time-unit, the process $Pers$ above models the ability to remember all messages  posted so far.

It is easy to see that the process ${Spy}()$ in a store  $\outp(m)$ may add messages of unbounded length. Take for example the process  ${Comp}()$ that will add the constraints $\outp(m)$, $\outp( ( m,(m,m)))$, $\outp(((m,m),m))$ and so on. 

To deal with the inherent state explosion problem in the model of the attacker, 
symbolic (compact) representations of the behavior of the attacker 
have been proposed, for instance in 
\cite{boreale01symbolic,compsym-fiore,Olarte:08:SAC,BBD10}. Here we  follow the approach of restricting  the number of states to be considered in the verification of the protocol, as for instance  in \cite{DBLP:journals/corr/abs-1105-5282,SongBP01,DBLP:journals/ijisec/ArmandoC08}. Roughly, we  shall cut the messages generated   of length greater  than a given $\kappa$, thus allowing us to model a bounded version of the attacker. 

Before defining the abstraction, we notice that the constraint system  we are considering includes existentially quantified syntactic  equations. For this kind of equations it is necessary to refer to a solved form of them  in order to have a uniform way to compute an approximation of the constraint system. We then consider constraints of the shape $\exists \vy (x_1 = t_1(\vy) \sqcup ... \sqcup x_n=t_n(\vy))$ where $\vx= x_1,...x_n$ are pairwise distinct   and $\vx \cap \vy = \emptyset$. Here,  
$t(\vec{y})$  refers to a term where $\fv(t(\vy)) \subseteq \vy$. Given a constraint, its normal form  can be obtained by applying the algorithm proposed in \cite{Mah88b} where: quantifiers are moved to the outermost position and equations of the form $f(t_1,...,t_n) =f(t_1',...,t_n')$
are replaced by  $t_1=t_1' \sqcup ... \sqcup t_n = t_n'$; equations such as   $x=x$ are deleted; equation of the form $t=x$ are replaced by $x=t$; and given $x=t$, if $x$ does not occur in $t$, $x$ is replaced by $t$ in $t'$ in all equation of the form $x'=t'$.  For instance, the solved form of $ \exists z, y ( x=f(y) \sqcup y=g(z) )$ is the constraint $\exists z (x=f(g(z) ))$.

\begin{definition}[Abstract secure constraint system]\label{def:abs-sec-cs}
Let $\cM$ be the set of terms (messages) generated from the signature $\Sigma$ in Definition \ref{def:css}. Let   ${\it lg}: \cM \to \mathbb{N}$ be defined as 
${\it lg}(m) = 0$ if $m\in \cP\cup\cK\cup Var$; ${\it lg}(\{m_1\}_{m_2}) = {\it lg} (\  (m_1,m_2)\ ) = 1 + {\it lg}(m_1) + {\it lg}(m_2)$. 
Let $cut_\kappa(m)= m $ if  ${\it lg}(m) \leq \kappa$. Otherwise,  $cut_\kappa(m)=m_{\top}$  
where $m_{\top}\notin \cM$  represents  all the messages  whose length is greater than $\kappa$. 
We define $\alpha(c)$ as $\alpha_{\kappa}(NF(c))$ where
\[
\begin{array}{llll l llll}
\alpha_{\kappa}(c(m)) &=&c(cut_\kappa(m))
& \quad &
\alpha_{\kappa}(d_{xt}) &=& d_{xt'} \mbox{ where } t' = cut_\kappa(t)
\\
\alpha_{\kappa}(c \sqcup c') &=& \alpha_{\kappa}(c) \sqcup \alpha_{\kappa}(c') 
& \quad &
\alpha_{\kappa}(\exists \vx c ) &=& \exists \vx \alpha_{\kappa}(c) \\
\end{array}
\]
and $NF(c)$ is a solved form  of the constraint $c$.
We  omit the superscript $\alpha$ in the abstract operators $\sqcup^{\alpha}$, $\exists^{\alpha}$ and $d_{\vec{x}\vec{t}}^{\alpha}$ to simplify the notation. 
\end{definition}

We note that the previous abstraction reminds of 
the $depth\mbox{-}\kappa$ abstractions typically done 
in the analysis of logic programs (see e.g., \cite{ST84}). 

We shall illustrate the use of the abstract constraint system above by performing a secrecy analysis on the
Needham-Schr\"oder  (NS) protocol \cite{lowe95attack}. This protocol aims at distributing two \emph{nonces} in a secure way. 
Figure \ref{fig:ns}(a) shows the steps of NS where $m$ and $n$ represent  the nonces generated, respectively,  by the principals $A$ and $B$.
The protocol initiates when 
$A$ sends to $B$ a new 
nonce $m$ together with her own agent name $A$, both encrypted with $B$'s public key. When $B$ receives the message,
he decrypts it with his secret private key. Once decrypted, 
$B$ prepares an encrypted message for $A$ that contains 
a new nonce $n$ together with the nonce  $m$ and his name $B$.   $A$ then 
recovers the clear text using her private key. $A$ convinces 
herself that this message really comes from B by checking 
whether she got back the same nonce sent out in the first 
message. If that is the case, she acknowledges B by returning his nonce. 
$B$ does a similar test. 
\begin{figure}
\resizebox{.8\textwidth}{!}{
\subfloat[]{
$
\begin{array}{llll}
{\mathsf{M_1}   \ } & A \to B & :  & \{(m,A)\}_{\pub(B)}\\
{\mathsf{M_2}   \ } & B \to A & :  & \{(m,n,B)\}_{\pub(A)}\\
{\mathsf{M_3}   \ } & A \to B & :  & \{n\}_{\pub(B)}\\
\end{array}
$
}
\quad
\subfloat[]{
$
\begin{array}{llll}
{\mathsf{M_1}   \ }& A\to C & : & \{(m,A)\}_{\pub{(C)}}\\
{\mathsf{M_1'}   \ }& C\to B & : & \{(m,A)\}_{\pub{(B)}}\\
{\mathsf{M_2}   \ }& B\to A & : & \{(m,n,B)\}_{\pub{(A)}}\\ 
{\mathsf{M_3}   \ }& A\to C & : & \{n\}_{\pub{(C)}}\\
\end{array}
$
}
}
\caption {Steps of the Needham-Schroeder Protocol \label{fig:ns}}
\end{figure}

Assume the execution of the protocol in Figure \ref{fig:ns}(b). Here $C$ is an intruder, i.e. a malicious agent playing the role of a principal in the protocol. As it was shown in \cite{lowe95attack}, this execution leads to a secrecy flaw where the attacker $C$ can reveal  $n$ which is meant to be known only by $A$ and $B$. 
In this execution, the attacker replies to $B$ the message sent by $A$ and $B$ believes that he is establishing a session key with $A$. Since the attacker knows the private key $\priv(C)$, she can decrypt the message $\{n\}_{\pub{(C)}}$ and $n$ is no longer a secret between $B$ and $A$ as intended.

We model the behavior of the principals of the NS protocol with the process definitions in Figure \ref{ns:utcc:procs}. 
\begin{figure}
\resizebox{\textwidth}{!}{
$
\begin{array}{lll}
{Init}(i,r) &\defsymbol &  \localp{m}{}   \tellp{\outp(\{(m,i)\}_{pub(r))}} \parallel \\
&& \qquad\qquad\qquad  \nextp \absp{x}{\outp(\{(m,x,r)\}_{\pub(i)})}{\tellp{\outp(\{x\}_{\pub(r))}}}\\
&&  \parallel \nextp {Init}(i,r) \\
{Resp(r)} &\defsymbol &  \absp{x,u}{\outp(\{(x,u)\}_{\pub(r)})}\nextp\\
&& \tabp\tabp \ \localp{n}{} {(Secrete}(n)  \parallel \tellp{\outp(\{x,n,r\}_{\pub{(u)})}})\\ 
 && \parallel \nextp {Resp(r)}\\
{Secrete}(x) &\defsymbol & \tellp{\secretp(x)} \parallel \nextp{{Secrete}(x)} \\
{SpKn}() &\defsymbol &  \parallel_{A\in \cP} \ \tellp{\outp(A)\sqcup \outp(\pub(A))}\\
&&  \parallel_{A\in Bad} \ \tellp{\outp(\priv(A))}\\
&&  \parallel \nextp{SpKn}()
\end{array}
$
}
\caption{\utcc\ model of the Needham-Schr\"oder Protocol \label{ns:utcc:procs}}
\end{figure}
Nonce generation is modeled by $\mathbf{local}$ constructs and the process $\tellp{\outp(m)}$ models the broadcast of the  message $m$. Inputs (message reception) are modeled by 
 $\mathbf{abs}$ processes as in Example \ref{ex:ex-mob-sos}. 
 In ${Resp}$, we use the process ${Secrete}(n)$ to state that the nonce  $n$ cannot be revealed. Finally, the process ${SpKn}$ corresponds to the initial knowledge of the attacker:   the names of the principals, their public keys and the 
 leaked keys in the set $Bad$ (e.g., the private key of $C$ in the configuration of Figure \ref{fig:ns} (b)).

Consider the following process:
\begin{equation}\label{eq:ns-proc}
{NS} :-\ \   {Spy} \parallel {SpKn} \parallel {Init}(A,C) \parallel {Resp}(B)
\end{equation}

By using the composition of  $\alpha_3$ (as in Definition \ref{def:abs-sec-cs}) and the sequence abstraction $2$-$sequence$, we  obtain the abstract semantics of ${NS}$ as showed in  Figure \ref{fig:semantics-ns}. This allows us  to exhibit the  secrecy flaw of the NS protocol pointed out in \cite{lowe95attack}: 
Let $s=c_1.c_2$ s.t. $s  \in \os \mathbf{NS}\cs^{\alpha}$. Then, there exist a $m_1$-$n_1$-variant $s'=c_1'.c_2'$ of $s$ s.t.
 \[
\begin{array}{l}
c_1'  \entails   \outp(\{m_1,A\}_{\pub(C)}) \sqcup \outp(\priv(C)) \sqcup \outp(\{m_1,A\}_{\pub(B)})\\
c_2'  \entails  \outp(\{m_1,n_1,A\}_{\pub(A)}) \sqcup \outp(\{n_1\}_{\pub(C)}) \sqcup \outp(\secretp(n_1)) \sqcup \outp(\outp(n_1))
\end{array}
\]
This means that the nonce $n_1$ appears as plain text in the network and it is no longer a secret between $A$ and $B$ as intended.

\begin{figure}
\resizebox{\textwidth}{!}{
$
 \begin{array}{lcl}
	\os {Init}(A,C)\cs^{\alpha} &= & \Exists\  m_1 \ \Exists\  m_2 \ \ (\{c_1.c_2 \ | \ c_1\entailsa \outp(\{m_1,A\}_{\pub{(C)}}) ,\ c_2\entailsa \outp(\{A,m_2\}_{\pub{(C)}}) \} \cap \\
	& & \tabp\tabp\tabp\cA. \Forall  x (	\{ c_2 \mid \mbox{if } c_2 \absconcentails \outp(\{m_1,x,C\}_{\pub(A)}) \mbox{ then } c_2 \entailsa \outp(\{x\}_{\pub(C)})
	\}))\\\\
	\os {resp}(B)\cs^{\alpha} &= &  \Forall   x,u   (
	\ \Exists \ n_1\ \{c_1.c_2 \ |\  \mbox{if } c_1 \absconcentails \outp(\{x,u\}_{\pub(B)}) \mbox{ then }\\
	& &  \tabp\tabp\tabp\tabp\tabp\tabp\ \ \  
	 c_2 \entailsa \secretp(n_1) \sqcup \outp(\{x,n_1,B\}_{\pub(u)}) \}
	) \\\\
\os {Spy}\cs^{\alpha} & = &  \Forall \ x\  ( \{c_1.c_2 \ | \ \mbox{ if } c_1 \absconcentails \outp(x) \mbox{ then } c_2 \entailsa \outp(x)\}) \cap S.S
 \mbox{ where }\\\\

S & = &  \Forall  x,y   (  \{c \ |\  \mbox{ if } c\absconcentails \outp(x) \sqcup \outp(y) \mbox { then } c \entailsa  \outp(\{x,y\}) \sqcup \outp(\{x\}_{\pub(y)}) \} \cap \\
& & \ \ \ \tabp \tabp 
\{c \ |\  \mbox{ if } c\absconcentails \outp(\{x,y\}) \mbox{ then } c \entailsa \outp(x) \sqcup \outp(y)\} \cap \\
& & \ \ \  \tabp \tabp 
 \{c \ | \ \mbox{ if } c \absconcentails \outp(\{x\}_{\pub(y)}) \sqcup \outp(\priv(y)) \mbox{ then } c \entailsa \outp(x)\}) \\\\

\os {SpKn}\cs^{\alpha} & = &  \{
c_1.c_2 \ | \ c_i \entailsa \outp(\pub(A))\sqcup \outp(\pub(B))\sqcup \outp(\pub(C))\} \cap \\
& &  \{ c_1.c_2 \ | \ c_i \entailsa \outp(A)\sqcup \outp(B) \sqcup \outp(C)\sqcup \outp(\priv(C))
\} \\\\
\os {NS}\cs^{\alpha} & = & 
\ \ \os {Spy}\cs^{\alpha}  \cap \os {SpKn}\cs^{\alpha}  \cap \os \mathbf{init}(A,C)\cs^{\alpha}  \cap \os \mathbf{resp}(B)\cs^{\alpha}\\
  \end{array}
$
}
\caption{Abstract semantics of the process ${NS}$ in Equation \ref{eq:ns-proc}\label{fig:semantics-ns}}
\end{figure}

\subsection{Groundness Analysis}\label{sec:ground}
In logic programming one useful  analysis is groundness. It aims at determining if a variable will always be bound to a ground term. This information can be used, e.g., for optimization in the compiler
 or as base for other data flow analyses such as independence analysis, suspension analysis, etc. Here we present a  groundness analysis for a \tccp\ program. To this end, we shall use as concrete domain the Herbrand Constraint System and the following running example. 
   \begin{figure}
\resizebox{\textwidth}{!}{
$
\begin{array}{lll}
{\it gen}_a(x) & :- & \localp{x'}{}( {\it assign}(x,[a|x']) \parallel \\
& & \ \tabp\tabp\tabp \whenp{\it go_a=[]}{\nextp{{\it gen}_a(x')}} \parallel  \whenp{\it stop_a=[]}{{\it assign}(x',[]))} \\\\
 
 {\it assign}(x,y) & :- & \tellp{x = y} \parallel \nextp{{\it assign}(x,y)}  \\\\

 {\it append}(x,y,z) & :- & \whenp{x=[ ]}{{\it assign}(y,z)} \parallel \\
 & &  \whenp{\exists_{x',x''}(x=[x' \ | x''])} \\
 & & \ \ \ \ \  \localp{x',x'',z'}{({\it assign}(x, [x'  |  x''])} \parallel  {{\it assign}(z, [x'  |  z'])} \parallel   \nextp{{\it append}(x'',y,z')} )
\end{array}
$
}
\caption{Appending streams (Example \ref{ex:append}). The process definition $gen_b$ is similar to  $gen_a$ but replacing the constant $a$ with $b$. \label{fig-ex-append}}
\end{figure}
 \begin{example}[Append]\label{ex:append}
Assume the process definitions  in Figure \ref{fig-ex-append}.  The process ${\it gen}_a(x)$ adds an ``$a$'' to the stream  $x$  when the environment provides $go_a=[]$ as input. Under input $stop_a=[]$, ${\it gen}_a(x)$ terminates the stream binding its tail to the empty list. The process  ${\it gen}_b$ can be explained similarly.  The process ${\it assign}(x,y)$ 
persistently equates $x$ and $y$. Finally,  ${\it append}(x,y,z)$ binds $z$ to the concatenation of $x$ and $y$. 
\end{example}

We shall use $Pos$ \cite{armstrong98two} as abstract domain for the groundness analysis.  In $Pos$, positive propositional formulas 
  represent groundness dependencies among variables. 
  For instance, $\alpha_G(x=[a | b]) = x$ meaning that $x$ is a ground variable and $\alpha_G(x=[y | z]) = x \leftrightarrow (y \wedge z)$   meaning that    $x$ is ground if and only if both $y$ and $z$ are ground. 
Elements in this domain are ordered by logical implication, e.g., 
$x \sqcup (x \leftrightarrow (y \wedge z) )\entails_{\alpha_G} y$.

\begin{observation}[Precision of Pos with respect to Synchronization]
Notice that $Pos$ does not distinguish between the empty list and a list of ground terms:  $d_\kappa = \alpha_G(x=[]) = \alpha_G(x=[a]) = x$ and then,   $d_\kappa \not\absconcentails x=[]$ (see Definition \ref{def:absconcentails}).  This affects the precision of the analysis. For instance, let $
P=\tellp{x=[]}$ and $Q= \whenp{x=[]}{\tellp{y=[]}}
$. 
One would expect that the groundness analysis of $P \parallel Q$ determines that $x$ and $y$ are ground variables. Nevertheless, it is easy to see that  $x.true^\omega\in \os P\cs^{\alpha_G}$ and then, the information added by $\tellp{y=[]}$ is lost.
\end{observation}

We  improve the accuracy of the analysis by using the abstract domain defined in \cite{CodishD94} to derive information about type dependencies on terms. The abstraction is defined as follows:
\[
\alpha_{T}(x=t) = \left\{
\begin{array}{lll}
\listp(x,x_s) & \mbox{if} & t=[y \ |\  x_s] \mbox{ for some $y$}\\
\nilp(x) & \mbox{if} & t=[]
\end{array}
\right.
\]
Informally, $list(x,x_s)$ means $x$ is a list iff $x_s$ is a list and  $nil(x)$  means $x$ is the empty list. If $x$ is a list we write $list(x)$ and $nil(x) \entails^{\alpha_T} list(x)$. Elements in the domain are ordered by logical implication.

The following constraint systems result from the reduced product  \cite{CC92} of the previous abstract domains, thus allowing us to capture groundness and type dependency information. 

\begin{definition}[Groundness-type Constraint System] \label{def:gt-domain}
Let 
$
{\bf A_{GT}}=  \l \cA,\leq^{\alpha_{GT}}\, \sqcup^{\alpha_{GT}},\true^{\alpha_{GT}}, \false^{\alpha_{GT}}, {\it Var}, \exists^{\alpha_{GT}}, d^{\alpha_{GT}} \r
$. 
Given $c\in \cC$, $\alpha_{GT}(c) = \l \alpha_G(c) , \alpha_T(c) \r$. 
The operations $\sqcup^{\alpha_{GT}}$ and  $\exists^{\alpha_{GT}}$ 
correspond to logical conjunction and existential quantification on the components of the tuple and   $\dxt^{\alpha_{GT}}$ is defined as $\l \alpha_G(\vx=\vt), \alpha_T(\vx=\vt)\r$. 
Finally, $\l c_\kappa,d_\kappa\r \leq^{\alpha_{GT}}\l c_\kappa',d_\kappa'\r $ 
iff 
$c_\kappa' \entails_{\alpha_G} c_\kappa$ and $d_\kappa' \entails_{\alpha_T} d_\kappa$.
\end{definition}

Consider the Example \ref{ex:append} and the 
abstraction $\alpha$ resulting from the composition of   $\alpha_{GT}$ above and $sequence_{\kappa}$. Note that the program  makes use of guards of the form $\exists x',x'' (x=[x'|x''])$ and $x=[]$. Note also that    $list(x,x') \absconcentails \exists x',x'' (x=[x'|x'']) $ and  $nil(x) \absconcentails x=[]$. Roughly speaking, this guarantees that the chosen domain is accurate w.r.t. the ask processes in the program.   

The semantics of the process $P=gen_a(x) \parallel  gen_b(y) \parallel append(x,y,z)$ is depicted in Figure \ref{fig:sem:append}. Assume that $s= c_1.c_2...c_\kappa \in \os P\csa$. Let
 $n \leq \kappa$ and assume that for  $i<n$,  $c_i \absconcentails go_a=[]$ and $c_n \absconcentails stop_a=[]$. Since $s\in \os P\csa$, we know that $s\in \os gen_a(x) \csa$ and then, we can verify that $c_n \entailsa \l x,list(x)\r$. Similarly, take $m \leq \kappa$ and assume that for  $j<m$,  $c_j \absconcentails go_b=[]$ and $c_m \absconcentails stop_b=[]$. We can verify that 
 $c_m \entailsa \l y,list(y)\r$. Finally, since $s\in append(x,y,z)$, we can show that $c_{max(n,m)} \entailsa \l z,list(z)\r$.  In words, the process $P$ binds $x$, $y$ and $z$ to  ground lists whenever the environment provides as input a series of constraints $go_a=[]$ (resp. $go_b=[]$) followed by an input $stop_a=[]$ (resp.  $stop_b=[]$). 
 
 \begin{figure}
\resizebox{.8\textwidth}{!}{
$
\begin{array}{lll}
\multicolumn{3}{l}{
 \os gen_a(x) \parallel  gen_b(y) \parallel append(x,y,z) \csa =  \Exists x_1 (GA_1) \cap \Exists y_1 (GB_1)\cap A_1}
 \mbox{ where } \\\\
GA_1 &=& \up{}\l x\leftrightarrow x_1, list(x,x_1)\r.\cA \ \cap \\
  & & \{c.s \ | \mbox{ if } c \absconcentails\ go_a=[]  \mbox{ then } s \in \Exists x_2 (GA_2)\} \cap \\
  & & \{c.s \ | \mbox{ if } c \absconcentails\ stop_a=[]  \mbox{ then } \l x_1,nil(x_1)\r^\omega \leqa c.s \}  \\
\cdots \\
GA_\kappa &=& \up{}\l x_{\kappa-1}\leftrightarrow x_\kappa, list(x_{\kappa-1},x_{\kappa})\r.\epsilon \ \cap \\
  & & \{c.\epsilon \ | \mbox{ if } c \absconcentails\ stop_a=[]  \mbox{ then } c \entailsa \l x_\kappa,nil(x_\kappa)\r \}  \\\\

A_1& = & \{c.s \ | \ \mbox{if } c \absconcentails x=[]  \mbox{ then }  (d_{yz}^\alpha)^\omega \leqa c.s  \} \ \cap \\
& & \{c.s \ | \ \mbox{if } c \absconcentails \exists x', x_2(x=[x' | x_2])  \mbox{ then }  \\
& & \tabp\tabp\  c.s \in \Exists x' \Exists x_2\Exists z_2   
(\up(\l x \leftrightarrow  x_2, list(x,x_2)\r^\omega)  \ \cap \\
& & \tabp\tabp\tabp\tabp\tabp\tabp\tabp\tabp\ \ \   
\up(\l z \leftrightarrow  z_2, list(z,z_2)\r^\omega) \cap \cA.A_2 )\}\\
\cdots \\
A_\kappa& = & \{c.\epsilon \ | \ \mbox{if } c \absconcentails x_\kappa=[]  \mbox{ then } d_{y_{\kappa}z_{\kappa}}^\alpha  \leqa c  \} \ \cap \\
& & \{c.\epsilon \ | \ \mbox{if } c \absconcentails \exists x', x_{\kappa'}(x=[x' | x_{\kappa'}])  \mbox{ then }  \\
& & \tabp\tabp\  c.\epsilon \in \Exists x' \Exists x_{\kappa'}\Exists z_{\kappa'}
(\up(\l x_\kappa \leftrightarrow  x_{\kappa'}, list(x_\kappa,x_{\kappa'})\r).\epsilon \ \cap \\
& & \tabp\tabp\tabp\tabp\tabp\tabp\tabp\tabp\ \ \   
\ \ \up(\l z_\kappa \leftrightarrow  z_{\kappa'}, list(z_\kappa,z_{\kappa'})\r).\epsilon)\}
\end{array}
$}
\caption{Abstract semantics of the process $P=gen_a(x) \parallel  gen_b(y) \parallel append(x,y,z)$. Definitions of $gen_a(x), gen_b(y)$ and  $append(x,y,z)$ are given in Example \ref{ex:append}. Sets $GB_1,..,GB_{\kappa}$
are similar to $GA_1,..,GA_{\kappa}$ and omitted here. 
\label{fig:sem:append}
}
\end{figure}

\subsubsection{Reactive Systems} 
Synchronous data flow languages   \cite{BeGo92} such as Esterel and  Lustre  can be encoded as \tccp\ processes \cite{tcc-lics94,Tini99}. This  makes \tccp\ an expressive declarative framework for the modeling and verification of reactive systems. 
Take for instance the program in Figure \ref{fig:micro},
taken and slightly modified from \cite{FalaschiV06}, 
 that models a control system for a microwave checking  that the door must
be closed when it is turned on. Otherwise, it must emit an error signal. In this model, $\texttt{on}$, $\texttt{off}$, $\texttt{closed}$ and $\texttt{open}$ represent the constraints $on=[], {\it off}=[],close=[]$ and $open=[]$ and the symbols $yes$, $no$, $stop$ denote constant symbols. 

The analyses developed here  can provide additional reasoning techniques in \tccp\ for the verification of  such systems. For instance, by using the groundness analysis   in the previous section, we can show that if   $c_1.c_2....c_\kappa \in \os micCtrl(Error,Button)\cs^\alpha$ and there exists $1\leq i \leq \kappa$ s.t. $c_i \absconcentails (open=[] \sqcup on=[])$, then, it must be the case that  $c_1 \entails^{\alpha} \l Error , \listp(Error)\r$, i.e., $Error$ is  a ground variable. This means, that the system correctly binds the list $Error$ to a ground term whenever the system reaches an inconsistent state.

\begin{figure}[t]
$
\begin{array}{ll}
 {\it micCtrl}(Error, Signal)  \defsymbol \\
\tabp  \localp{Error', Signal',er,sl}{(}\\
\tabp\tabp  \bangp \tellp{Error = [er \ | \ Error'] \sqcup  Signal = [sl \ |
\ Signal']} \\
\tabp\tabp  \parallel \whenp{\texttt{on} \sqcup \texttt{open}}
{\bangp\tellp{er= yes \sqcup Error' = [] \sqcup sl = {stop } }} \\
\tabp\tabp  \parallel \whenp{\texttt{off} }{(\bangp\tellp{er=no} \parallel
\nextp{\it micCtrl(Error',Signal')}}) \\
\tabp\tabp  \parallel \whenp{\texttt{closed}}{(\bangp\tellp{er=no} \parallel
\nextp{\it micCtrl(Error',Signal')}}))
\end{array}
$
\caption{Model for a microwave controller (see Notation  \ref{rem:bang} for the definition of $\bangp$).  \label{fig:micro} }
\end{figure}

\begin{observation}[Synchronization constraints]
 In several  applications of \tccp\ and \utcc\  the environment interact with the system by adding as input some constraints that only appear in the guard of ask processes as   $\texttt{on}, \texttt{off}, \texttt{open}, \texttt{close} $ in Figure \ref{fig:micro} and 
$go_a$, $stop_a$  in the Figure \ref{fig-ex-append}. These constraints can be thought of as ``synchronization constraints''  \cite{DBLP:journals/iandc/FagesRS01}. Furthermore, since these constraints are inputs from the environment, they are not expected to be produced by the program, i.e., they do not appear in the scope of a tell process. In these situations, in order to improve the accuracy of the analyses, one can orthogonally add  those constraints in the abstract domain. This can be done,  for instance,  with a reduced product as we did in Definition \ref{def:gt-domain} to  give a finer approximation of the inputs $go_a$ and $stop_a$ by adding type dependency information.
\end{observation}



\subsection{Suspension Analysis}
In a concurrent setting it is important to know whether  a given system 
reaches a state where no further evolution is possible. 
Reaching a deadlocked situation is something to be avoided. There are many studies on this problem and 
several works developing analyses in (logic) concurrent languages
(e.g. \cite{CFM94,CFMW97}). However, we are not aware of studies 
available for \ccp\ and its temporal extensions. A  suspended state in the context of 
\ccp\  may happen when the guard of the ask processes  are not carefully chosen and then, none of them can be entailed. In this section we develop an analysis  that 
aims at determining the  constraints that a program needs as input from 
the environment to proceed. This can be used to derive information about 
the suspension of the system. We start by   extending  the concrete semantics to a 
collecting semantics that keeps information about the suspension of processes. 
For this, we define  the following constraint system.

\begin{definition}[Suspension  Constraint System]\label{def:sync-cs}
Let  $\cS  = \{\susp,\nsusp\}$ s.t. $\susp \leq \nsusp$. Given a constraint system ${\bf C}= \l \cC,\leq,\sqcup,\true,\false,{\it Var}, \exists, d\r$, the suspension-constraint system $S({\bf C})$ is defined as 
\[
{\bf S} = \l \cC \times \cS, \leq^s, \sqcup^s, \l \true, \susp \r,  
\l \false,\nsusp \r, Var, \exists^s, d^s
\r
\]
where $\leq^s, \sqcup^s$ are pointwise defined, $\exists^s_{\vec{x}}(\l c, c' \r) = \l \exists_x c, c' \r$ and $d^s_{\vx\vt} = \l \dxt , \susp \r$.
Given a constraint $c\in \cC$, we shall use 
 $\widehat{c}$ to denote the constraint  $\l c, \susp \r$. 
\end{definition}

Let us illustrate  how   $S({\bf C})$ allows us to derive information about suspension. 

\begin{example}[Collecting Semantics]\label{ex:col-sem}
Let  $\cC= \{\true,a,b,c,d,\false\}$ be a complete lattice where $b\entails a$ and $d\entails c$, 
$P=\whenp{a}{\tellp{b}}$ and 
$Q=\whenp{c}{\tellp{d}}$.  We know that 
$
\os P\cs = \{\true, b, c , d , \false \}.\cC^\omega
$ (note  that $P$  does not suspend on 
$b$ and $\false$). Let $\widehat{P}$ and $\widehat{Q}$ be defined over  $S({\bf C})$ as:
\[
\begin{array}{lll}
\widehat{P} = \whenp{\widehat{a}}{(\tellp{\widehat{b}} \parallel \tellp{\l a, \nsusp \r})} & \!\! &
\widehat{Q} = \whenp{\widehat{c}}{(\tellp{\widehat{d}} \parallel \tellp{\l c, \nsusp \r})} \\
\end{array}
\]
We then have:
\[
\begin{array}{rll}
\os \widehat{P}\cs &=& \{\l \true,\up\susp\r, \l b,\nsusp\r, \l c ,\up\susp\r , \l d ,\up\susp\r, \l \false,\nsusp\r \}.(\cC\times \cS)^\omega\\
\os \widehat{Q}\cs &=& \{\l \true,\up\susp\r, \l a,\up\susp\r, \l b ,\up\susp\r , \l d ,\nsusp\r, \l \false,\nsusp\r \}.(\cC\times \cS)^\omega\\
\os \widehat{P}\parallel \widehat{Q}\cs & = & \{\l \true,\up\susp\r, \l b,\nsusp\r, \l d ,\nsusp\r, \l \false,\nsusp\r \}.(\cC\times \cS)^\omega
\end{array}
\]


\end{example}



where $\l c, \up{\susp}\r$ is a shorthand for the couple of tuples  $\l c, \susp\r, \l c, \nsusp\r$. 
The process $P$ suspends on input $c$ (since $c\not\entails a$) while $Q$ under input $c$ outputs $d$ and it does not suspend. Notice that the system  $P\parallel Q$ does not block on input $b,d$ or $\false$ and it does on input $\true$. Notice also  that $\l c,\susp \r.s \not\in \os \widehat{P} \parallel \widehat{Q} \cs$. This means that in a store $c$, 
 at least one the ask processes in  $\widehat{P} \parallel \widehat{Q}$ is able to proceed. 
The key idea is that the process $\tellp{\l c,\nsusp \r}$ in $\widehat{Q}$ ensures that if $\l e, e'\r \in \os \widehat{Q}\cs$ and  $e\entails c$, then it must be the case that  $e'=\nsusp$. This corresponds to the intuition that if an ask process can evolve on a store $c$, it can evolve under any store greater than $c$ (Lemma \ref{lemma:redi-properties}). 

Next we define a program transformation that allows us to scatter suspension information when we want to verify that none of the ask processes suspend. 

\begin{example} Let $P$ and $Q$ be as in Example  \ref{ex:col-sem}. Let also  $\widehat{P}  =  \whenp{\widehat{a}}{(\tellp{\widehat{b}} )}$, $\widehat{Q} = \whenp{\widehat{c}}{(\tellp{\widehat{d}} )} $ and
$
\widehat{R} = \widehat{P} \parallel \widehat{Q} \parallel \whenp{\widehat{a} \sqcup \widehat{c}}{(\tellp{a\sqcup c, \nsusp} )}
$. Therefore, 
\[
\begin{array}{rll}
\os \widehat{P}\cs &=& \{\l \true,\up\susp\r, \l b,\up\susp\r, \l c ,\up\susp\r , \l d ,\up\susp\r, \l \false,\up\susp\r \}.(\cC\times \cS)^\omega\\
\os \widehat{Q}\cs &=& \{\l \true,\up\susp\r, \l a,\up\susp\r, \l b ,\up\susp\r , \l d ,\up\susp\r, \l \false,\up\susp\r \}.(\cC\times \cS)^\omega\\
\os \widehat{R}\cs & = & \{\l \true,\up\susp\r, \l b,\up\susp\r, \l d ,\up\susp\r, \l \false,\nsusp\r \}.(\cC\times \cS)^\omega
\end{array}
\]
Hence we can conclude that only under input $\false$  
 neither $P$ nor $Q$  suspend. 
\end{example}

The previous program transformation can be arbitrarily applied to subterms of the form $P=\prod\limits_{i\in I} \whenp{c_i}{P_i}$. Similarly, for verification purposes, a subterm of the form 
$P = \absp{\vec{x_1}}{c_1}{P_1} \parallel ... \parallel \absp{\vec{x_n}}{c_n}{P_n} 
$ can be replaced by 
\[
P' = \widehat{P} \parallel \whenp{(\exists{\vec{x_1}} \widehat{c_1} \sqcup...\sqcup \exists {\vec{x_n}} \widehat{c_n})}\tellp{\l c_1 \sqcup ... \sqcup c_n, \nsusp\r }
\]

We conclude with an example showing  how an abstraction of the previous collecting semantics allows us to analyze a protocol programmed in \utcc.  For this we shall use the abstraction in Definition \ref{def:abs-sec-cs} to cut the terms up to a given length. 
\begin{example}\label{ex:sus-pro}
Assume a protocol where agent $A$ 
has to send a message to $B$ through a proxy server $S$. This situation can be modeled as follows:
\[
\begin{array}{lll}
	A(x,y) &\!\!\defsymbol\!\!& \localp{m}(\tellp{\outp(\{x,y,m\}_{pub(srv)})}) \\
	S &\!\!\defsymbol\!\!& \absp{x,y,m}{\outp(\{x,y,m\}_{pub(srv)})}{}{\tellp{\outp(\{x,m\}_{pub(y)})}} \!\parallel\! \nextp{S()}\\
	B(y) &\!\!\defsymbol\!\!& \absp{x,m}{\outp(\{x,m\}_{pub(y)})}{B_c} \\
	Protocol   &\!\!\defsymbol\!\!& A(x,y) \parallel S() \parallel B(y)
\end{array}
\]

where  $B_c=\skipp$ is the continuation of the protocol that we left unspecified.

This code is correct if  the message can flow  from $A$ to $B$ without any input from the environment.   
This holds if the ask process in $B(y)$ does not block. We shall then analyze the program above by replacing all   $c$ with  $\widehat{c}$ and   $B(y)$ with 
\[
B'(y) \defsymbol \absp{x,m}{\outpw(\{x,m\}_{pub(y)})}{(\tellp{\l \outp(\{x,m\}_{pub(y)}), \nsusp \r})} 
\]

Let $\alpha_\kappa$ be as in  Definition  \ref{def:abs-sec-cs}. We choose as abstract domain $\cA=S(\alpha_\kappa(\cC))$ and  we  consider sequences of length one. 
In  Figure \ref{fig:sus-sems} we show the abstract semantics.  We notice that $\langle c, \nsusp \rangle$ where $c=
\exists m (\outp(\{x,y,m\}_{pub(srv)}) \sqcup \outp(\{x,m\}_{pub(y)}) )
$ is in the semantics $\os Protocol\csa$ and $\l c,\susp\r \notin \os Protocol \csa$.  We then conclude that the protocol is able to correctly deliver the message to $B$. 

Assume now that the code for the server is (wrongly) written as 
\[
S' \defsymbol \absp{x,y,m}{\outp(\{x,y,m\}_{pub(srv)})}{}{\tellp{\outp(\{x,m\}_{pub(x)})}} \parallel \nextp{S'()}
\]
where we changed $\tellp{\outp(\{x,m\}_{pub(y)})}$ to $\tellp{\outp(\{x,m\}_{pub(x)})}$. 
We can verify that  $\l c,\susp \r \in \os Protocol'\csa$  where 
$c= \exists m (\outp(\{x,y,m\}_{pub(srv)}) \sqcup \outp(\{x,m\}_{pub(x)}) )$. This can warn the programmer that there is a mistake in the code. 
\end{example}
\begin{figure}
$
\begin{array}{rll}
\os Protocol \csa & = & A.\epsilon  \cap S.\epsilon \cap B.\epsilon \mbox { where } \\
A &= & \Exists m(\up(\outpw(\{x,y,m\}_{pub(srv)}))) \\
S &= & \Forall x,y, m (  \{ \l d,c\r \ | \ \mbox{ if } \l d,c\r \absconcentails \outpw(\{x,y,m\}_{pub(srv)}) \\
 & &\tabp\tabp\tabp\tabp\tabp\tabp \tabp
\mbox{then } \l d,c\r \leq^\alpha \outpw(\{x,m\}_{pub(y)}) \} )\\
B &= & 
\Forall x,m (\{
\l d,c \r \  | \ \mbox{ if } \l d,c\r \absconcentails \outpw(\{x,m\}_{pub(y)}) \\
 & &\tabp\tabp\tabp\tabp\tabp\tabp \ 
\mbox{then } \l d,c\r \leq^\alpha \l \outp(\{x,m\}_{pub(y)}),\nsusp \r\} )\
\} 
\end{array}
$
\caption{Semantics of the protocol in Example \ref{ex:sus-pro}. \label{fig:sus-sems}}
\end{figure}


\section{Concluding Remarks}\label{sec:concluding}

Several frameworks and abstract domains for the analysis of logic programs 
have been defined 
(see e.g. \cite{CC92,Codish99,armstrong98two}). Those works differ from ours since  they do not 
deal with the temporal behavior and synchronization mechanisms present in \tccp-based languages.  
On the contrary,  since our framework is  parametric w.r.t. the abstract domain, it  can benefit from 
those works.

We defined in \cite{FalaschiOPV07} a framework for the declarative debugging of \ntcc\ 
\cite{NPV02} programs (a non-deterministic extension of \tccp). 
The framework presented here is 
more general since it was designed for the static analysis of  
\tccp\ and \utcc\ programs and not 
only for debugging. Furthermore, as mentioned above, it is parametric w.r.t 
an abstract domain. 
In \cite{FalaschiOPV07} we also dealt with infinite sequences of constraints and a similar finite cut over sequences  was proposed there.

In \cite{Olarte:08:SAC} a symbolic semantics for \utcc\ was proposed to deal with the infinite internal reductions of non well-terminated processes. This semantics, by means of temporal formulas, represents finitely  the infinitely many constraints (and substitutions) the  SOS may produce. The work in  \cite{Olarte:08:PPDP} introduces a denotational semantics for \utcc\ based on (partial) closure operators over sequences of \emph{temporal logic formulas}. This semantics captures compositionally the \emph{symbolic strongest postcondition} and it  was shown to be fully abstract w.r.t. the symbolic semantics for the fragment of locally-independent (see Definition \ref{def:LI}) and abstracted-unless free processes (i.e., processes not containing occurrences of {\bf unless} processes in the scope of abstractions). 
The semantics here presented turns out to be more appropriate to develop the abstract  interpretation framework in Section \ref{sec:absframework}. Firstly, the inclusion relation between the strongest postcondition  and the semantics  is verified for the whole language (Theorem \ref{theo:sound}) -- in \cite{Olarte:08:PPDP} this inclusion is verified 
only for the abstracted-unless free fragment--. Secondly, this semantics 
makes use of the  entailment relation over constraints rather than 
the more involved entailment  over  first-order linear-time temporal 
formulas as in \cite{Olarte:08:PPDP}. 
Finally, our semantics allows us to capture 
the behavior of \tccp\ programs with recursion. This is not possible with the 
semantics in \cite{Olarte:08:PPDP} which was thought only for \utcc\ programs 
where recursion can be encoded. 
This work then provides the theoretical basis for building  
tools for the data-flow analyses 
of \utcc\ and \tccp\ programs. 

For the kind of 
applications that stimulated the development of \utcc, it was defined entirely 
deterministic. The semantics 
here presented could smoothly be  extended to deal with some forms of 
non-determinism like those  in  \cite{Falaschi:97:TCS}, thus widening 
the spectrum of applications of  our framework. 




A framework for the abstract diagnosis of timed-concurrent constraint  programs
has been defined in \cite{CTV11} where the authors consider a denotational 
semantics similar to ours, although with several technical
differences. The language studied in \cite{CTV11} corresponds to \texttt{tccp} \cite{bgm99},  a  temporal \ccp\ language where the stores are monotonically accumulated along the time-units and whose operational semantics 
relies on the notion of true parallelism. We note that the framework developed in \cite{CTV11}  is used for abstract diagnosis rather than for  general analyses.

Our results should foster the development of analyzers 
for different  systems modeled in \utcc\ and its sub-calculi such as security protocols, reactive and timed systems, biological systems, etc (see \cite{DBLP:journals/constraints/OlarteRV13} for a survey of applications of \ccp-based languages).  
We plan also to perform  freeness, suspension, 
type  and  independence analyses among others. It is well known that this 
kind of analyses  have many applications, e.g. for code optimization in compilers, for improving run-time 
execution, and for approximated verification. We also plan to use abstract model checking techniques based on the proposed semantics to automatically analyze \utcc\ and \tccp\ code. 
\\

\noindent{\bf Acknowledgments.}
We thank Frank D. Valencia, Fran\c cois Fages and R\'emy Haemmerl\'e for insightful discussions on different   subjects related to this work. We also thank   the anonymous reviewers for their detailed comments.
Special thanks to Emanuele D'Osualdo for his careful remarks and suggestions for improving the paper.  This work has been partially supported by grant 1251-521-28471 from Colciencias, and by Digiteo and  DGAR funds for visitors.



\newpage
\appendix

\section{Detailed proofs Section \ref{sec:observables}} \label{app:sos}
Before presenting the proof that $\utcc$ is deterministic, we shall prove the following auxiliary result. 
\begin{lemma}[Confluence]\label{lemma:confluence}
Suppose that $\gamma_0 \redi \gamma_1$,
$\gamma_0 \redi \gamma_2$ and $\gamma_1 \not\equiv \gamma_2$.  Then, there exists 
$\gamma_3$ such that $\gamma_1 \redi \gamma_3$ and $\gamma_2 \redi \gamma_3$. 
\end{lemma}
\begin{proof} 
Let $\gamma_0 = \mconf{\vx; P;c}$. The proof proceed by structural induction on $P$.  In each case where $\gamma_0$ has two different transitions  (up to $\equiv$)  $\gamma_0 \redi \gamma_1$ and $\gamma_0 \redi \gamma_2$, one shows the existence of $\gamma_3$ s.t. $\gamma_1 \redi \gamma_3$ and $\gamma_2 \redi \gamma_3$. 

 Given a configuration $\gamma=\mconf{\vx; P;c}$
 let us define the size of $\gamma$ as the size of  $P$  as follows:
$M(\skipp)=0$, $M(\tellp{c})=M(p(\vt))=1$, $M(\absp{\vx}{c;D}{P'})=M(\localp{\vx}{P'}) = M(\nextp{P'})= M(\unlessp{c}{P'}) = 
1 + M(P')$ and $M(Q\parallel R) =M(Q) + M(R)$.
 Suppose  that $\gamma_0\equiv \mconf{\vx; P; c_0 }$, $\gamma_0  \redi \gamma_1$, $\gamma_0 \redi \gamma_2$ and $\gamma_1 \not\equiv \gamma_2$. 
 The proof proceeds by induction on the size of $\gamma_0$. 
 From the assumption $\gamma_1 \not\equiv \gamma_2$, it must be the case that the transition $\redi$ is not an instance of the rule $\rStructVar$; moreover, 
   $P$ is neither  a process of the form $\tellp{c}$, $\localp{\vx}{P}$, 
 $p(\vt)$
 or  $\unlessp{c}{P'}$ (since those processes have a unique possible transition modulo structural congruence) nor $\nextp{P}$ or $\skipp$ (since they do not exhibit any internal derivation). 

 For the case $P= Q \parallel R$, we have to consider three cases.
 Assume that  $\gamma_1 \equiv \langle\vx_1; Q_1 \parallel R, c_1\rangle$
 and $\gamma_2 \equiv \langle \vx_2; Q_2 \parallel R, c_2\rangle$.
 Let $\gamma'_0 \equiv \langle \vx; Q;c_0\rangle$, 
$\gamma'_1 \equiv \langle\vx_1; Q_1;c_1\rangle$ and
$\gamma'_2 \equiv \langle\vx_2; Q_2;c_2\rangle$. 
We know by induction that if $\gamma'_0 \redi 
\gamma'_1$ and 
 $\gamma_0' \redi \gamma'_2$ then there exists
 $\gamma_3'  \equiv \langle \vx_3;Q_3; c_3\rangle$
 such that $\gamma_1' \redi \gamma_3'$ and $\gamma_2' \redi \gamma_3'$.
We conclude by noticing that  $\gamma_1 \redi \gamma_3 $ and $\gamma_2 \redi \gamma_3$ where $\gamma_3 \equiv \langle \vx_3;Q_3 \parallel R; c_3\rangle$. The  remaining cases when  (1) $R$ has two possible transitions and (2) when $Q$ moves to $Q'$ and then $R$ moves to $R'$ are similar. 

Let  $\gamma_0\equiv \langle \vx;P; c_0 \rangle$ with $P=\absp{\vy}{c;D}{Q}$. One can verify that  $\gamma_1\equiv \langle \vx \cup \vx_1; P_1; c_0 \rangle$  where $P_1$ takes the form $\absp{\vy}{c;D\cup\{d_{\vy\vec{t_1}} \} }{Q} \parallel Q[\vec{t_1}/\vy]$ and  $\gamma_2\equiv \langle \vx \cup  \vx_2; P_2; c_0\rangle$  where $P_2$ takes the form $\absp{\vz}{c; D \cup \{d_{\vy\vec{t_2}}\}}{Q} \parallel Q[\vec{t_2}/\vy]$.  From the assumption $\gamma_1 \not\equiv \gamma_2$, it must be the case that  $d_{\vy\vec{t_1}} \not\equivC d_{\vy\vec{t_2}}$. By alpha conversion we assume that $\vx_1 \cap \vx_2 = \emptyset$. 
Let $\gamma_3 \equiv \langle \vx \cup \vx_1 \cup \vx_2; P_3; c_0\rangle$  where $P_3=\absp{\vy}{c; D \cup \{ d_{\vy\vec{t_1}},d_{\vy\vec{t_2}}\}}{Q} \parallel Q[\vec{t_1}/\vy] \parallel Q[\vec{t_2}/\vy]$.  Clearly  $\gamma_1 \redi \gamma_3$  and $\gamma_2 \redi \gamma_3$ as wanted. 
\end{proof}

\begin{observation}[Finite Traces]\label{obs:finite-traces}
Let $
\gamma_1  \redi \cdots \redi \gamma_n\not\redi
$  by a finite internal derivation. The 
number of possible internal transitions (up to $\equiv$) in any $\gamma_i = \mconf{\vx_i;P_i;c_i}$ in the above derivation is finite.
\end{observation}
\begin{proof}
We proceed  by structural induction on $P_i$. The interesting case is the {\bf abs} process. Let $Q= \absp{\vx}{c}{P}$. 
Suppose, to obtain a contradiction, that   $c_i \entails c\sxt$ for infinitely many $\vt$ (to have infinitely many possible internal transitions). In that case, it is easy to see that  we must have infinitely many internal derivation, thus contradicting the assumption that   $\gamma_n \not\redi$. 
\end{proof}

\begin{lemma}[Finite Traces]
If there is a finite internal derivation of the form $
\gamma_1 \redi \gamma_2 \redi \cdots \redi \gamma_n\not\redi
$ 
then, any derivation starting from $\gamma_1$ is finite. 
\end{lemma}
\begin{proof}
We observe that
recursive calls must be guarded by a {\bf next} processes. Then, any infinite behavior inside a time-unit is due to an {\bf abs} process. 
From  Observation \ref{obs:finite-traces} and Lemma \ref{lemma:confluence}, it follows that any derivation starting from $\gamma_1$ is finite. 
\end{proof}

\appResult{Theorem \ref{theo:SOS-determinism}}{Determinism}{
Let $s,w$ and $w'$ be (possibly infinite) sequences of constraints. If both
$(s,w)$, $(s,w') \in \iobehav{P}$ then $w\equivC w'$. 
}
\begin{proof}
Assume that $P\rede{(c,\exists \vx(d))}\localp{\vx}{F(Q)}$, $P\rede{(c,\exists \vx'(d'))}\localp{\vx'}F(Q')$
and let $\gamma_1\equiv\langle \emptyset;P;c \rangle$, $\gamma_2\equiv\langle \emptyset;P;c \rangle$. If $\gamma_1 \not\redi$ then trivially $\gamma_2 \notredi$, $d\equivC d'$ and $Q \equiv Q'$. 
Now assume that $\gamma_1 \redi^{*} \gamma_1' \not\redi$ and $\gamma_2 \redi^{*} 
\gamma_2'\not\redi $ where $\gamma_1' \equiv \langle \vx; Q;d\rangle$ and $\gamma_2' \equiv \langle \vx'; Q';d'\rangle$.  By   repeated applications of Lemma \ref{lemma:confluence}  we conclude $\gamma_1' \equiv \gamma_2'$ and then, $d\equivC d'$ and $Q\equiv Q'$. 
\end{proof}

\appResult{Lemma \ref{lem:COP}}{Closure Properties}{
Let $P$ be a process. Then,
\begin{enumerate}
\item[(1)] $\iobehav{P}$ is a function.
\item[(2)] $\iobehav{P}$  is a  partial closure operator,  namely it satisfies:
\\ \noindent {\bf Extensiveness}: If $(s,s') \in  \iobehav{P}$  then $s \leq s'$.\\
\noindent {\bf Idempotence}: If $(s,s') \in  \iobehav{P}$  then $(s',s') \in  \iobehav{P}$.\\
\noindent {\bf Monotonicity}: Let $P$ be a  monotonic process such that  $(s_1,s_1') \in  \iobehav{P}$. If $(s_2 ,s_2' ) \in  \iobehav{P}$ and  $s_1 \leq s_2$, then  $s_1' \leq s_2'$.
\end{enumerate}
}
\begin{proof}
We shall assume here that the input and output sequences are infinite. The proof for the case when the sequences are finite is analogous. 
The proof of (1) is immediate from Theorem \ref{theo:SOS-determinism}.  For (2), assume that $s=c_1.c_2...$, $s'=c_1'.c_2'...$ and that $(s,s') \in  \iobehav{P}$. We then have a derivation of the form:
\[
P \equiv P_1 \rede{(c_1,c_1')} P_2 \rede{(c_2,c_2')} ... P_i \rede{(c_i,c_i')} P_{i+1} ... 
\]
For $ i\geq 1$, we also know that there is an internal derivation of the form
$
\langle \emptyset;  P_i ; c_i\rangle \redi^*  \langle \vx; P_i' ; c_i'\rangle \not\redi
$
 where  $P_{i+1} = \localp{\vx}{F(P_i')} $.  

 \noindent{\bf Extensiveness} follows from  (1) in 
Lemma \ref{lemma:redi-properties}. 

 \noindent{\bf Idempotence} is proved by repeated applications of (3) in Lemma \ref{lemma:redi-properties}.
 
As for  {\bf Monotonicity}, we proceed as in  \cite{NPV02}. Let $\preceq$ be the minimal ordering relation on  processes satisfying: (1)  $\skipp \preceq P$. (2) If $P \preceq Q$ and $P\equiv P'$ and $Q\equiv Q'$ then 
	$P' \preceq Q'$. (3) If $P \preceq Q$, for every context $C[\cdot]$, $C[P] \preceq C[Q]$.
	Intuitively, $P \preceq Q$  represents the fact that $Q$ contains ``at least as much code'' as $P$. We have to show that for every $P$, $P'$, $c$, $c'$  and $\vx,\vx'$ if 
$\langle \vx; P;c \rangle \redi^*\langle \vx';P';c' \rangle\not\redi$  then for every $d\entails c$ and $Q$ s.t. $P\preceq Q$ there 
$\langle \vx; Q;d \rangle \redi^*\langle \vy;Q';d' \rangle\not\redi$
for some  $\vy$ and $Q'$ with $\localp{\vx'}{F(P')} \preceq \localp{\vy}{F(Q')}$ and  $\exists \vy(d') \entails \exists \vx'(c')$.
 This can be proved by induction on the  length of the derivation using the following two properties:\\
\noindent ({\bf a}) 
$\redi$ is monotonic w.r.t. the store, in the sense that,  if $\langle \vx; P;c \rangle \redi \langle \vx';P';c' \rangle$ then for every $d \entails c$ and $Q$ s.t. 
$P \preceq Q$, 
$\langle\vx; Q;d \rangle \redi \langle \vy; Q';d' \rangle$ where $\exists \vy(d') \entails \exists \vx'(c')$
and $\localp {\vx'}{P'} \preceq \localp{\vy}{Q'}$.

\noindent ({\bf b}) For every monotonic process $P$, if $\langle  \vx;
	P;c\rangle \not\redi$ then for every $d\entails c$ and  $Q$ 
	such that $P \preceq Q$ we have either 
 $\langle\vx; Q;d\rangle \not\redi$ or 
 $\langle \vx; Q;d\rangle \redi^* \langle \vx';Q';d' \rangle \not\redi$ where 
 $\exists \vx'(d') \entails \exists \vx(d)$ and  $\localp{\vx}{F(P)} \preceq \localp{\vx'}{F(Q')}$.
The restriction to  programs which do not contain {\bf unless}  constructs is essential  here. 
\end{proof}

\appResultNN{Theorem \ref{the:col:CO}}{
Let $min$ be the minimum function w.r.t. the order induced by $\leq$ and  $P$ be a monotonic process. Then, 
$(s, s') \in \iobehav{P} \mbox{\ \ iff\ \  } 
s' = min(\spbehav{P} \cap \{w \ | \ s \leq w \})
$
}
\begin{proof}
Let $P$ be a monotonic process and 
$(s,s') \in \iobehav{P}$.
	By extensiveness $s\leq s'$ and by 
	idempotence,   $(s',s')\in \iobehav{P}$. 
	Let $s'' = min(\spbehav{P} \cap \{w \ | \ s \leq w \})$. Since $s' \in \spbehav{P}$ and $s\leq s'$, it must be the case that 
	$s\leq s'' \leq s'$. If $(s'',s''') \in \iobehav{P}$, by monotonicity $s'\leq s'''$. Since $s''\in \spbehav{P}$, $s''\equivC s'''$ and then, $s' \leq s''$. We conclude $s' \equivC s''$. 
\end{proof}

\section{Detailed Proofs Section \ref{sec:denotsem}} \label{app:proofs-den}
\appResult{Observation \ref{l-vx-eq}}{Equality and  $\vx$-variants}{
Let $S \subseteq \cC^\omega$, $\vz \subseteq {\it Var}$ and   $s,w$ be $\vx$-variants such that $\dxt^\omega \leq s$,  $\dxt^\omega \leq w$  and $\adm{\vx}{\vt}$. (1)  $s\equivC w$. (2)   $\exists \vz(s) \in \Forall \vx (S)$ iff $s \in \Forall \vx (S)$. }
\begin{proof}
(1) Let $i\geq1$,  $c=s(i)$ and $d=w(i)$. We  prove that $c\entails d$ and $d\entails c$.  We know that  
$c\sqcup \dxt \equivC c$, $d\sqcup \dxt \equivC d$ and $\exists \vx (c\sqcup \dxt) \equivC \exists \vx (d\sqcup \dxt)$. Hence, $c\sxt \equivC d\sxt$. Since $c \entails \exists \vx( c)$, we can  show that 
$c \entails \exists \vx(d \sqcup \dxt)$ and then, $c\entails  d\sxt $.
Since $d\sxt \sqcup \dxt \entails d$ (Notation \ref{not:terms})  we conclude $c \entails d$. The  ``$d\entails c$'' side is analogous and we conclude $c \equivC d$. \\
Property (2) follows directly from the definition of $\Forall (\cdot)$. 
\end{proof}

\appResultNN{Lemma \ref{lem:soundness}}{
Let $\os \cdot \cs$ be as in Definition \ref{def:conc-semantics}.
If $P\rede{(d, d')}{R}$ and $d \equivC d'$,  then $d.\os R \cs \subseteq  \os P\cs$.
}
\begin{proof}
Assume that $\langle\vx; P;d\rangle \redi^* \langle \vx'; P';d'\rangle \not\redi$,  $\exists \vx(d) \equivC \exists \vx'(d')$.
 We shall prove that $\exists \vx(d). \Exists \vx' (\os F(P')) \cs \subseteq \Exists \vx(\os P\cs)$. 
	We proceed by induction on the lexicographical order on the length of the internal derivation   and the structure of  $P$, where the predominant component is the length of the derivation. Here we present the missing cases in the body of the paper. 

\noindent \underline{{\bf Case} $P=\skipp$}. This case is trivial. 

\noindent \underline{{\bf Case}  $P=\tellp{c}$}. If $\langle \vx;\tellp{c};d \rangle \redi \langle \vx; \skipp,d \rangle $ then it must be the case that $d\equivC d\sqcup  c$ and  $d\entails c$. We conclude  $\exists \vx(d).\os \skipp \cs \subseteq \Exists \vx(\os \tellp{c}\cs)$.

\noindent \underline{{\bf Case}  $P = \localp{\vx;c}{Q}$}.  Consider the following derivation 
 \[
 \mconf{\vy;\localp{\vx}{Q};d} \redi  \mconf{\vy\cup\vx;Q;d} \redi^* \mconf{\vy\cup\vx';Q';d'}\notredi
 \]
 where, by alpha-conversion,  $\vx \cap \vy = \emptyset$ and $\vx \cap \fv(d) = \emptyset$. Assume that 
   $\exists \vy(d) \equivC \exists \vy \exists\vx'  (d')$. 
 Since the derivation starting from $Q$ is shorter than that starting from $P$, we conclude
$
  \exists \vy(d).\Exists \vy,\vx' \os F(Q')\cs \subseteq \Exists \vx,\vy \os Q\cs$.
 
%
  
\noindent \underline{{\bf Case}  $P=\nextp{Q}$}. This case is trivial since
	$d.\os Q\cs \subseteq \os P\cs$ for any $d$. 
	
\noindent \underline{{\bf Case }$P=\unlessp{c}{Q}$}. We distinguish two cases:
(1) If $d \entails c$, then we have $\langle \vx; \unlessp{c}{Q};d\rangle \redi \langle \vx; \skipp ;d \rangle \not\redi$ and  we conclude $\Exists \vx(d).\os \skipp \cs \subseteq \Exists \vx \os \unlessp{c}{P}\cs$. (2), the case when     $d\not\entails c$  is similar to the case of $P=\nextp{Q}$.

\end{proof}

\appResult{Lemma \ref{theo:comp}}{Completeness}{
Let  $\cD.P$ be a locally independent program s.t. $d.s \in \os P \cs$. If $P\rede{(d,d')}R$ then $d'\equivC d$ and 
$s \in \os R \cs$. 
}
\begin{proof} 
Assume that   $P$ is locally independent, $d.s \in \os P \cs$ and there is a derivation of the form  $\langle \vx;P;d\rangle \redi^*\langle \vx';P';d'\rangle \not\redi$. We shall prove that $\exists x(d) \equivC \exists \vx'(d')$ and
$s\in \Exists \vx' \os  F(P')\cs$. 
	We proceed by induction on the lexicographical order on the length of the internal derivation  ($\redi^*$)  and the structure of  $P$, where the predominant component is the length of the derivation. 
 The locally independent condition is used for the case $P=\localp{\vx;c}{Q}$.
 We present here the missing cases in the body of the paper. 

\noindent \underline{{\bf Case}  $\skipp$}. This case is trivial

\noindent \underline{{\bf Case}  $P = \tellp{c}$}. This case is trivial since it must be the case that $d\entails c$ and hence $d\sqcup c\equivC d$. 

\noindent \underline{{\bf Case}  $P=\nextp{Q}$}. 
	This case is trivial since $\langle \vx;P;d \rangle \not\redi$ for any $d$ and $\vx$ and $F(P)=Q$. 
	
\noindent \underline{{\bf Case}  $P=\unlessp{c}{Q}$}.If  $d \entails c$ the case is trivial. If  $d\not\entails c$  the case is similar to that of $P=\nextp{Q}$. 

\noindent \underline{{\bf Case}  $P=p(\vt)$}.  Assume that 
$p(\vx) :- Q \in  \cD$.   If $d.s \in \os p(\vt)\cs$ then   $d.s\in \os  Q\sxt \cs $. By using the rule $\rCall$ we can show that there is a derivation
\[
\langle \vy; p(\vx);d\rangle \redi \langle \vy; Q\sxt ;d\rangle \redi^* \langle \vy'; Q';d'\rangle \not\redi
\]

By inductive hypothesis we know that $\exists y'(d')\equivC \exists \vy(d)$ and 
$s\in \Exists \vy' \os F( Q') \cs$. 

\end{proof}

\section{Detailed Proofs Section \ref{sec:absframework}} \label{app-sec-abs}
\appResult{Theorem \ref{teo:corr}}{
Soundness of the approximation}
{
Let $(\cC, \alpha, \cA )$  be a description and   ${\mathbf{A}}$  be upper correct w.r.t. $\mathbf{C}$.   Given a  \utcc\ program $\cD.P$, if 
$ s \in \os P \cs $ then  $\alpha(s) \in \os P \cs^{\alpha}$.  
}
\begin{proof} 
Let $d_{\alpha}.s_{\alpha} = \alpha(d.s)$ and assume that $d.s \in  \os P \cs$. Then, $d.s \in \os P\cs_I$ where $I$ is the lfp of $T_{\cD}$. By the continuity of $T_{\cD}$,  there exists $n$ s.t.  $I = T_{\cD}^n(I_\bot)$ (the $n$-th application of $T_{\cD}$). 
We proceed by  induction on the lexicographical order on 
the pair $n$ and the structure of $P$, where the predominant component is the length $n$. We present here the missing cases in the body of the paper. 

\noindent \underline{{\bf Case}  $P=\skipp$}. This case is trivial.

\noindent \underline{{\bf Case}   $P=\tellp{c}$}. We must have $d\entails c$ and by monotonicity of $\alpha$,  $d_{\alpha} \absentails \alpha(c)$. We conclude  $d_{\alpha}.s_{\alpha} \in \os \tellp{c}  \cs^\alpha$.

\noindent \underline{{\bf Case}   $P = Q \parallel R$}. We must have that $s\in \os Q\cs$
	and $s\in \os R \cs$.  By inductive hypothesis we know that
	$s_{\alpha} \in \os Q\cs^\alpha$ and
	$s_{\alpha} \in \os R\cs^\alpha$ and then,
	$s_{\alpha} \in \os Q \parallel R \cs^\alpha$.


\noindent \underline{{\bf Case}   $P=\localp{\vx}{Q}$}. It must be the case that there exists $d'.s'$ $\vx$-variant of $d.s$ s.t.
		$d'.s' \in \os Q \cs$. Then,  by (structural) inductive hypothesis $\alpha(d'.s') \in \os Q\csa$. We conclude by using the properties of $\alpha$ in Definition \ref{dec:corapp} to show that $\existsa \vx(\alpha(d.s)) = \existsa \vx(\alpha(d'.s'))$, i.e., $\alpha(d.s)$ and $\alpha(d'.s')$ are $\vx$-variants, and then, $d_\alpha.s_\alpha \in \os \localp{\vx}{Q} \csa$. 

\noindent \underline{{\bf Case}   $P=\nextp{Q}$}. 
	We know that $s \in \os Q\cs$ and by inductive hypothesis
	$\alpha(s) \in \os Q\cs^\alpha$. We then conclude $d_{\alpha}.s_{\alpha} \in \os P \cs^\alpha$. 
	
\noindent \underline{{\bf Case}   $P=\unlessp{c}{Q}$}. This case is trivial since $\cA$ approximates every possible concrete computation.
\end{proof}

\section{Auxiliary results} \label{app:proofs-aux}
\begin{proposition}\label{prop:x-variants-free-vars}
Let $P$ be a process such that   $\vx \cap \fv(P) = \emptyset$ and let $d.s \in \os P\cs$. If   $d'.s'$ is an $\vx$-variant of $d.s$ then $d'.s' \in \os P \cs$. 
\end{proposition}
\begin{proof}
	The proof proceeds  by induction on the structure of $P$. We shall use the notation $c(\vy)$ and $P(\vy)$ to denote constraints and processes where the free variables are exactly $\vy$ and we shall assume that $\vy \cap \vx = \emptyset$. We assume that $d.s \in \os P(\vy) \cs$ and $d'.s'$ is an $\vx$-variant of $d.s$. 
 We consider the following cases.  The others are easy. 

\noindent \underline{{\bf Case}  $P = \whenp{c(\vy)}{Q(\vy)}$}. If $d' \entails c(\vy)$ then, by monotonicity, $\exists \vx (d') \entails \exists \vx (c(\vy))$ and then $\exists \vx(d) \entails c(\vy)$. Hence, it must be the case that $d \entails c(\vy)$ and  $d.s \in \os Q(\vy)\cs$.  By induction we conclude $d'.s' \in \os Q(\vy)\cs$. If $d' \notentails c(\vy)$, then $\exists \vx(d') \not\entails c(\vy)$ (since $\exists \vx (d') \leq d'$). Hence, $d \not\entails c(\vy)$ and trivially, $d.s \in \os P\cs$ and so $d'.s' \in \os P\cs$. 

\noindent \underline{{\bf Case} $P = \absp{\vz}{c(\vz,\vy)}{Q(\vz,\vy)}$}. 
			We know that 
		$d.s \in \Forall \vz \os \whenp{c(\vz,\vy)}{Q(\vz,\vy)}\cs$. 
		By definition of the operator $\Forall(\cdot)$, $\exists \vx (d.s) \in \os P\cs $. Since   $\exists \vx (d'.s') \equivC \exists \vx (d.s)$ we conclude $d'.s' \in \os P\cs$. 
\end{proof}

\begin{proposition}\label{prop:exists-free-vars}
If $\vx \cap \fv(P) = \emptyset$ then $\os P\cs = \Exists \vx \os P\cs$.
\end{proposition}
\begin{proof}
		The case $\os P\cs \subseteq \Exists \vx \os P \cs$ is trivial by the definition of $\Exists(\cdot)$. 
		The case $ \Exists \vx \os P \cs\subseteq \os P\cs $, follows directly from Proposition \ref{prop:x-variants-free-vars}.
\end{proof}

\begin{proposition} \label{prop:den-se-ext}
If $\vx \not\in \fv(Q)$ then $\Exists \vx (\os P\cs \cap \os Q\cs) = \Exists \vx(\os P\cs) \cap \os Q\cs$.
\end{proposition}
\begin{proof}
\noindent ($\subseteq$): Let $d.s \in \Exists \vx (\os P\cs \cap \os Q\cs)$. Then, there exists an $\vx$-variant $d'.s'$ s.t. $d'.s' \in \os P \cs \cap \os Q\cs$. Then, $d.s \in \Exists \vx (\os P\cs)$ (by definition) and $d.s \in \os Q \cs$ by Proposition \ref{prop:x-variants-free-vars}.

\noindent ($\supseteq$): Let $d.s \in \Exists \vx(\os P\cs) \cap \os Q\cs$.  Then, there exists $d'.s'$ $\vx$-variant of $d.s$ s.t. $d'.s' \in \os P \cs$. By Proposition \ref{prop:x-variants-free-vars}, 
$d'.s' \in \os Q \cs$ and therefore, $d.s \in \Exists \vx (\os P \cs \cap \os Q \cs)$. 
\end{proof}

In Theorem \ref{teo:corr}, the proof of the ${\bf abs }$   case  requires the following auxiliary results (similar to those in the concrete semantics). 

   \begin{observation}[Equality and $\vx$-variants]\label{l-vx-eq-abs}
Let $s_\alpha$ and $w_\alpha$ be $\vx$-variants such that $({\dxta})^{\omega} \leq^\alpha s_\alpha$,  $({\dxta})^{\omega} \leq^\alpha w_\alpha$  and $\adm{\vx}{\vt}$. Then $s_\alpha\equivC^\alpha w_\alpha$.
\end{observation}
\begin{proof}
Let $c_\alpha=s_\alpha(i)$ and $d_\alpha=w_\alpha(i)$ with $i\geq1$. We shall prove that $c_\alpha\absentails d_\alpha$ and $d_\alpha \absentails c_\alpha$.  We know that  
$c_\alpha \sqcup^\alpha \dxta \equivC^\alpha c_a$ and $d_\alpha \sqcup^\alpha \dxta \equivC^\alpha d_\alpha$. We also know that 
$\exists^\alpha \vx (c_\alpha\sqcup^\alpha \dxta) \equivC^\alpha \exists^\alpha \vx (d_\alpha\sqcup^\alpha \dxta)$. Since $c_\alpha \absentails \exists^\alpha \vx( c_\alpha)$, 
we can  show that 
$c_\alpha \absentails \exists^\alpha \vx(d_\alpha \sqcup^\alpha \dxta)$.
Furthermore, 
$\exists^\alpha \vx(d_\alpha \sqcup^\alpha \dxta)\sqcup^\alpha \dxta \absentails d_\alpha$ (see Notation \ref{not:terms}). 
Hence, we conclude 
$c_\alpha \absentails d_\alpha$. The proof of $d_\alpha\absentails c_\alpha$ is analogous. 
\end{proof}

\begin{proposition}\label{prop:forall-subs-abs}
$s_\alpha \in  \Forall \vx ( \os P \cs^\alpha_X)$ if and only if  $s \in \os P\sxt \cs^\alpha_X$ for all admissible substitution $\sxt$. 
\end{proposition}
\begin{proof}
($\Rightarrow$)Let $s_\alpha\in\Forall \vx ( \os P \cs^\alpha_X)$ and  $s_\alpha'$ be an $\vx$-variant of $s_\alpha$ s.t. $({\dxta})^{\omega} \leq^\alpha s_\alpha'$ where $\adm{\vx}{\vt}$. By definition of~ $\Forall$, we know that $s'_\alpha\in \os P \cs^\alpha_X$. Since $ ({\dxta})^\omega \leq^\alpha s_\alpha'$ then $s_\alpha' \in \os P \cs^\alpha_X \cap \up(({\dxta})^{\omega})$. Hence, $s_\alpha\in \Exists^\alpha \vx (\os P \cs^\alpha_X \cap \up(({\dxta})^{\omega}) )$ and we conclude $s_\alpha \in \os P\sxt\cs^\alpha_X$. 

\noindent($\Leftarrow$) Let $\sxt$ be an admissible substitution. Suppose, to obtain a contradiction, that $s_\alpha \in \os P\sxt\cs^\alpha_X$, there exists $s'_\alpha$ $\vec{x}$-variant of $s_\alpha$ s.t. $({\dxta})^{\omega} \leq^\alpha s'_\alpha$ and $s'_\alpha \notin \os P\cs^\alpha_X$ (i.e., $s_\alpha\notin \Forall\vx (\os P\cs^\alpha_X)$). Since $s_\alpha\in \os P\sxt\cs^\alpha_X$ then $s_\alpha \in \Exists^\alpha \vx (\os P \cs^\alpha_X \cap \up{ ({\dxta})^{\omega}})$.  Therefore, there exists $s''_\alpha$ $\vx$-variant of $s_\alpha$ s.t. $s''_\alpha \in \os P\cs^\alpha_X$ and ${\dxta}^{\omega} \leq^\alpha s''_\alpha$. By Observation \ref{l-vx-eq-abs},  $s'_\alpha\equivC^\alpha s''_\alpha$ and thus,  $s'_\alpha \in \os P \cs^\alpha_X$, a contradiction. 
\end{proof}


\begin{thebibliography}{}

\bibitem[\protect\citeauthoryear{Armando and Compagna}{Armando and
  Compagna}{2008}]{DBLP:journals/ijisec/ArmandoC08}
{\sc Armando, A.} {\sc and} {\sc Compagna, L.} 2008.
\newblock Sat-based model-checking for security protocols analysis.
\newblock {\em Internation Journal of Information Security\/}~{\em 7,\/}~1,
  3--32.

\bibitem[\protect\citeauthoryear{Armstrong, Marriott, Schachte, and
  S{\o}ndergaard}{Armstrong et~al\mbox{.}}{1998}]{armstrong98two}
{\sc Armstrong, T.}, {\sc Marriott, K.}, {\sc Schachte, P.}, {\sc and} {\sc
  S{\o}ndergaard, H.} 1998.
\newblock Two classes of {Boolean} functions for dependency analysis.
\newblock {\em Science of Computer Programming\/}~{\em 31,\/}~1, 3--45.

\bibitem[\protect\citeauthoryear{Berry and Gonthier}{Berry and
  Gonthier}{1992}]{BeGo92}
{\sc Berry, G.} {\sc and} {\sc Gonthier, G.} 1992.
\newblock The {{\sc Esterel}} synchronous programming language: Design,
  semantics, implementation.
\newblock {\em Science of Computer Programming\/}~{\em 19,\/}~2, 87--152.

\bibitem[\protect\citeauthoryear{Bodei, Brodo, Degano, and Gao}{Bodei
  et~al\mbox{.}}{2010}]{BBD10}
{\sc Bodei, C.}, {\sc Brodo, L.}, {\sc Degano, P.}, {\sc and} {\sc Gao, H.}
  2010.
\newblock Detecting and preventing type flaws at static time.
\newblock {\em Journal of Computer Security\/}~{\em 18,\/}~2, 229--264.

\bibitem[\protect\citeauthoryear{Boreale}{Boreale}{2001}]{boreale01symbolic}
{\sc Boreale, M.} 2001.
\newblock Symbolic trace analysis of cryptographic protocols.
\newblock In {\em ICALP}, {F.~Orejas}, {P.~G. Spirakis}, {and} {J.~van
  Leeuwen}, Eds. LNCS, vol. 2076. Springer, 667--681.

\bibitem[\protect\citeauthoryear{Codish and Demoen}{Codish and
  Demoen}{1994}]{CodishD94}
{\sc Codish, M.} {\sc and} {\sc Demoen, B.} 1994.
\newblock Deriving polymorphic type dependencies for logic programs using
  multiple incarnations of prop.
\newblock In {\em SAS}, {B.~L. Charlier}, Ed. LNCS, vol. 864. Springer,
  281--296.

\bibitem[\protect\citeauthoryear{Codish, Falaschi, and Marriott}{Codish
  et~al\mbox{.}}{1994}]{CFM94}
{\sc Codish, M.}, {\sc Falaschi, M.}, {\sc and} {\sc Marriott, K.} 1994.
\newblock Suspension {A}nalyses for {C}oncurrent {L}ogic {P}rograms.
\newblock {\em ACM Transactions on Programming Languages and Systems\/}~{\em
  16,\/}~3, 649--686.

\bibitem[\protect\citeauthoryear{Codish, Falaschi, Marriott, and
  Winsborough}{Codish et~al\mbox{.}}{1997}]{CFMW97}
{\sc Codish, M.}, {\sc Falaschi, M.}, {\sc Marriott, K.}, {\sc and} {\sc
  Winsborough, W.} 1997.
\newblock A {C}onfluent {S}emantic {B}asis for the {A}nalysis of {C}oncurrent
  {C}onstraint {L}ogic {P}rograms.
\newblock {\em Journal of Logic Programming\/}~{\em 30,\/}~1, 53--81.

\bibitem[\protect\citeauthoryear{Codish, S{\o}ndergaard, and Stuckey}{Codish
  et~al\mbox{.}}{1999}]{Codish99}
{\sc Codish, M.}, {\sc S{\o}ndergaard, H.}, {\sc and} {\sc Stuckey, P.} 1999.
\newblock Sharing and groundness dependencies in logic programs.
\newblock {\em ACM Transations on Programming Languages and Systems\/}~{\em
  21,\/}~5, 948--976.

\bibitem[\protect\citeauthoryear{Comini, Titolo, and Villanueva}{Comini
  et~al\mbox{.}}{2011}]{CTV11}
{\sc Comini, M.}, {\sc Titolo, L.}, {\sc and} {\sc Villanueva, A.} 2011.
\newblock Abstract diagnosis for timed concurrent constraint programs.
\newblock {\em TPLP\/}~{\em 11,\/}~4-5, 487--502.

\bibitem[\protect\citeauthoryear{Cousot and Cousot}{Cousot and
  Cousot}{1979}]{CC79}
{\sc Cousot, P.} {\sc and} {\sc Cousot, R.} 1979.
\newblock Systematic design of program analysis frameworks.
\newblock In {\em POPL}, {A.~V. Aho}, {S.~N. Zilles}, {and} {B.~K. Rosen}, Eds.
  ACM Press, 269--282.

\bibitem[\protect\citeauthoryear{Cousot and Cousot}{Cousot and
  Cousot}{1992}]{CC92}
{\sc Cousot, P.} {\sc and} {\sc Cousot, R.} 1992.
\newblock Abstract {I}nterpretation and {A}pplications to {L}ogic {P}rograms.
\newblock {\em Journal of Logic Programming\/}~{\em 13,\/}~2\&3, 103--179.

\bibitem[\protect\citeauthoryear{de~Boer, Gabbrielli, Marchiori, and
  Palamidessi}{de~Boer et~al\mbox{.}}{1997}]{deBoer:97:TOPLAS}
{\sc de~Boer, F.~S.}, {\sc Gabbrielli, M.}, {\sc Marchiori, E.}, {\sc and} {\sc
  Palamidessi, C.} 1997.
\newblock Proving concurrent constraint programs correct.
\newblock {\em ACM Transactions on Programming Languages and Systems\/}~{\em
  19,\/}~5, 685--725.

\bibitem[\protect\citeauthoryear{de~Boer, Gabbrielli, and Meo}{de~Boer
  et~al\mbox{.}}{2000}]{bgm99}
{\sc de~Boer, F.~S.}, {\sc Gabbrielli, M.}, {\sc and} {\sc Meo, M.~C.} 2000.
\newblock A timed concurrent constraint language.
\newblock {\em Inf. Comput.\/}~{\em 161,\/}~1, 45--83.

\bibitem[\protect\citeauthoryear{de~Boer, Pierro, and Palamidessi}{de~Boer
  et~al\mbox{.}}{1995}]{BoerPP95}
{\sc de~Boer, F.~S.}, {\sc Pierro, A.~D.}, {\sc and} {\sc Palamidessi, C.}
  1995.
\newblock Nondeterminism and infinite computations in constraint programming.
\newblock {\em Theoretical Computer Science\/}~{\em 151,\/}~1, 37--78.

\bibitem[\protect\citeauthoryear{Dolev and Yao}{Dolev and
  Yao}{1983}]{dolev-yao}
{\sc Dolev, D.} {\sc and} {\sc Yao, A.~C.} 1983.
\newblock On the security of public key protocols.
\newblock {\em IEEE Transactions on Information Theory\/}~{\em 29,\/}~12,
  198--208.

\bibitem[\protect\citeauthoryear{Escobar, Meadows, and Meseguer}{Escobar
  et~al\mbox{.}}{2011}]{DBLP:journals/corr/abs-1105-5282}
{\sc Escobar, S.}, {\sc Meadows, C.}, {\sc and} {\sc Meseguer, J.} 2011.
\newblock State space reduction in the maude-nrl protocol analyzer.
\newblock {\em CoRR\/}~{\em abs/1105.5282}.

\bibitem[\protect\citeauthoryear{Fages, Ruet, and Soliman}{Fages
  et~al\mbox{.}}{2001}]{DBLP:journals/iandc/FagesRS01}
{\sc Fages, F.}, {\sc Ruet, P.}, {\sc and} {\sc Soliman, S.} 2001.
\newblock Linear concurrent constraint programming: Operational and phase
  semantics.
\newblock {\em Inf. Comput.\/}~{\em 165,\/}~1, 14--41.

\bibitem[\protect\citeauthoryear{Falaschi, Gabbrielli, Marriott, and
  Palamidessi}{Falaschi et~al\mbox{.}}{1993}]{FalaschiGMP93}
{\sc Falaschi, M.}, {\sc Gabbrielli, M.}, {\sc Marriott, K.}, {\sc and} {\sc
  Palamidessi, C.} 1993.
\newblock Compositional analysis for concurrent constraint programming.
\newblock In {\em LICS}. IEEE Computer Society, 210--221.

\bibitem[\protect\citeauthoryear{Falaschi, Gabbrielli, Marriott, and
  Palamidessi}{Falaschi et~al\mbox{.}}{1997a}]{Falaschi:97:TCS}
{\sc Falaschi, M.}, {\sc Gabbrielli, M.}, {\sc Marriott, K.}, {\sc and} {\sc
  Palamidessi, C.} 1997a.
\newblock Confluence in concurrent constraint programming.
\newblock {\em Theoretical Computer Science\/}~{\em 183,\/}~2, 281--315.

\bibitem[\protect\citeauthoryear{Falaschi, Gabbrielli, Marriott, and
  Palamidessi}{Falaschi
  et~al\mbox{.}}{1997b}]{DBLP:journals/iandc/FalaschiGMP97}
{\sc Falaschi, M.}, {\sc Gabbrielli, M.}, {\sc Marriott, K.}, {\sc and} {\sc
  Palamidessi, C.} 1997b.
\newblock Constraint logic programming with dynamic scheduling: A semantics
  based on closure operators.
\newblock {\em Inf. Comput.\/}~{\em 137,\/}~1, 41--67.

\bibitem[\protect\citeauthoryear{Falaschi, Olarte, and Palamidessi}{Falaschi
  et~al\mbox{.}}{2009}]{Falaschi:PPDP:09}
{\sc Falaschi, M.}, {\sc Olarte, C.}, {\sc and} {\sc Palamidessi, C.} 2009.
\newblock A framework for abstract interpretation of timed concurrent
  constraint programs.
\newblock In {\em PPDP}, {A.~Porto} {and} {F.~J. L{\'o}pez-Fraguas}, Eds. ACM,
  207--218.

\bibitem[\protect\citeauthoryear{Falaschi, Olarte, Palamidessi, and
  Valencia}{Falaschi et~al\mbox{.}}{2007}]{FalaschiOPV07}
{\sc Falaschi, M.}, {\sc Olarte, C.}, {\sc Palamidessi, C.}, {\sc and} {\sc
  Valencia, F.} 2007.
\newblock Declarative diagnosis of temporal concurrent constraint programs.
\newblock In {\em ICLP}, {V.~Dahl} {and} {I.~Niemel{\"a}}, Eds. LNCS, vol.
  4670. Springer, 271--285.

\bibitem[\protect\citeauthoryear{Falaschi and Villanueva}{Falaschi and
  Villanueva}{2006}]{FalaschiV06}
{\sc Falaschi, M.} {\sc and} {\sc Villanueva, A.} 2006.
\newblock Automatic verification of timed concurrent constraint programs.
\newblock {\em TPLP\/}~{\em 6,\/}~3, 265--300.

\bibitem[\protect\citeauthoryear{Fiore and Abadi}{Fiore and
  Abadi}{2001}]{compsym-fiore}
{\sc Fiore, M.~P.} {\sc and} {\sc Abadi, M.} 2001.
\newblock Computing symbolic models for verifying cryptographic protocols.
\newblock In {\em CSFW}. IEEE Computer Society, 160--173.

\bibitem[\protect\citeauthoryear{Giacobazzi, Debray, and Levi}{Giacobazzi
  et~al\mbox{.}}{1995}]{DBLP:journals/jlp/GiacobazziDL95}
{\sc Giacobazzi, R.}, {\sc Debray, S.~K.}, {\sc and} {\sc Levi, G.} 1995.
\newblock Generalized semantics and abstract interpretation for constraint
  logic programs.
\newblock {\em J. Log. Program.\/}~{\em 25,\/}~3, 191--247.

\bibitem[\protect\citeauthoryear{Haemmerl{\'e}, Fages, and
  Soliman}{Haemmerl{\'e} et~al\mbox{.}}{2007}]{DBLP:conf/fsttcs/HaemmerleFS07}
{\sc Haemmerl{\'e}, R.}, {\sc Fages, F.}, {\sc and} {\sc Soliman, S.} 2007.
\newblock Closures and modules within linear logic concurrent constraint
  programming.
\newblock In {\em FSTTCS}, {V.~Arvind} {and} {S.~Prasad}, Eds. LNCS, vol. 4855.
  Springer, 544--556.

\bibitem[\protect\citeauthoryear{Hentenryck, Saraswat, and Deville}{Hentenryck
  et~al\mbox{.}}{1998}]{HentenryckSD98}
{\sc Hentenryck, P.~V.}, {\sc Saraswat, V.~A.}, {\sc and} {\sc Deville, Y.}
  1998.
\newblock Design, implementation, and evaluation of the constraint language
  cc(fd).
\newblock {\em Journal of Logic Programming\/}~{\em 37,\/}~1-3, 139--164.

\bibitem[\protect\citeauthoryear{Hildebrandt and L{\'o}pez}{Hildebrandt and
  L{\'o}pez}{2009}]{Lopez09}
{\sc Hildebrandt, T.} {\sc and} {\sc L{\'o}pez, H.~A.} 2009.
\newblock Types for secure pattern matching with local knowledge in universal
  concurrent constraint programming.
\newblock In {\em ICLP}, {P.~M. Hill} {and} {D.~S. Warren}, Eds. LNCS, vol.
  5649. Springer, 417--431.

\bibitem[\protect\citeauthoryear{Jaffar and Lassez}{Jaffar and
  Lassez}{1987}]{DBLP:conf/popl/JaffarL87}
{\sc Jaffar, J.} {\sc and} {\sc Lassez, J.-L.} 1987.
\newblock Constraint logic programming.
\newblock In {\em POPL}. ACM Press, 111--119.

\bibitem[\protect\citeauthoryear{Jagadeesan, Marrero, Pitcher, and
  Saraswat}{Jagadeesan et~al\mbox{.}}{2005}]{JagadeesanM05}
{\sc Jagadeesan, R.}, {\sc Marrero, W.}, {\sc Pitcher, C.}, {\sc and} {\sc
  Saraswat, V.~A.} 2005.
\newblock Timed constraint programming: a declarative approach to usage
  control.
\newblock In {\em PPDP}, {P.~Barahona} {and} {A.~P. Felty}, Eds. ACM, 164--175.

\bibitem[\protect\citeauthoryear{L{\'o}pez, Olarte, and P{\'e}rez}{L{\'o}pez
  et~al\mbox{.}}{2009}]{Lopez-Places09}
{\sc L{\'o}pez, H.~A.}, {\sc Olarte, C.}, {\sc and} {\sc P{\'e}rez, J.~A.}
  2009.
\newblock Towards a unified framework for declarative structured
  communications.
\newblock In {\em PLACES}, {A.~R. Beresford} {and} {S.~J. Gay}, Eds. EPTCS,
  vol.~17. 1--15.

\bibitem[\protect\citeauthoryear{Lowe}{Lowe}{1996}]{lowe95attack}
{\sc Lowe, G.} 1996.
\newblock Breaking and fixing the needham-schroeder public-key protocol using
  fdr.
\newblock {\em Software - Concepts and Tools\/}~{\em 17,\/}~3, 93--102.

\bibitem[\protect\citeauthoryear{Maher}{Maher}{1988}]{Mah88b}
{\sc Maher, M.~J.} 1988.
\newblock Complete axiomatizations of the algebras of finite, rational and
  infinite trees.
\newblock In {\em LICS}. IEEE Computer Society, 348--357.

\bibitem[\protect\citeauthoryear{Mendler, Panangaden, Scott, and Seely}{Mendler
  et~al\mbox{.}}{1995}]{MendlerPSS95}
{\sc Mendler, N.~P.}, {\sc Panangaden, P.}, {\sc Scott, P.~J.}, {\sc and} {\sc
  Seely, R. A.~G.} 1995.
\newblock A logical view of concurrent constraint programming.
\newblock {\em Nordic Journal of Computing\/}~{\em 2,\/}~2, 181--220.

\bibitem[\protect\citeauthoryear{Milner, Parrow, and Walker}{Milner
  et~al\mbox{.}}{1992}]{milner.parrow.ea:calculus-mobile}
{\sc Milner, R.}, {\sc Parrow, J.}, {\sc and} {\sc Walker, D.} 1992.
\newblock A calculus of mobile processes, {P}arts {I} and {II}.
\newblock {\em Inf. Comput.\/}~{\em 100,\/}~1, 1--40.

\bibitem[\protect\citeauthoryear{Nielsen, Palamidessi, and Valencia}{Nielsen
  et~al\mbox{.}}{2002a}]{NPV02}
{\sc Nielsen, M.}, {\sc Palamidessi, C.}, {\sc and} {\sc Valencia, F.} 2002a.
\newblock Temporal concurrent constraint programming: Denotation, logic and
  applications.
\newblock {\em Nordic J. of Computing\/}~{\em 9,\/}~1, 145--188.

\bibitem[\protect\citeauthoryear{Nielsen, Palamidessi, and Valencia}{Nielsen
  et~al\mbox{.}}{2002b}]{DBLP:conf/ppdp/NielsenPV02}
{\sc Nielsen, M.}, {\sc Palamidessi, C.}, {\sc and} {\sc Valencia, F.~D.}
  2002b.
\newblock On the expressive power of temporal concurrent constraint programming
  languages.
\newblock In {\em PPDP}. ACM, 156--167.

\bibitem[\protect\citeauthoryear{Olarte, Rueda, and Valencia}{Olarte
  et~al\mbox{.}}{2013}]{DBLP:journals/constraints/OlarteRV13}
{\sc Olarte, C.}, {\sc Rueda, C.}, {\sc and} {\sc Valencia, F.~D.} 2013.
\newblock Models and emerging trends of concurrent constraint programming.
\newblock {\em Constraints\/}~{\em 18,\/}~4, 535--578.

\bibitem[\protect\citeauthoryear{Olarte and Valencia}{Olarte and
  Valencia}{2008a}]{Olarte:08:PPDP}
{\sc Olarte, C.} {\sc and} {\sc Valencia, F.~D.} 2008a.
\newblock The expressivity of universal timed {C}{C}{P}: undecidability of
  monadic {F}{L}{T}{L} and closure operators for security.
\newblock In {\em PPDP}, {S.~Antoy} {and} {E.~Albert}, Eds. ACM, 8--19.

\bibitem[\protect\citeauthoryear{Olarte and Valencia}{Olarte and
  Valencia}{2008b}]{Olarte:08:SAC}
{\sc Olarte, C.} {\sc and} {\sc Valencia, F.~D.} 2008b.
\newblock Universal concurrent constraint programing: symbolic semantics and
  applications to security.
\newblock In {\em SAC}, {R.~L. Wainwright} {and} {H.~Haddad}, Eds. ACM,
  145--150.

\bibitem[\protect\citeauthoryear{Saraswat}{Saraswat}{1993}]{cp-book}
{\sc Saraswat, V.~A.} 1993.
\newblock {\em Concurrent Constraint Programming}.
\newblock MIT Press.

\bibitem[\protect\citeauthoryear{Saraswat, Jagadeesan, and Gupta}{Saraswat
  et~al\mbox{.}}{1994}]{tcc-lics94}
{\sc Saraswat, V.~A.}, {\sc Jagadeesan, R.}, {\sc and} {\sc Gupta, V.} 1994.
\newblock Foundations of timed concurrent constraint programming.
\newblock In {\em LICS}. IEEE Computer Society, 71--80.

\bibitem[\protect\citeauthoryear{Saraswat, Rinard, and Panangaden}{Saraswat
  et~al\mbox{.}}{1991}]{SRP91}
{\sc Saraswat, V.~A.}, {\sc Rinard, M.~C.}, {\sc and} {\sc Panangaden, P.}
  1991.
\newblock Semantic foundations of concurrent constraint programming.
\newblock In {\em POPL}, {D.~S. Wise}, Ed. ACM Press, 333--352.

\bibitem[\protect\citeauthoryear{Sato and Tamaki}{Sato and Tamaki}{1984}]{ST84}
{\sc Sato, T.} {\sc and} {\sc Tamaki, H.} 1984.
\newblock Enumeration of {S}uccess {P}atterns in {L}ogic {P}rograms.
\newblock {\em Theoretical Computer Science\/}~{\em 34}, 227--240.

\bibitem[\protect\citeauthoryear{Shapiro}{Shapiro}{1989}]{shapiro90}
{\sc Shapiro, E.~Y.} 1989.
\newblock The family of concurrent logic programming languages.
\newblock {\em ACM Comput. Surv.\/}~{\em 21,\/}~3, 413--510.

\bibitem[\protect\citeauthoryear{Smolka}{Smolka}{1994}]{DBLP:conf/ccl/Smolka94}
{\sc Smolka, G.} 1994.
\newblock A foundation for higher-order concurrent constraint programming.
\newblock In {\em CCL}, {J.-P. Jouannaud}, Ed. LNCS, vol. 845. Springer,
  50--72.

\bibitem[\protect\citeauthoryear{Song, Berezin, and Perrig}{Song
  et~al\mbox{.}}{2001}]{SongBP01}
{\sc Song, D.~X.}, {\sc Berezin, S.}, {\sc and} {\sc Perrig, A.} 2001.
\newblock Athena: A novel approach to efficient automatic security protocol
  analysis.
\newblock {\em Journal of Computer Security\/}~{\em 9,\/}~1/2, 47--74.

\bibitem[\protect\citeauthoryear{Tini}{Tini}{1999}]{Tini99}
{\sc Tini, S.} 1999.
\newblock On the expressiveness of timed concurrent constraint programming.
\newblock {\em Electr. Notes Theor. Comput. Sci.\/}~{\em 27}, 3--17.

\bibitem[\protect\citeauthoryear{Zaffanella, Giacobazzi, and Levi}{Zaffanella
  et~al\mbox{.}}{1997}]{ZaffanellaGL97}
{\sc Zaffanella, E.}, {\sc Giacobazzi, R.}, {\sc and} {\sc Levi, G.} 1997.
\newblock Abstracting synchronization in concurrent constraint programming.
\newblock {\em Journal of Functional and Logic Programming\/}~{\em 1997,\/}~6.

\end{thebibliography}
\end{document}